\documentclass[useAMS,usenatbib]{mn2e}

\newcommand{\conj}[1]{\overline{#1}}
\newcounter{NameOfTheNewCounter}
\newtheorem{theorem}{Theorem}[NameOfTheNewCounter]
\newtheorem{lemma}[theorem]{Lemma}
\newtheorem{proposition}[theorem]{Proposition}

\newtheorem{definition}[theorem]{Definition}

\usepackage{graphicx}
\def\lp{\frac 1 {l_{0}}}
\def\kp{\frac 1 {m_{0}}}
\newcommand{\Gcal}{\bmath{\mathcal{G}}}
\newcommand{\Rcal}{\bmath{\mathcal{R}}}
\newcommand{\Mcal}{\bmath{\mathcal{M}}}
\newcommand{\Fcal}{\bmath{\mathcal{F}}}
\newcommand{\GG}{\bmath{G}}

\title[Calibration artefacts in radio interferometry: Ghosts in WSRT data]{Calibration artefacts in radio interferometry.\\I. Ghost sources in WSRT data}
\author[T.~L. Grobler , C.~D. Nunhokee, O.~M. Smirnov, A.~J. van Zyl, A.~G. de Bruyn]{T.~L. Grobler$^{12}$\thanks{E-mail:
t.grobler@ru.ac.za}, 
C.~D. Nunhokee$^{1}$, O.~M. Smirnov$^{12}$, A.~J. van Zyl$^{3}$, A.~G. de Bruyn$^{45}$\\
$^1$Department of Physics and Electronics, Rhodes University, PO Box 94, Grahamstown, 6140, South Africa\\
$^2$SKA South Africa, 3rd Floor, The Park, Park Road, Pinelands, 7405, South Africa\\
$^3$Department of Mathematics and Applied Mathematics, University of Pretoria, Private bag X20, Hatfield, Pretoria 0028, South Africa\\
$^4$ASTRON, P.O. Box 2, Dwingeloo, 7900AA, The Netherlands\\
$^5$Kapteyn Astronomical Institute, University of Groningen, 9700 AV, Groningen, The Netherlands }
\begin{document}

\date{in original form 2014 February 2}

\pagerange{\pageref{firstpage}--\pageref{lastpage}} \pubyear{2014}

\maketitle

\label{firstpage}

\begin{abstract}
This work investigates a particular class of artefacts, or ghost sources, in radio interferometric images. 
Earlier observations with  (and simulations of) the Westerbork Synthesis Radio Telescope (WSRT) suggested that 
these were due to calibration with incomplete sky models. A theoretical framework is derived that validates this 
suggestion, and provides predictions of ghost formation in a two-source scenario. The predictions are found to 
accurately match the result of simulations, and qualitatively reproduce the ghosts previously seen in observational data. 
The theory also provides explanations for many previously puzzling features of these artefacts (regular geometry,
PSF-like sidelobes, seeming independence on model flux), and shows that the 
observed phenomenon of flux suppression affecting unmodelled sources is due to the same mechanism. 
We demonstrate that this ghost formation mechanism is a fundamental feature of calibration, and exhibits
a particularly strong and localized signature due to array redundancy. To some extent this mechanism will 
affect all observations (including those with non-redundant arrays), though in most cases the ghosts remain 
hidden below the noise or masked by other
instrumental artefacts. The implications of such errors on future deep observations are discussed. 
\end{abstract}


\section{Introduction}

In the context of radio interferometry, the term {\em calibration} refers to the estimation and correction of instrumental errors (which are traditionally taken to also include effects of the troposphere and ionosphere) on the observed visibilities. Current calibration approaches boil down to a joint fit to the observations of a sky model and an instrumental model, such as that provided by the radio interferometer measurement equation \citep[RIME;][]{ME1,RRIME1}. A typical observing strategy will include intermittent {\em calibrator scans} of a known calibrator field, for which an accurate prior sky model is available; the obtained instrumental solutions can then be interpolated onto scans of the target field. These can be further refined through a process known as {\em self-calibration} or {\em selfcal} \citep{Cornwell:selfcal}. Selfcal is an iterative approach (the sky model is refined at each iteration) that minimizes the error between predicted visibilities corrupted by the instrumental model (the free 
parameters) 
and the observed visibilities in a least squares sense during each iteration. An initial sky model for selfcal can be obtained by imaging visibilities that have been corrected by the interpolated calibrator solutions. Where a reasonable initial sky model for the target field is available, it can even provide the starting point for selfcal without the need for an external calibrator.

Traditional selfcal assumes an instrumental model where all effects are direction-independent. The increased field of view of modern radio interferometers implies that direction dependent effects can no longer be ignored during calibration. Incorporating direction dependent effects into calibration solutions (third-generation calibration, or 3GC) has become a major research field over the past few years \citep{Intema2009,Smirnov2011,Kazemi2011,Kazemi2013,Wijnholds2009}. \citet{Veen2004,Rau2009} have conducted good literature reviews on calibration. 

It is well-established that calibration can lead to imperfect images, even to the generation of spurious source components,
elimination or suppression  of real components, and the deformation of the structure of extended sources \citep{Linfield1986,Wilkinson1988,Taylor1999,Marti2010,Marti2008}. 
\cite{Kazemi2013a} have proposed a novel calibration 
technique that is meant to minimize the amount of source suppression that occurs if the sky model is incomplete. 
This approach uses a $t$-distribution to model the residual noise. \cite{Marti2008} has shown that 
spurious point sources can form when performing calibration on data containing nothing but white noise. 

\begin{figure*}
 \centering
 \includegraphics[width=\textwidth]{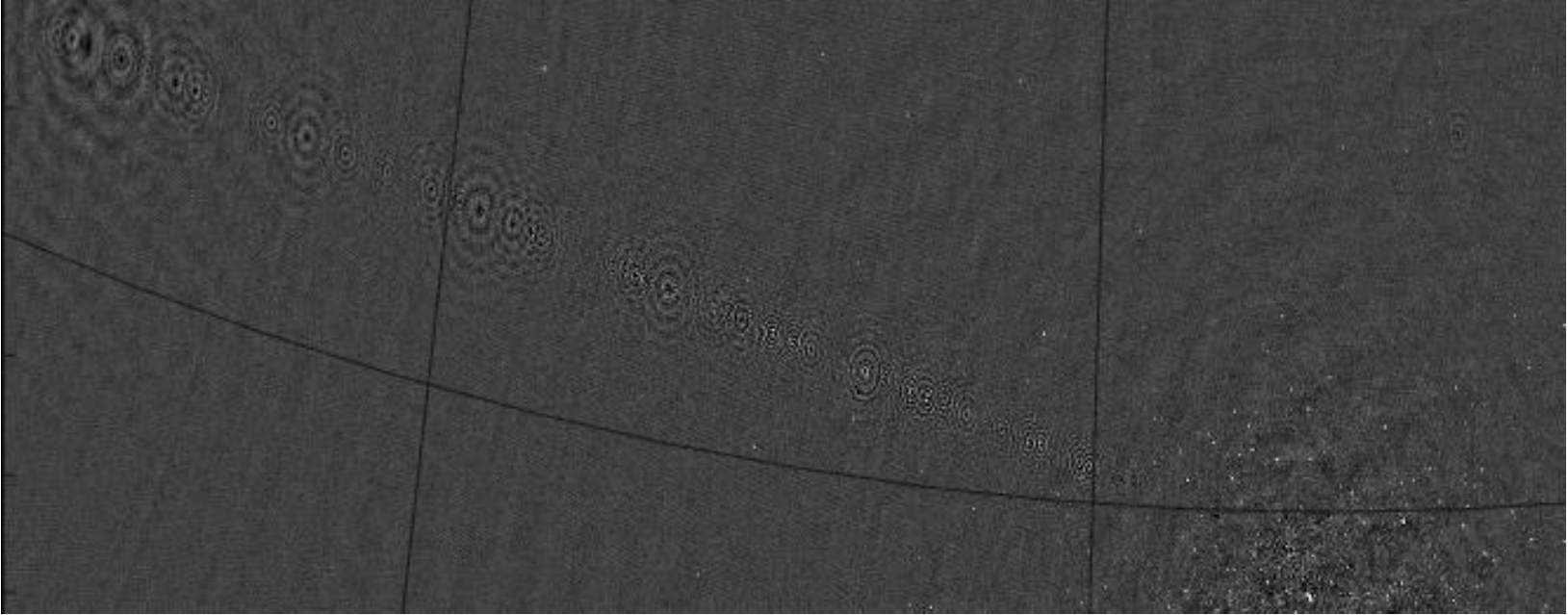}
 \caption{Ghost sources in a 92cm WSRT observation of J1819+3845. The target field is in the lower-image corner of the image, and Cyg A is to the upper left, just outside the image.}
 \label{fig:2004ghosts} 
\end{figure*}

One of the more striking sightings of calibration artefacts occurred in 2004, in a 92cm 
Westerbork Synthesis Radio Telescope (WSRT) observation of J1819+3845 by de Bruyn (Fig.~\ref{fig:2004ghosts}). 
After self-calibration, the map exhibited a string of ``ghosts'' -- point-source-like objects, mostly of negative 
flux, arranged along a line linking the brightest object in the field with Cyg A, which was about 20$^\circ$ away 
(i.e. in a distant sidelobe of the primary beam, and extremely attenuated by WSRT's extremely low sidelobe 
response). The pattern was highly peculiar as the ghosts were 
arranged with some regularity, and their positions did not vary with frequency. No other observations at the 
time were known to exhibit such features, and the problem remained open until a series of 21cm WSRT observations 
in 2010, which were done as part of the ``Quality Monitoring Committee'' project \citep{Smirnov2011qmc}. In
these observations, a large pointing error was deliberately introduced, and the resulting residual images 
(post self-cal) exhibited similar artefacts. The ghosts were fainter, but there were several strings of them, all associated with the brightest objects in the field (Fig.~\ref{fig:2010ghosts}). The problem was then investigated empirically, through the use of simulations \citep{Smirnov2010ghosts}, and this revealed a number of features:
\begin{itemize}
  \item The ghosts were associated with sky model errors (i.e. missing or incorrect flux in the sky model, or direction-dependent errors towards the brightest sources). In the QMC case, this was due to the large pointing error; in the J1819+3845 case this was due to insufficiently accurate modelling of Cyg A. Correcting for these errors (by solving for differential gains towards the brighter sources in the QMC case, and towards Cyg A in the J1819+3845 case) made the ghosts disappear.
  \item A simple simulation of a two-source (1 Jy and 1 mJy) field, where only the 1 Jy source was included in the calibration model, while the second source played the role of ``contaminator'', produced a similar ghost pattern in the residual visibilities. The peak intensity of the pattern was roughly at $\mu$Jy level, and appeared to be proportional to the flux of the contaminator source
  (but did not depend on the flux of the model source!) This suggested that ghosts should {\bf always} arise in the presence of incomplete sky models, but would generally be buried in the thermal noise, unless the observations were very sensitive, or the missing model sources
  were sufficiently bright.
  \item The ghosts always arranged themselves along a line (or lines) passing through the unmodelled or poorly modelled 
  source(s), and the dominant source(s) in the sky model. The positions of the ghosts corresponded to some (but not all) rational fractions of the interval between the sources (i.e. 1/2, 1/3, 2/3, 1/5, etc.), with significant variation in intensity. The positions did not depend on frequency.
  \item The ghosts exhibited sidelobes that were similar, but not identical to, the PSF of the telescope.
  \item Similar simulations with other telescopes (VLA) showed a far less regular artefact pattern.
\end{itemize}

\begin{figure*}
 \centering
 \includegraphics[width=\textwidth]{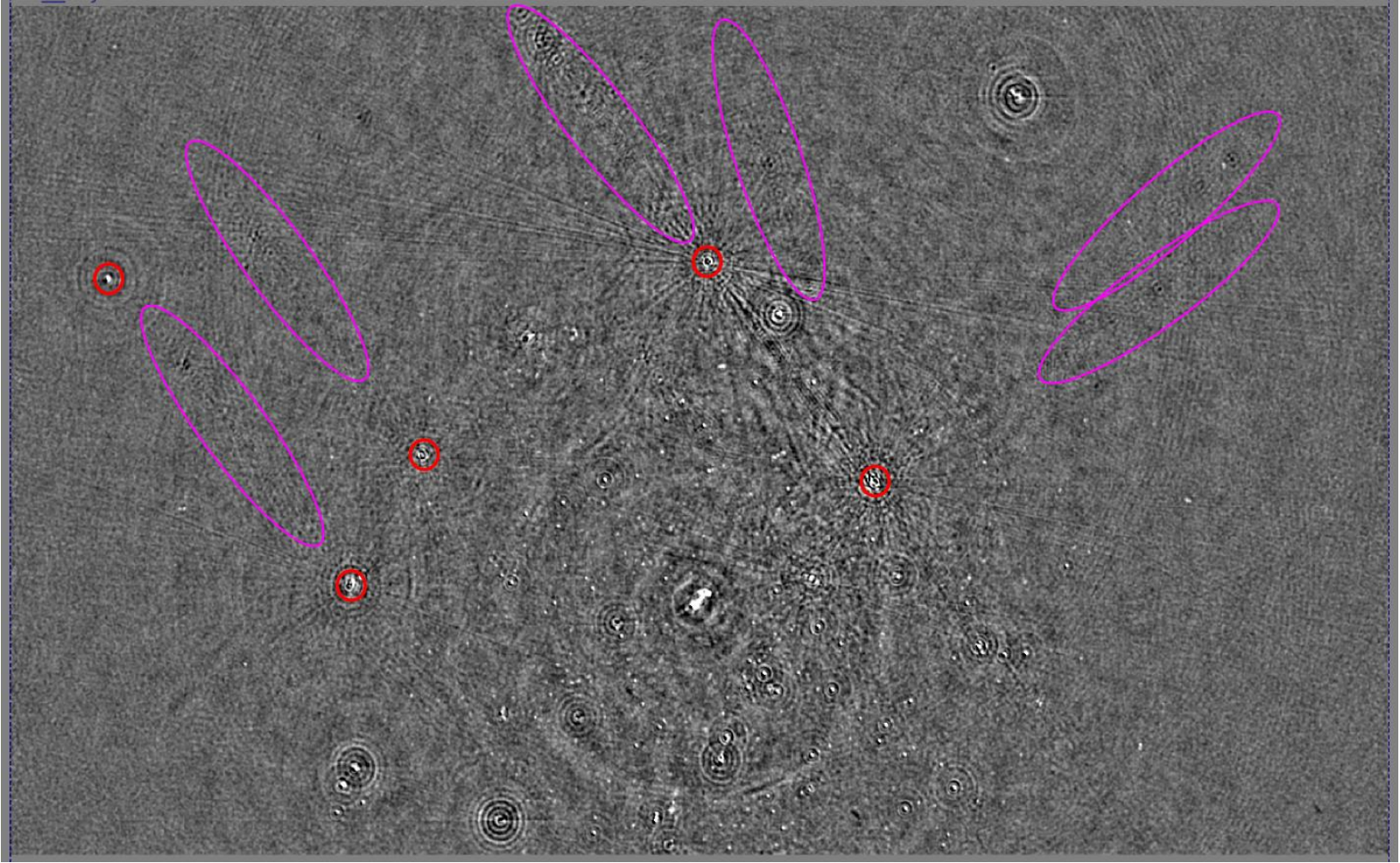}

 \caption{Ghost sources in a 21cm WSRT observation of the QMC2 field. Note that this is a residual dirty map, i.e. the sky model sources have been subtracted, and the visibly brighter sources here are in fact relatively faint. The positions of the brightest subtracted sources are indicated by red circles. Multiple strings of ghosts are visible, here highlighted by the ellipses. Note how the strings are firmly associated with the brightest sources.}
 \label{fig:2010ghosts} 
\end{figure*}

These facts strongly suggested that the regularity of the ghost patterns in Figs.~\ref{fig:2004ghosts} and \ref{fig:2010ghosts} was somehow related to the highly redundant geometry of the WSRT, but the mechanism by which they arose was not clear. Another conclusion of that study 
was that ghosts could be minimized and eventually driven below the noise by laboriously improving the sky model and/or applying direction-dependent solutions. 

The problem of calibration artefacts is becoming more important with the advent of new observational techniques and new radio telescopes such as the Low Frequency Array (LOFAR), the upgraded Jansky Very Large Array (JVLA), MeerKAT, etc., as well as the upcoming Square Kilometer Array (SKA). Not only do these telescopes promise (and in some cases already deliver) unprecedented sensitivity, they also increase the data rates substantially. The increased sensitivity means that fainter artefacts cannot be ignored, while the data rates require that calibration become largely automated, with careful and laborious manual data reduction no longer an option. Observations of the diffuse HI in the $0.5<z<20$  in order to probe galaxy formation \citep[e.g.][]{EOR2006}
or the nature of dark energy \citep{IM2010} need to face exquisite calibration in order to subtract foreground that are spatially and spectrally orders of magnitude brighter \citep[i.e.][]{Bernardi-EOR-foregrounds,Switzer-IM,Pober-PAPER-foregrounds}. Similar calibration requirements are needed for future stacking techniques for HI detection 
\citep[e.g.][]{Stacking2013} and continuum surveys aimed at revealing the nJy population \citep{Norris-SKA-continuum}. On the other hand, transient source 
detection pipelines require very fast on-the-fly calibration, which necessitates the use of shallow and incomplete sky models (Wijnholds 2013, priv. comm.) It is clear that a deeper and theoretical understanding of calibration artefacts, including those buried in the noise, is required if these new instruments and techniques are to achieve their scientific goals.

We have been investigating two classes of calibration artefacts, namely spurious point sources (ghosts) and source 
suppression (i.e. the reduction in observed flux of sources not included in the sky model). This work establishes that
both are manifestation of the same underlying mechanism, and aims to provide a theoretical 
understanding of this. The objective is to extend the results of 
previous papers \citep{Linfield1986,Wilkinson1988,Marti2008,Taylor1999}, by studying the underlying theoretical 
principles which are responsible for ghost formation and source suppression (as opposed to conducting an 
empirical survey of the different types of calibration artefacts that have been identified). 

This paper is the first in a series; ongoing work is concentrating on ghost formation and source suppression in non-redundant
interferometers, and on direction-dependent calibration. Ghosts have already been spotted in LOFAR data (de Bruyn priv. comm.),
and patricularly in the presense of transient sources (Fender priv. comm.), and early indications are that the same mechanism
is responsible. Future papers in the series will cover these phenomena.

%The ghost phenomenon is the most apparent after the first iteration, since the model is the most inaccurate 
%during the first iteration of self-calibration.

\section{Problem overview and definitions}

In this paper, we will concentrate on the WSRT example, since its highly redundant East-West geometry makes for prominent and regular ghosts. The results can be extended to other telescopes, which the follow-up work on source suppression will also deal with. We make a number of further simplifications:

\begin{itemize}
  \item We consider a case where the true sky consists of two discrete point sources with fluxes $A_1$ and $A_2$, the former at the phase 
  centre, and the calibration model consists of just the central source $A_1$.
  \item The sources are unpolarized, and we consider only a single frequency channel.
  \item Only direction-independent calibration (i.e. regular self-cal) is performed.
\end{itemize}

Multiple sources, multiple frequencies, polarization and studying the effects of direction-dependent solutions will be the subject of future work.

This section will anticipate some results of the following sections, in order to provide a logical outline that the rest of the paper will fill in.

\subsection[]{Calibration}
\label{sec:ALS}
%Under the assumption of no noise: but we can already assume noise

In its general form, (unpolarized) calibration entails finding a diagonal antenna gain matrix $\GG=\rmn{diag}(\bmath{g})=$ 
$\rmn{diag}([g_1,g_2,\cdots,g_n]^T)$ such that

\begin{equation}
\label{eq:cal}
||\Rcal - \GG\Mcal\GG^H||
\end{equation}

is minimized at each observational time-step. The superscript notation $()^H$ in Eq. \ref{eq:cal} denotes the Hermitian transpose. The Hermitian matrix $\Rcal$ is the observed unpolarized visibility matrix, where element $r_{pq}$ of $\Rcal$ is the visibility measured by the baseline
formed by antennas $p$ and $q$. The matrix $\Mcal$ is the corresponding visibility matrix generated from the calibration sky model.  
Eq.~\ref{eq:cal} can then be restated as

\begin{equation}
\label{eq:cal2}
||\Rcal-\bmath{g}\bmath{g}^H\odot\Mcal||=||\Rcal-\Gcal\odot\Mcal||,
\end{equation}

where ``$\odot$'' represents element-by-element multiplication (Hadamard product), and $\Gcal=\bmath{g}\bmath{g}^H$ is the matrix product of the gain solution vector with its own Hermitian transpose. The elements of $\Gcal$ will be denoted by $g_{pq}$. Crucially, $\Gcal$ is a rank one matrix by construction, and conversely, any rank one matrix can be decomposed into a product of the form $\bmath{g}\bmath{g}^H$.

Conventional approaches to radio interferometric calibration ignore the autocorrelations (i.e. the diagonal of the visibility matrix), since these are subject to a high additional self-noise term, and employ non-linear optimization techniques such as 
Levenberg-Marquardt \citep{Levenberg1944,Marquardt1963} to find a maximum likelihood (ML) solution for the 
off-diagonal terms of Eq.~\ref{eq:cal}. When a Gaussian noise model is assumed, this becomes 
equivalent to least squares (LS) minimization: 

\begin{equation}
\label{eq:LM}
\min_{\bmath{g}} \sum_{p \neq q} (r_{pq} - g_p m_{pq} \conj{g}_q)^2, 
\end{equation}
where $\conj{g}_q$ denotes the complex conjugate of $g_q$.

We will use the term LS calibration to refer to an LS solution of Eq.~\ref{eq:LM}. Most reduction packages in current use perform some sort of LS calibration. \citet{Kazemi2013a} propose an alternative approach called {\em robust calibration}, where a ML solution is obtained under the assumption of a $t$-distribution for the noise. Since implementations of robust calibration are not yet publicly available, we do not study it in this work. Where required, we make use of the MeqTrees package \citep{meqtrees} to do LS calibration.

%if the autocorrelation constraints are ignored. In \emph{MeqTrees}, Eq.~\ref{eq:LM} is solved by using Levenberg-Marquardt \citep{Noordam2010}. %The term LS (Least Squares) calibration will be used to refer to the calibration approach that ignores the autocorrelation constraints (i.e. %Eq.~\ref{eq:LM}). 

If autocorrelations are included in the optimization problem, an approach called Alternating Least Squares gain estimation \citep[ALS;][]{Boonstra2003,Wijnholds2009} can be used to obtain a solution to Eq.~\ref{eq:cal2}. It is not clear whether ALS provides a practical advantage over traditional LS without autocorrelations, since implementations of ALS compatible with conventional radio interferometric 
data do not exist. The issue deserves to be investigated in a separate study. For our purposes, ALS turns out to provide a vital theoretical framework in which ghost formation can be understood analytically. 

\subsection{Ghost formation}
\label{sec:ghostform}

In a nutshell, ghost sources are produced when Eq.~\ref{eq:cal} is solved for with an incomplete or incorrect model $\Mcal$. 
Consider the simple case where the observed visibilities $\Rcal$ correspond to two point sources, and the 
calibration model consists of a single point source at centre, $\Mcal=\bmath{1}$, where $\bmath{1}$ (boldface 1) represents a matrix of
all ones (not a unity matrix!) 
If we then assume a perfect instrument with unity gains, the actual solutions for $\GG$ will not be quite equal to unity, as they will attempt to fit for the difference between $\Mcal$ and $\Rcal$. Qualitatively, this process can be understood as follows: calibration attempts to move some ``real flux'' from the model $\Mcal$ to compensate for the unmodelled flux of the second point source. When these solutions are applied to the data, the resulting corrected visibilities

\begin{equation}
  \Rcal^\mathrm{(c)} = \GG^{-1}\Rcal\GG^{-H},
  \label{eq:corr}  
\end{equation}

will contain ghost sources in addition to real sources. In Eq.~\ref{eq:corr}, $()^{-1}$ denotes standard matrix inversion, while $()^{-H}$ designates $(()^H)^{-1}$. The rest of this paper analyses the mechanism by which this comes about. Note that we do not consider the effects of noise in our analysis; earlier empirical work \citep{Smirnov2010ghosts} has shown that the same pattern arises with or without noise.

In this context, LS calibration has proven to be very difficult to study theoretically. By contrast, the ALS formulation does yield the necessary insights. In this paper we therefore approach the problem of ghosts from several directions:

\begin{itemize}
  \item We develop a theoretical framework based on ALS that predicts ghost formation;
  \item We empirically compare the results of ALS and LS calibration, and show that they yield similar ghost patterns (with minor differences that are explained);
  \item We provide empirical results for ghost formation using ALS and LS, and show that these match the theoretical predictions;
  \item We show that all of the above match observed ghost patterns in real data, such as those seen in Fig.~\ref{fig:2004ghosts}.
\end{itemize}

These results suggest that the theoretical insights gained from the ALS framework are valid for the LS approaches, while the last 
point demonstrates that our simplified assumptions provide a good fit to real observations.

\subsection{Distillation}
\label{sec:dist}

Since ghost sources are relatively faint (as we'll show below), they can be difficult to detect over the thermal noise and the PSF sidelobes of actual sources. In hindsight, this probably explains why the phenomena was not spotted earlier. A straightfoward way to detect the ghosts in simulations is to ``distill'' them into residual visibilities as follows:

\begin{enumerate}
  \item We form predicted visibilities from a ``true'' sky ($\Mcal_0$), and an incomplete calibration sky model ($\Mcal$). 
  \item The ``observed visibilities'' $\Rcal$ then correspond to $\Mcal_0$.
  \item We obtain calibration solutions $\GG$ by solving Eq.~\ref{eq:cal} using $\Rcal$ and $\Mcal$.
  \item We apply the solutions to $\Rcal$ (Eq.~\ref{eq:corr}), yielding corrected visibilities $\Rcal^\mathrm{(c)}$.   
  \item We image the residuals $\Rcal^\Delta = \Rcal^\mathrm{(c)}-\Rcal$. The real sources then (mostly) cancel out, the noise term, if any, also (mostly) cancels out, and the resulting image yields the ``distilled'' ghost sources.
\end{enumerate}

Note that in real-life observations, actual gains are never unity, and the residuals $\Rcal^\mathrm{(c)}-\Rcal$ would not reveal much since real sources would not cancel out. However, in our perfect telescope simulation, the gain solutions account for sky model incompleteness and nothing more, and the ghosts are easily visible in images of the residuals.

In the two-source, noise-free case considered here, the true sky $\Mcal_0$ is equal to

\[
  \Rcal = \Mcal_0 = A_1 \bmath{1} + A_2 \bmath{K},
\]

where $\bmath{K}$ is a Fourier kernel matrix of complex phase terms corresponding to the offset of the second source w.r.t. the phase centre. The residuals then correspond to 

\[
  \Rcal^\Delta = A_1 \GG^{-1} \bmath{1} \GG^{-H} +  A_2 \GG^{-1} \bmath{K} \GG^{-H} - A_1 \bmath{1} - A_2 \bmath{K}
\]

\newcommand{\Gtop}{\Gcal^{\top}}

By defining the matrix $\Gtop = \GG^{-1}\bmath{1}\GG^{-H} = \{g_{pq}^{-1}\} = \big\{\frac{1}{g_{pq}}\big\}$ (i.e. the element-by-element inverse or the Hadamard inverse of
$\Gcal$), we can rewrite this as

\begin{equation}
  \Rcal^\Delta = A_1 (\Gtop - \bmath{1}) + A_2(\Gtop - \bmath{1}) \odot \bmath{K}.
\end{equation}

The matrix $\Gtop - \bmath{1}$ is in some sense fundamental. As will be shown below, it yields the basic ghost pattern corresponding to one source. From the equation above, we can see that the residuals will contain a superposition of two ghost patterns, scaled by $A_1$ and $A_2$, with the second pattern shifted to the position of the second source. In the general case, the residuals will correspond to a convolution of the true sky with $\Gtop - \bmath{1}$. Since in practice $A_2 \ll A_1$ (i.e. the missing flux in the model is usually considerably less than the flux accounted for), the first realization of the pattern is dominant (moreover, as will also be shown below, in the WSRT case the positions of the ghosts in the two patterns fall on top of one another). We shall refer to $\Gtop - \bmath{1}$ as the {\em distilled ghost pattern}. 

\section[]{Theoretical derivation}

\label{sec:t_der}
In this section analytic expressions for the elements of $\Gtop$ are derived. In the image domain each element of $\Gcal$ represents a different ghost pattern.
The ghost patterns that are associated with $\Gcal$ form due to a loss of information. Since $\Gcal$ is of a lower rank than 
$\Rcal$ (assuming a two source sky and a single source in the model) some information is lost when $\Gcal$ is computed. The inadequate rank one model $\Gcal$ leads to a significant change in the Fourier characteristics
of the original matrix $\Rcal$. The change in Fourier characteristics manifest as ghost patterns when $\Gcal$ is imaged. 
When $\Gtop$ is calculated the ghost patterns remain the same (the fluxes of the sources do however change). When the antenna gain solutions are applied
to $\Rcal$ the ghost patterns of $\Gtop$ get convolved with the true sky, which implies that $\Rcal^\mathrm{(c)}$ will contain ghost sources. 

A brief introduction to the Appendix is given in Sect.~\ref{sec:app_intro}, since it is crucial to the theory in this section. The mathematical definition of a regularly-spaced array is given in Sect.~\ref{sec:conf}, while Sect.~\ref{sec:sky} gives a better description of the experimental test case that is considered. Analytic expressions for the elements of $\Gtop$ are then derived in Sect.~\ref{sec:ghost_p}. Section~\ref{sec:imaging} describes how this results in ghost patterns in the dirty images, while Sect.~\ref{sec:corrvis} analyzes the effect of $\Gtop$ on the corrected visibilities.

The derivations in this section are highly mathematical; the crucial result is Eq.~\ref{eq:IpqG}, which shows that
the calibrated visibilities on each baseline, in case of a two-source sky and one-source model, are sampled from a periodic
one-dimensional $uv$-distribution, which in turn corresponds to a string of delta functions in the image plane. The reader
wishing to skip the heavier mathematics is encouraged to take Eq.~\ref{eq:IpqG} at face value, and skip to 
Sect.~\ref{sec:imaging}, which explains how the ``strings'' corresponding to each baseline combine to form 
ghosts in the final image.

\subsection[]{Introduction to Appendix}
\label{sec:app_intro}

The brief introductory explanation from above will be expanded upon in the rest of Sect.~\ref{sec:t_der}. The Appendix will be one of the main tools we will use to accomplish this. A brief introduction
to the Appendix is therefore needed. The appendix contains lemmas and propositions. The propositions are the main results that are used to derive the theoretical
results in this section. The lemmas are the dependencies that are required by these propositions. The relation between the lemmas and propositions are discussed in greater detail in the Appendix itself. The Appendix proves certain properties of $\Rcal(\bmath{b})$, $\Gcal(\bmath{b})$ and
$\Gtop(\bmath{b})$, which are the extrapolated counterparts of $\Rcal$, $\Gcal$ and $\Gtop$ (see Definition~\ref{def:R} and Definition~\ref{def:G}). 
These properties turn out to be essential in deriving the distilled ghost pattern. The following propositions are proven in the Appendix:
\begin{enumerate}
 \item Proposition~\ref{prop:1}, the rank of $\Rcal(\bmath{b})$ is rank two. This proposition quantifies the amount of information that is being lost during the computation of $\Gcal(\bmath{b})$. 
 \item Proposition~\ref{prop:2}, the elements of the function-valued matrix $\Gcal(\bmath{b})$ are periodic, effectively one-dimensional, differentiable, Hermitian functions.
 \item Proposition~\ref{prop:3}, it follows from Proposition~\ref{prop:2} that the elements of $\Gcal(\bmath{b})$ can be written as an effectively one-dimensional Fourier-series (which ultimately leads to the formation of ghosts).
 \item Proposition~\ref{prop:4}, the elements of $\Gtop(\bmath{b})$ are also periodic, effectively-one dimensional, differentiable, Hermitian functions and therefore by Proposition~\ref{prop:3} can also be expressed as an effectively one-dimensional Fourier-series.
\end{enumerate}

\subsection[]{Regular and redundant array geometries}
\label{sec:conf}
Since the geometric regularity of the WSRT layout will turn out to have an important effect on ghost formation, let's provide a formal mathematical definition here. 

\begin{definition}[Regularly-spaced array] Let us pick a coordinate system with origin at the first antenna position $\bmath{u}_1=0$.
We shall call a set of antenna positions $\{\bmath{u}_p\}$ {\em regularly-spaced} if there exists a {\em common quotient baseline} (CQB) $\bmath{b}_0$ such that each antenna position is an integer multiple of $\bmath{b}_0$, i.e. that $\bmath{u}_p = \phi_p \bmath{b}_0$, with $\phi_p$ being a whole number. We will also require that $\bmath{b}_0$ is the largest such baseline (equivalently, the greatest common divisor of $\{\phi_p\}$ is 1).
\end{definition}

\begin{definition}[Array geometry matrix]
\label{def:phi}
The array geometry matrix $\mathbf{\Phi}$ is an $n \times n$ integer matrix with elements $\phi_{pq}=\phi_q-\phi_p$. 
\end{definition}

Ovbiously, a regularly-spaced array defined in this way is necessarily one-dimensional. Note that regularity is preserved under rotation (but only in an East-West array). Note also that $\bmath{b}_0$ does not necessarily correspond to a real baseline. Most commonly-used configurations of the WSRT are regularly-spaced: the 10 fixed antennas have a CQB of 144m, while the CQB of the array as a whole is determined by the positions of the movable antennas RTA to RTD, with typical CQB lenghts of 6 or 12m\footnote{Since the WSRT movable antennas can in principle be placed at any position along a continuum, a non-regularly-spaced configuration is technically possible, but never used in practice.}. A {\em redundant} array will have many identical entries in $\mathbf{\Phi}$. A regularly-spaced array is not necessarily redundant, but WSRT itself is highly redundant. 

The matrix $\mathbf{\Phi}$ has a few interesting mathematical properties, which will be fully derived in the Appendix. Note that the actual $uv$-coordinates of each baseline are given by $\bmath{b}_0\bmath{\Phi}$. The matrix $\bmath{\Phi}$ can be thought of as 
representing a whole number scaling relationship between all the $uv$-tracks of the interferometer, and the {\em reference track} given by the CQB $\bmath{b}_0(t)$, which is a function of time due to the Earth's rotation. 

\subsection[]{The two source problem}
\label{sec:sky}

Let us consider a sky composed of two unpolarized point sources of flux $A_1$ and $A_2$, and a calibration sky model consisting of just the primary source $A_1$. Since the solutions to the calibration equation (Eq.~\ref{eq:cal}) are invariant with respect to amplitude rescaling and positional shifts that are applied to both the sky and the model, we may, without loss of generality, restrict ourselves to the case where the primary source has unity flux and is located at the phase centre. The ``true sky'' as a function of position $\bmath{s}=(l,m)$ (where $l$ and $m$ are the direction cosines) is then equal to 
$I_{\mathcal{R}}(\bmath{s})=A_1\delta(\bmath{s})+A_2\delta(\bmath{s}-\bmath{s}_0),$ and the ``model sky'' to $I_{\mathcal{M}}(\bmath{s})=A_1\delta(\bmath{s})$, where $A_1 = 1$, $\bmath{s_0}=(l_0,m_0)\neq 0$ is the position of the secondary source, and $\delta$ is the Kronecker delta function. Let us further assume a perfect interferometer with unity gains, a monochromatic observation, and integration intervals sufficiently short to make smearing negligible. The 
``observed visibility'' corresponding to the true sky $I_{\Rcal}$ can be written as

\begin{equation}
\label{eq:rtrue}
r(\bmath{u}) = A_1 + A_2e^{-2\pi i \bmath{u}\cdot\bmath{s}_0},\\
\end{equation}

where $\bmath{u}\cdot\bmath{s}_0$ is a dot product. If the array is regularly-spaced as defined above, then the visibility observed 
by baseline $pq$ at $uv$-coordinates $\bmath{u}_{pq}=\phi_{pq}\bmath{b}_0$ is

\begin{equation}
\label{eq:phi}
V_{pq} = r(\bmath{u}_{pq}) = r(\phi_{pq}\bmath{b}_0) = A_1 + A_2e^{-2\pi i \phi_{pq}\bmath{b}_0\cdot\bmath{s}_0},\\
\end{equation}
where $\bmath{b}_0=\bmath{b}_0(t)$ is the CQB. The ``model visibilities'' corresponding to the model sky $I_{\Mcal}$ above, are trivially all unity.

\subsection[]{The extrapolated visibility matrix}
The observed visibilities for each observational time step can be packed into a two dimensional matrix 
\begin{equation}
\Rcal = \left[ \begin{array}{cccc}
V_{11} & V_{12} & \cdots & V_{1n}\\
V_{21} & V_{22} & \cdots & V_{2n}\\
\vdots&\vdots&\vdots&\vdots\\
V_{n1} & V_{n2} & \cdots & V_{nn}\end{array} \right].
\end{equation} 
The elements of $\Rcal$ are functions that depend on time. For a regularly-spaced array, Eq.~\ref{eq:phi} can be utilized to rewrite all the elements of $\Rcal$ as functions of $\bmath{b}_0$. We can express this formally via the following definition:

\begin{definition}[Extrapolated visibility matrix]
\label{def:R}
Let $\Rcal(\bmath{b}):\mathbfss{R}^2\rightarrow\mathbfss{C}^{n\times n}$ be an $n \times n$ Hermitian function-valued matrix with entries
\begin{equation}
\label{eq:rpq}
r_{pq}(\bmath{b}) = r(\phi_{pq}\bmath{b}),
\end{equation}
where $r$ is given by Eq.~\ref{eq:rtrue}, $\phi_{pq}$ is given by the array geometry matrix $\mathbf{\Phi}$, $\bmath{b}=(u,v)$ and 
$\bmath{s}_0=(l_0,m_0)\neq 0$ are real two-vectors, $A_1=1$, and $0<A_2<1$.\\ 
\end{definition}

This allows us to formally define $\Rcal(\bmath{b})$ over the entire $uv$-plane, i.e. for any value $\bmath{b}$. The actual observed visibilities $\Rcal$ at time $t$ are given by $\Rcal(\bmath{b}_0(t))$. For any given baseline $pq$, $r_{pq}(\bmath{b}_0)$ corresponds to the visibilities measured by that baseline. Since $\bmath{b}_0(t)$ follows an elliptical track, our actual ``measurements'' on baseline $pq$ (i.e. the values subject to calibration) are restricted to that series of $uv$-points. However, by replacing $\bmath{b}_0$ by the free variable $\bmath{b}$ in Eq.~\ref{eq:rpq}, we automatically define an ``extrapolated'' visibility function over the entire $uv$-plane. By definition, the values of $r_{pq}$ over the track $\bmath{b}_0(t)$ are equal to visibilities measured by baseline $pq$ over the track $\phi_{pq}\bmath{b}_0(t)$. 

Note also that Eq.~\ref{eq:rpq} can also be seen as a coordinate scaling relationship between the observed visibility distribution $r(\bmath{u})$ and any given $r_{pq}(\bmath{b})$. To emphasize this, we use the variable $\bmath{u}$ to represent coordinates in the  ``observed'' $uv$-plane (where $r$ lives), and $\bmath{b}$ for coordinates in the ``scaled'' $uv$-planes (where the $r_{pq}$'s live). This also implies that the ``sky'' corresponding to any $r_{pq}$ (i.e. the inverse Fourier transform of $r_{pq}$) is a scaled and stretched version of the true sky.

Finally and most crucially (as we'll see in the discussion of ALS below), the $\Rcal(\bmath{b})$ matrix for any $\bmath{b}\neq 0$ can be shown to have rank two (see Proposition~\ref{prop:1} in the Appendix).

\subsection[]{The calibration matrix}
\label{sec:ghost_p}

Since our model visibilities are all unity, the calibration process (Eq.~\ref{eq:cal2}) entails finding some kind of ``best fit'' 
rank one matrix $\Gcal$, given $\Rcal$. In effect, the calibration process results in a mapping $\Rcal\to\Gcal$; by extension, this also defines a mapping $\Rcal(\bmath{b})\to\Gcal(\bmath{b})$ for any $\bmath{b}$. For LS calibration, the best fit is given by Eq.~\ref{eq:LM}. This has proven difficult to explore analytically, so we will consider ALS calibration instead (and later empirically show that it yields similar results).

In a nutshell, ALS calibration obtains a $\Gcal$ by ``de-ranking'' $\Rcal$, i.e. keeping just its largest eigenvalue. More precisely:

\begin{definition}[ALS calibration matrix]
\label{def:G}
Let $\Gcal(\bmath{b})=\lambda(\bmath{b})\mathbf{x}(\bmath{b})\mathbf{x}^{H}(\bmath{b})$, where $\lambda(\bmath{b})$ is the largest eigenvalue of $\Rcal(\bmath{b})$, and $\mathbf{x}(\bmath{b})$
is its associated normalized eigenvector.\\
\end{definition}

We will designate the elements of $\Gcal$ as $g_{pq}(\bmath{b})$.

To provide a specific example of the above, let us create a theoretical three-element interferometer with a geometry matrix of
\begin{equation}
\mathbf{\Phi}=\left[ \begin{array}{ccc}
0 & 3 & 5\\
-3 & 0 & 2\\
-5& -2 & 0 \end{array} \right],
\end{equation}

\begin{figure}
 \centering
 \includegraphics[width=0.48\textwidth]{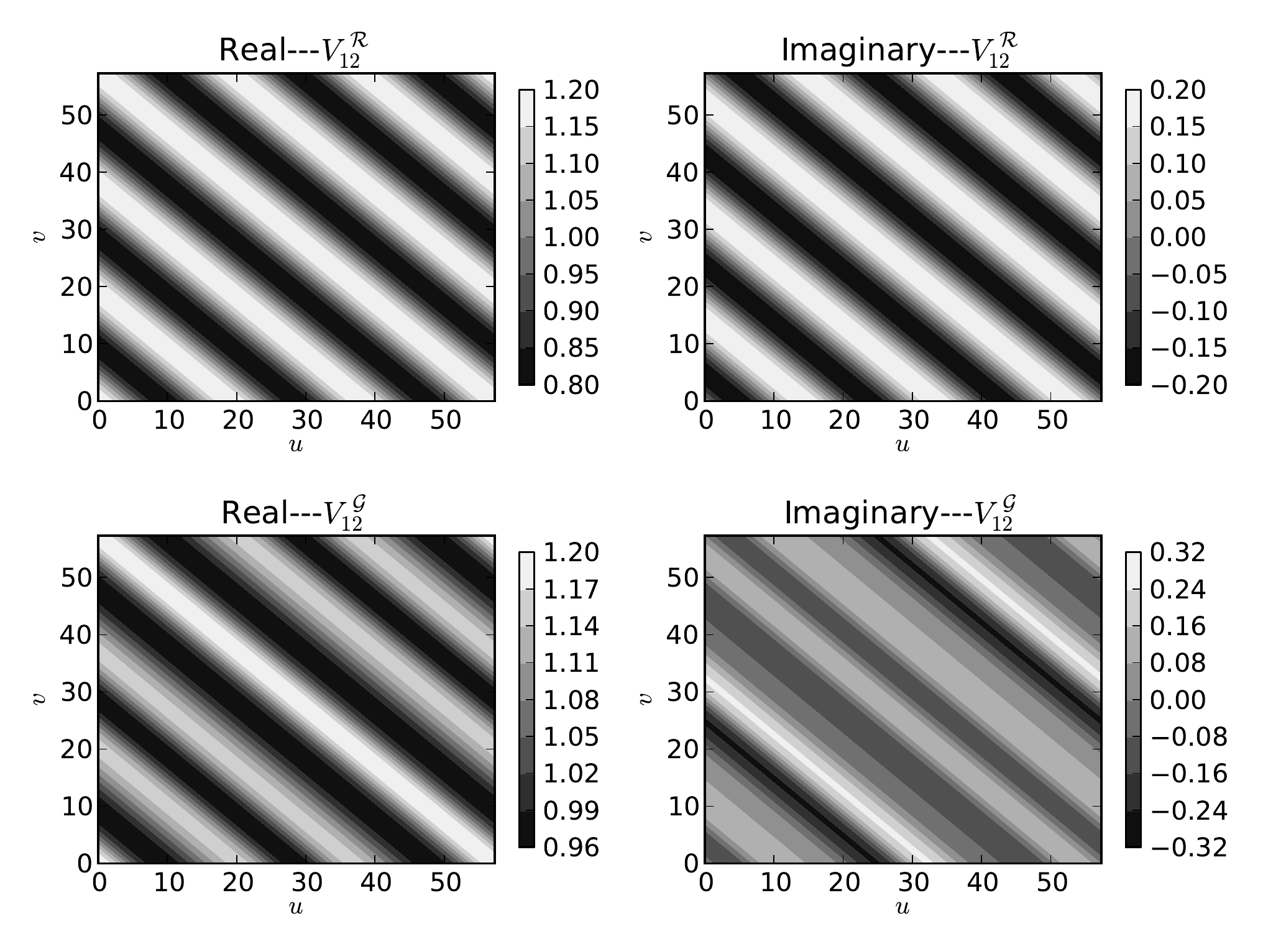}
 % V_R_3.pdf: 585x441 pixel, 72dpi, 20.64x15.56 cm, bb=0 0 585 441
 \caption{The functions $r_{12}(\bmath{b})$ and $g_{12}(\bmath{b})$.}
 \label{fig:fig1} 
\end{figure}
\begin{figure}
 \centering
 \includegraphics[width=0.48\textwidth]{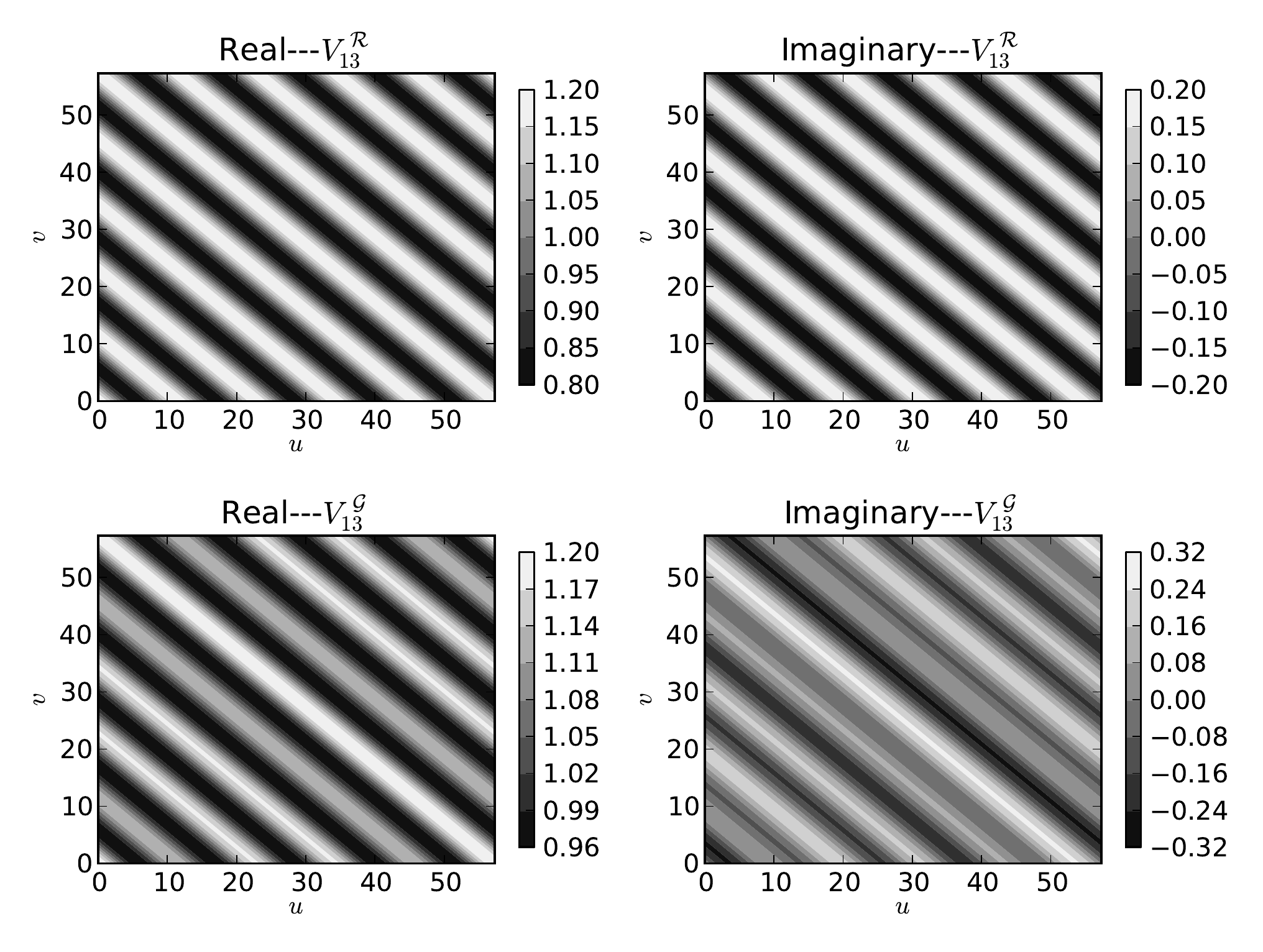}
 % V_R_3.pdf: 585x441 pixel, 72dpi, 20.64x15.56 cm, bb=0 0 585 441
 \caption{The functions $r_{13}(\bmath{b})$ and $g_{13}(\bmath{b})$.}
 \label{fig:fig2} 
\end{figure}
\begin{figure}
 \centering
 \includegraphics[width=0.48\textwidth]{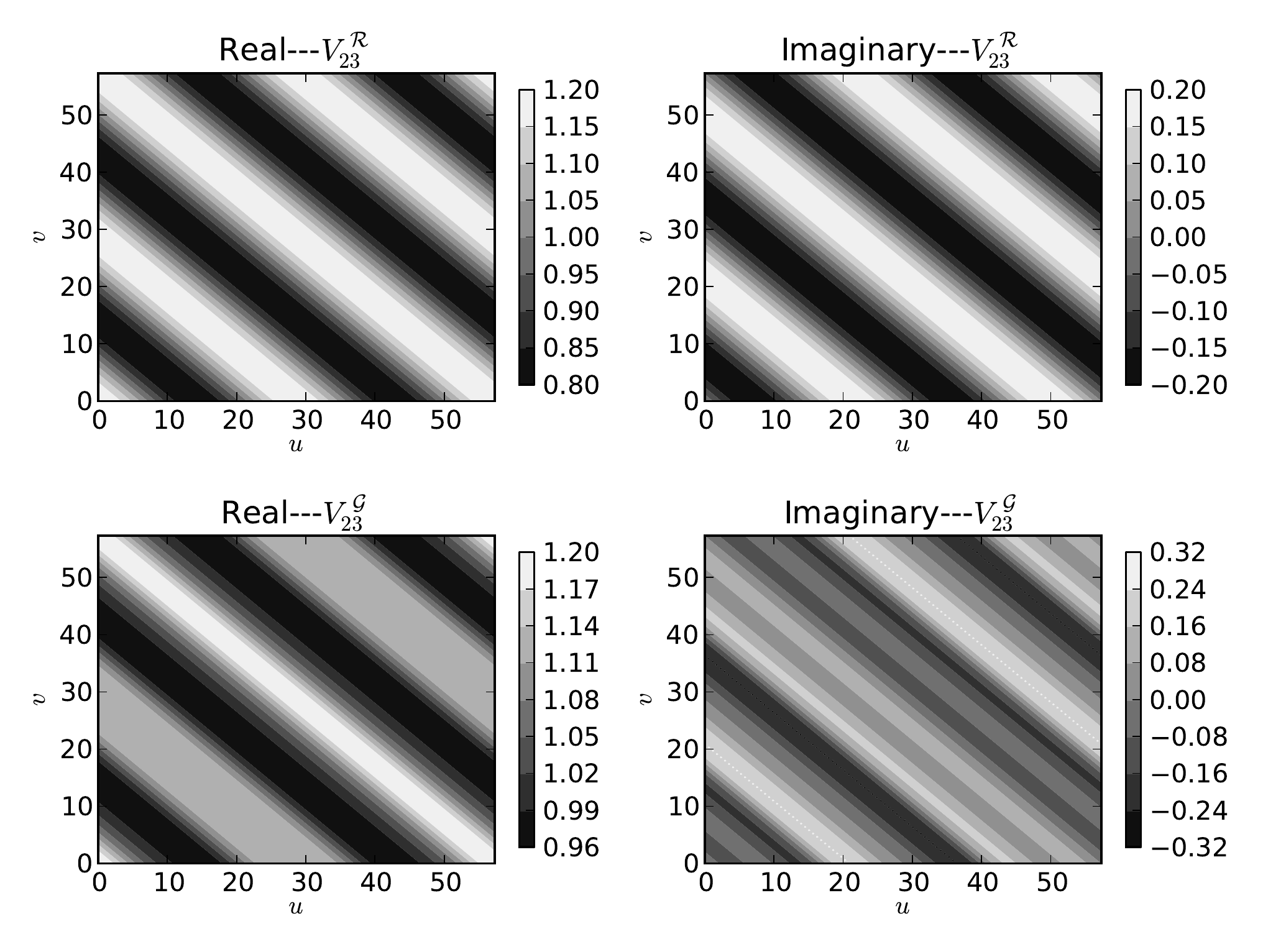}
 % V_R_3.pdf: 585x441 pixel, 72dpi, 20.64x15.56 cm, bb=0 0 585 441
 \caption{The functions $r_{23}(\bmath{b})$ and $g_{23}(\bmath{b})$.}
 \label{fig:fig3} 
\end{figure}

and place a secondary source of flux $A_2=0.2$Jy at $l_0=1^{\circ},m_0=1^{\circ}$. Fig.~\ref{fig:fig1}, Fig. ~\ref{fig:fig2} and Fig. ~\ref{fig:fig3} graphically display the resulting $\Rcal(\bmath{b})$ and $\Gcal(\bmath{b})$ matrices. The following observations can also be made. The functions $r_{pq}(\bmath{b})$ (top row of each figure) are trivial phase gradients, and are thus continuous, differentiable, Hermitian and periodic in 
the $u$ and $v$ direction with periods of  $\frac{1}{\phi_{pq}|l_0|}$ and $\frac{1}{\phi_{pq}|m_0|}$. They are effectively one-dimensional,
i.e. constant along each line $v=-\frac{l_0}{m_0}u+c$ for any $c$. The $\Gcal(\bmath{b})$ functions  (bottom row of each figure) have a more interesting structure, but are also differentiable, Hermitian, one-dimensional and periodic, with periods along $u$ and $v$ of $\frac{1}{|l_0|}$ and $\frac{1}{|m_0|}$. Moreover, this holds for any ALS calibration matrix as defined above (see Proposition~\ref{prop:2}). The difference between the $g_{pq}$'s is that 
each of them has a different secondary harmonic which is determined by $\phi_{12},\phi_{13}$ and $\phi_{23}$.

Any periodic, one-dimensional and differentiable function can be written out in terms of a one-dimensional discrete Fourier transform. We can therefore decompose each $g_{pq}$ as follows:

\begin{equation}
\label{eq:m_v}
g_{pq}(\bmath{b}) = \sum_{j=-\infty}^{\infty}c_{j,pq}^{\mathcal{G}}e^{-2\pi i j\bmath{b}\cdot\bmath{s}_0}
\end{equation}

Proposition~\ref{prop:3} derives this result formally. The coefficients $c_{j,pq}^{\mathcal{G}}$ (real, since $g_{pq}$ is Hermitian) 
have a very non-trivial structure, but they can be calculated using Eq.~\ref{eq:eq_2}. 

Now, since $g_{pq}(\bmath{b})$ represents the predicted corrupted visibilities (given that we have a unity model), it is fair to ask, what 
image-plane distribution corresponds to the visibility distribution $g_{pq}(\bmath{b})$? Doing an inverse 2D Fourier transform, we obtain:

\begin{equation}
\mathcal{F}^{-1}\{\ g_{pq} \}(\bmath{s}) = \sum_{j=-\infty}^{\infty}c_{j,pq}^{\mathcal{G}}\delta(\bmath{s}-j\bmath{s}_0),
\end{equation}

i.e. a sum of delta-functions whose locations are integer multiples of $\bmath{s}_0=(l_0,m_0)$. 

Let us now define something we'll call the ``$\Gcal$-sky of baseline $pq$'' as follows:

\begin{eqnarray}
I_{pq}^{\mathcal{G}}(\bmath{s}) &=& \mathcal{F}^{-1}\bigg\{\ g_{pq}\bigg(\frac{\bmath{b}}{\phi_{pq}}\bigg) \bigg\} \nonumber\\
\label{eq:IpqG}
&=&\sum_{j=-\infty}^{\infty}c_{j,pq}^{\mathcal{G}}\,\delta\bigg(\bmath{s}-\frac{j \bmath{s}_0}{\phi_{pq}}\bigg). 
\end{eqnarray}

The physical meaning of $I_{pq}^{\mathcal{G}}$ is as follows: it is a sky distribution whose Fourier transform yields a visibility
distribution that, along the $uv$-track given by $\phi_{pq}\bmath{b}_0(t)$, is consistent with the \emph{predicted corrupted visibilities} $g_{pq}$ along the track given by $\bmath{b}_0(t)$ (note how the scaling relationship of Eq.~\ref{eq:rpq} enters into Eq.~\ref{eq:IpqG}).  In other words, after the best-fitting calibration gains have been applied, the resulting predicted visibilities for each baseline $pq$ {\em will be consistent with a sky of delta functions spaced at intervals of} $s_0/\phi_{pq}$, with intensities given by $\{c_{j,pq}^{\mathcal{G}}\}$.

These delta functions are the fundamental ingredients of the ghosts observed in Fig.~\ref{fig:2004ghosts}. We will shortly show that the corrected visibilities exhibit a similar structure, but first let us consider what happens to the visibilities given by $g_{pq}$ during imaging.

\subsubsection{More on extrapolation} 

It has been our experience that the mathematical construct of extrapolated visibility functions, which is key to the
above arguments, is particularly difficult to explain or justify clearly. In this section we attempt to reformulate
the argument again in general terms.

In order to understand ghost formation, we need to understand the behaviour of the best-fitting visibilities
produced by the calibration process. The actual visibilities per each baseline are sampled along an elliptical track 
in the $uv$-plane. Analysing the mathematical properties of functions
defined along a specific $uv$-track proved to be a difficult problem, analysing continuous functions defined 
over the entire $uv$-plane proved more fruitful. We therefore proceed by finding a unique mapping from the 
former (specific) problem to the latter (general) problem, and back. 

More specifically, the extrapolation operation defined above (Eq.~\ref{eq:rpq}) provides a formal recipe for mapping 
sets of per-baseline visibilities onto functions defined over the entire $uv$-plane (giving us a per-baseline 
``virtual $uv$-plane'' that is consistent with the visibilities over the one specific track of that baseline). We then define the 
ALS calibration process in terms of operations on such virtual $uv$-planes. The virtual $uv$-planes corresponding 
to the best-fit visibilities ($g_{pq}(\bmath{b})$, bottom row of Figs.~\ref{fig:fig1}-\ref{fig:fig3}), of which
the actual visibilities are a subset (given by the baseline's $uv$-track) turn out to have certain mathematical properties: they are Hermitian, one-dimensional and periodic (with the same period 
across all baselines), and therefore correspond to a string of delta-functions in the image domain. Note how these 
properties are straightforward to establish for functions defined over the entire $uv$-plane, but are a lot less
obvious if one only considers a subset of the $uv$-plane along a track. (This observation is the
main justification for the extrapolation formalism.) This establishes that the best-fit visibilities 
per each baseline are consistent with a string of delta functions. Finally, the geometric scaling relationship 
implicit in Eq.~\ref{eq:rpq} causes the spacing of the delta-functions to be inversely proportional to baseline length.

\subsection{Imaging}
\label{sec:imaging}

In the situation above, each baseline's predicted corrupted visibilities correspond to its own apparent sky $I_{pq}$. During conventional interferometric imaging, the per-baseline visibilities are interpolated onto a ``common'' $uv$-plane using convolutional gridding, and the result is Fourier transformed back into an estimate of the sky (the so-called ``dirty image''). Mathematically, this can be described as follows: 

\begin{equation}
I_D = \mathcal{F}^{-1}\bigg\{\ \sum_{pq} S_{pq}\mathcal{F}\{ I_{pq} \}\ \bigg\}. 
\end{equation}

Here, $S_{pq}$ is the \emph{sampling function} of baseline $pq$. The sampling function is only non-zero in the neighbourhood of the track described by $\bmath{u}_{pq}$, and accounts for both the imaging weights and the interpolation coefficients of the gridding process. This can be rewritten as

\begin{equation}
\label{eq:Idpq}
I_D = \sum_{pq} P_{pq} \circ I_{pq}, 
\end{equation}
where ``$\circ$'' denotes convolution, and $P_{pq}=\mathcal{F}^{-1}\{S_{pq}\}$ is the (unnormalized) PSF associated with baseline $pq$. Note that in the case of each baseline seeing a common sky $I$, the above becomes
\begin{equation}
I_D = (\sum_{pq} P_{pq} ) \circ I, 
\end{equation}
which is the familiar result that the dirty image $I_D$ is the convolution of the true sky $I$ by the point spread function of the array $P$ given by
\begin{equation}
\label{eq:psf}
P = \sum_{pq} P_{pq}.
\end{equation}

Now, recall that Eq.~\ref{eq:IpqG} describes a string of delta functions spaced at intervals of $\bmath{s}_0/\phi_{pq}$. 
If we define $\phi_0$ as the least common multiple of all $\phi_{pq}$, we can rewrite the equation as a sequence of delta functions spaced
at intervals of $\bmath{s}_0/\phi_0$, some of them possibly of zero amplitude: 

\begin{eqnarray}
\label{eq:IpqG1}
I_{pq}^{\mathcal{G}}(\bmath{s}) 
= \sum_{k=-\infty}^{\infty}d_{k,pq}^{\mathcal{G}}\,\delta\bigg(\bmath{s}-\frac{k \bmath{s}_0}{\phi_0}\bigg),
\end{eqnarray}

where $d_{k,pq}^{\mathcal{G}}=c_{j,pq}^{\mathcal{G}}$ if there is an integer $j$ such that $k\phi_{pq}=j\phi_0$, and zero otherwise. To simplify further equations, we'll use the $\delta_k$ as shorthand for the $k$-th delta function above:

\[
\delta_k(\bmath{s}) = \delta\bigg(\bmath{s}-\frac{k \bmath{s}_0}{\phi_0}\bigg).
\]

Substituting this into Eq.~\ref{eq:Idpq}, we get
\begin{eqnarray}
I_D^\mathcal{G} &=& \sum_{pq} P_{pq} \circ \bigg ( \sum_{k=-\infty}^{\infty}d_{k,pq}^{\mathcal{G}}\,\delta_k\bigg) \\
\label{eq:ID}
 &=& \sum_{k=-\infty}^{\infty} \bigg ( \sum_{pq} d_{k,pq}^{\mathcal{G}}P_{pq} \bigg ) \circ \delta_k.
\end{eqnarray}

Physically, this can be interpreted as follows. The dirty image $I_D^\mathcal{G}$ which we get as a result of imaging the predicted corrupted visibilities consists of a string of delta functions at regularly-spaced locations $k \bmath{s}_0/\phi_0$, each one convolved with its own {\em ghost spread function} (GSF) $P^\mathcal{G}_k$:

\begin{equation}
\label{eq:gsf}
P^\mathcal{G}_k = \sum_{pq} d_{k,pq}^{\mathcal{G}}P_{pq}.
\end{equation}

Comparing this to Eq.~\ref{eq:psf}, we can now understand the previously puzzling observation that the ghost sources in 
Fig.~\ref{fig:2004ghosts} appear to be convolved with differing point spread functions, similar but not identical 
to the nominal PSF of the WSRT. 

Furthermore, ghost positions do not depend on frequency (only on array and source 
geometry) -- though the GSF of course does. This is also consistent with previous observations.

\subsection{Corrected visibilities}
\label{sec:corrvis}

In real life, one would typically be imaging the \emph{corrected visibilities} (Sect.~\ref{sec:ghostform}) given by

\begin{equation}
\label{eq:cor}
 \Rcal^\mathrm{(c)} = \GG^{-1}\Rcal\GG^{-H} = \Gcal^{\top}\odot\Rcal,
\end{equation}

and our real goal is to understand the effect of $\Gcal$ on the \emph{corrected sky} $I^\mathrm{(c)} = \Fcal^{-1}\{\Rcal^\mathrm{(c)}\}$. To get there, we need to take an intermediate step. First, let us define  a ``$\Gcal^\top$-sky'' whose Fourier transform is consistent with the visibility distribution given by $g_{pq}^{-1}$. Proposition~\ref{prop:4} shows\footnote{Note that this proposition implicitly assumes
$g_{pq}\neq 0$, i.e. that the ALS calibration solutions are not null. Intuition suggests that this is a safe assumption: a null gain solution would yield null predicted visibilities, which could hardly be a ``best fit'' to the calibration equation in any sense. However, obtaining a rigorous proof of this has been surprisingly difficult, so we will let the assumption stand as is.} that the visibility distribution $g^{-1}_{pq}(\bmath{b})$ can also be decomposed into a Fourier series:

\begin{equation}
g^{-1}_{pq}(\bmath{b}) = \sum_{j=-\infty}^{\infty}c^\top_{j,pq} e^{2\pi i j\bmath{b}\cdot\bmath{s}_0},
\end{equation}

which implies that the corresponding ``$\Gcal^\top$-sky'' has a similar form to Eq.~\ref{eq:IpqG1}, but with a different set of coefficients:

\begin{equation}
\label{eq:G_inv}
I_{pq}^{\mathcal{G}^{\top}} = \sum_{k=-\infty}^{\infty}d_{k,pq}^{\top}\,\delta_k.
\end{equation}

Now, consider the matrix $\Gcal^\top(\bmath{b}) \odot \Rcal(\bmath{b})$. We'll designate its elements as $r_{pq}^\top(\bmath{b})$. The inverse Fourier transform of each element is then 

\begin{equation}
\mathcal{F}^{-1}\{ r_{pq}^\top \} = \mathcal{F}^{-1}\{ g^{-1}_{pq} \} \circ \mathcal{F}^{-1}\{ r_{pq} \},
\end{equation}

and the inverse Fourier transforms of both components have already been derived above. This means that the ``corrected sky'' 
corresponding to the corrected visibilities of baseline $pq$ is given by

\begin{equation}
I_{pq}^\mathrm{(c)} = I_{pq}^{\mathcal{G}^{\top}} \circ I^{\mathcal{R}},
\end{equation}
i.e. is simply a convolution of the real sky $I^{\mathcal{R}}$ with the ``ghost pattern'' of delta functions given by $I_{pq}^{\mathcal{G}^{\top}}$ above. In other words, the ``corrected sky'' will contain multiple instances of the fundamental ghost 
pattern (what we call the distilled ghost pattern), centered on each source, and scaled by the flux of that source. It 
can be seen that the sky corresponding to the residuals $\Rcal^\Delta$ is given by

\begin{equation}
\label{eq:IpqDelta}
I_{pq}^\Delta = I_{pq}^{\mathcal{G}^{\top}\!-1} \circ I^{\mathcal{R}} = (I_{pq}^{\mathcal{G}^{\top}}-\delta) \circ I^{\mathcal{R}}.
\end{equation}

The $I_{pq}^{\mathcal{G}^{\top}\!-1}$ term is the per-baseline sky associated with the distilled ghost pattern $\Gcal^\top-\bmath{1}$. 
By analogy with Eq.~\ref{eq:ID}, we can derive an expression for the full dirty image:

\begin{equation}
\label{eq:IdGT1}
I_D^{\mathcal{G}^{\top}\!-1} = \sum_{k=-\infty}^{\infty} \bigg ( 
\sum_{pq} \hat{d}_{k,pq}^{\top} P_{pq} 
\bigg ) \circ \delta_k,
\end{equation}

where $\hat{d}_{k,pq}^\top = d_{k,pq}^\top-1_{\{k=0\}}$, and the notation $1_{\{k=0\}}$ represents a series whose coefficients 
are 1 at $k=0$ and 0 elsewhere. Note the physical meaning of this operation: the $d_{0,pq}^\top$ coefficient corresponds to the
position of the real source $A_1$ in the image, and is close to unity, since the corrected visibilities contain the real source 
as well as all the ghosts. Subtracting 1 from this corresponds to taking the residual visibilities.

In our simple case the real sky consists of two sources, and the resulting corrected sky is a superposition of 
two patterns given by $I_{pq}^{\mathcal{G}^{\top}}$, scaled by $A_1$ and $A_2$, and centered on origin and on $\bmath{s}_0$, 
respectively. Since each pattern yields ghosts at discrete intervals of $\bmath{s}_0/\phi_{0}$, the two sets of positions
align, and we can work out the amplitudes of the resulting superposed ghost sources by summing up the corresponding coefficients. 
By analogy with Eq.~\ref{eq:ID}, we can derive the following equation for the dirty image formed from corrected visibilities:

\begin{equation}
\label{eq:Idcorr}
I_D^\mathrm{(c)} = \sum_{k=-\infty}^{\infty} \bigg ( 
\sum_{pq} (A_1d_{k,pq}^\top + A_2d_{k-\phi_0,pq}^\top)
P_{pq} 
\bigg ) \circ \delta_k.
\end{equation}

From Eq. \ref{eq:IpqDelta}, it follows that the dirty image corresponding to the residuals can be obtained by subtracting unity from the corresponding $d$ coefficient:

\begin{equation}
\label{eq:Idres}
I_D^\Delta = \sum_{k=-\infty}^{\infty} \bigg ( 
\sum_{pq} (A_1 \hat{d}_{k,pq}^\top + A_2 \hat{d}_{k-\phi_0,pq}^\top)
P_{pq} 
\bigg ) \circ \delta_k,
\end{equation}

where again $\hat{d}_{k,pq}^\top=d_{k,pq}^\top-1_{\{k=0\}}$.

Equations~\ref{eq:ID}, \ref{eq:IdGT1}, \ref{eq:Idcorr} and \ref{eq:Idres} summarize the formation of ghost patterns in the predicted corrupted, corrected and residual visibilities. 

For our purposes, it is important to derive a theoretical prediction for the amplitudes of individual ghost sources as a 
fraction of the ``missing flux'' $A_2$. Consider the $k$-th ghost source located at $k\bmath{s}_0/\phi_0$. 
From Eq.~\ref{eq:IdGT1}, it follows that the amplitude of the ghost source is given by the weighted sum

\[
\sum_{pq} \hat{d}_{k,pq}^{\top} P_{pq}(0), 
\]

where the per-baseline weights $P_{pq}(0)$ are ultimately determined by the imaging weights. For simplicity, let us consider the case of natural weighting, in which case the sum becomes unweighted. We can then define the quantity

\newcommand{\pqavg}[1]{\left\langle#1\right\rangle_{pq}}

\begin{equation}
\zeta_k = \pqavg{\hat{d}_{k,pq}^{\top}}
\end{equation}

where $\pqavg{\cdot}$ denotes averaging over all baselines $pq$. This gives us the amplitude of the $k$-th ghost source in the distilled pattern (assuming natural weighting), and can be computed analytically from the results above. Likewise, the $k$-th source in the corrected residuals is a superposition of two appropriately scaled sources from the distilled pattern, and its amplitude is given by (assuming natural weighting):

\begin{equation}
\label{eq:zeta}
\zeta_k^\Delta = A_1 \zeta_k + A_2 \zeta_{k-\phi_0}
\end{equation}

Of particular interest is the quantity $\zeta_{\phi_0}^\Delta/A_2$, as this gives the relative amplitude of the 
``flux suppression ghost'' sitting on top of source $A_2$. Indirectly, this one ghost has been observable since the 
invention of selfcal, since it corresponds to the previously noted phenomenon of flux suppression in 
unmodelled sources. The theoretical derivation given here provides an explanation for this.

\section[]{Results}
\label{sec:results}
%\subsection[]{Ghost pattern results}

% The primary focus of this section is to discuss some of the results obtained by implementing the theory in Section~\ref{sec:ghost_p}. 
% As explained in Section~\ref{sec:ghost_p}, Eq.~\ref{eq:cor} should be interpreted as a collection of convolutions. Each baseline $pq$
% has a unique ghost pattern $I_{pq}^{\mathcal{G}^{\top}}(l,m)$ (see Eq.~\ref{eq:G_inv}) associated with it that is convolved with the true sky $I_{\mathcal{R}}(l,m)$. The ghost patterns are dominated by a bright source at $(0^{\circ},0^{\circ})$. Distilation
% was used to make the ghost patterns more visible. The distilled ghost pattern is equal to (see Section~\ref{sec:dist})
% \begin{equation}
% \Gcal^{\top}-\bmath{1}.
% \end{equation}

Section~\ref{sec:ghost_p} provides a theoretical framework for understanding ghost formation, as well as a mechanism for predicting the distribution and amplitudes of ghosts in the two-source case. In this section, we apply the mechanism to predict ghost formation for a specific observational scenario, and compare the results with simulations.

As discussed above, the ghost pattern is highly dependent on the array configuration. The results in this section were all generated with a traditional (36,108,1332,1404m) WSRT configuration. Unless specified otherwise, $A_2 = 0.2$ Jy, $A_1 = 1$, Jy, $l_0 = 1^{\circ}$ and $m_0=0^{\circ}$, and we assumed monochromatic observations at a frequency of 1.45 GHz. To verify the theory developed above, we compare the distribution of ghost sources obtained  by three methods:

\begin{itemize}
\item A theoretically predicted distribution, using the framework of Section~\ref{sec:ghost_p}.
\item ALS calibration of simulated data (using a custom-made implementation).
\item LS calibration of simulated data using the MeqTrees \citep{meqtrees} package.
\end{itemize}

\begin{figure*}
\includegraphics[width=.949\textwidth]{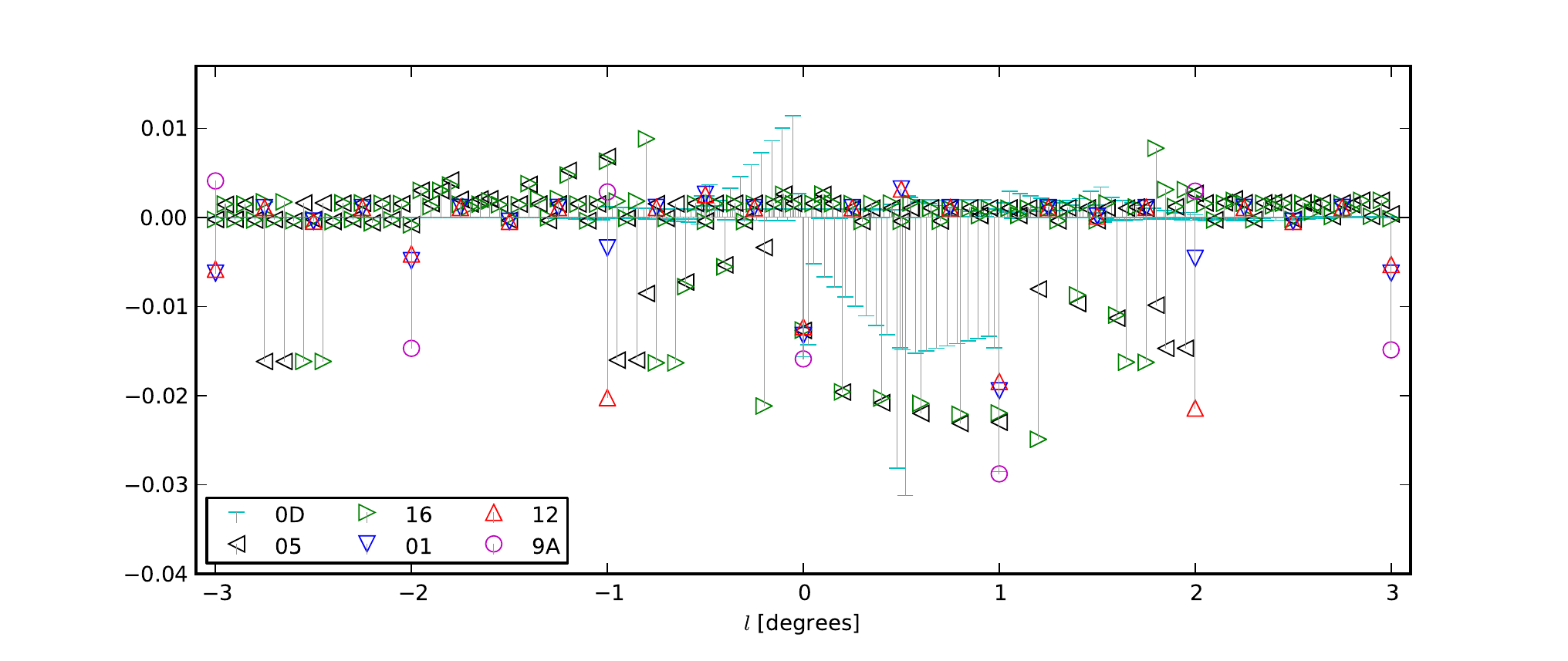}
\caption{Theoretical ghost pattern for baselines  9A (36m), 01 and 12 (144m), 05 and 16 (720m), 0D (2.7km).}
\label{fig:theor_stem} 
\includegraphics[width=.949\textwidth]{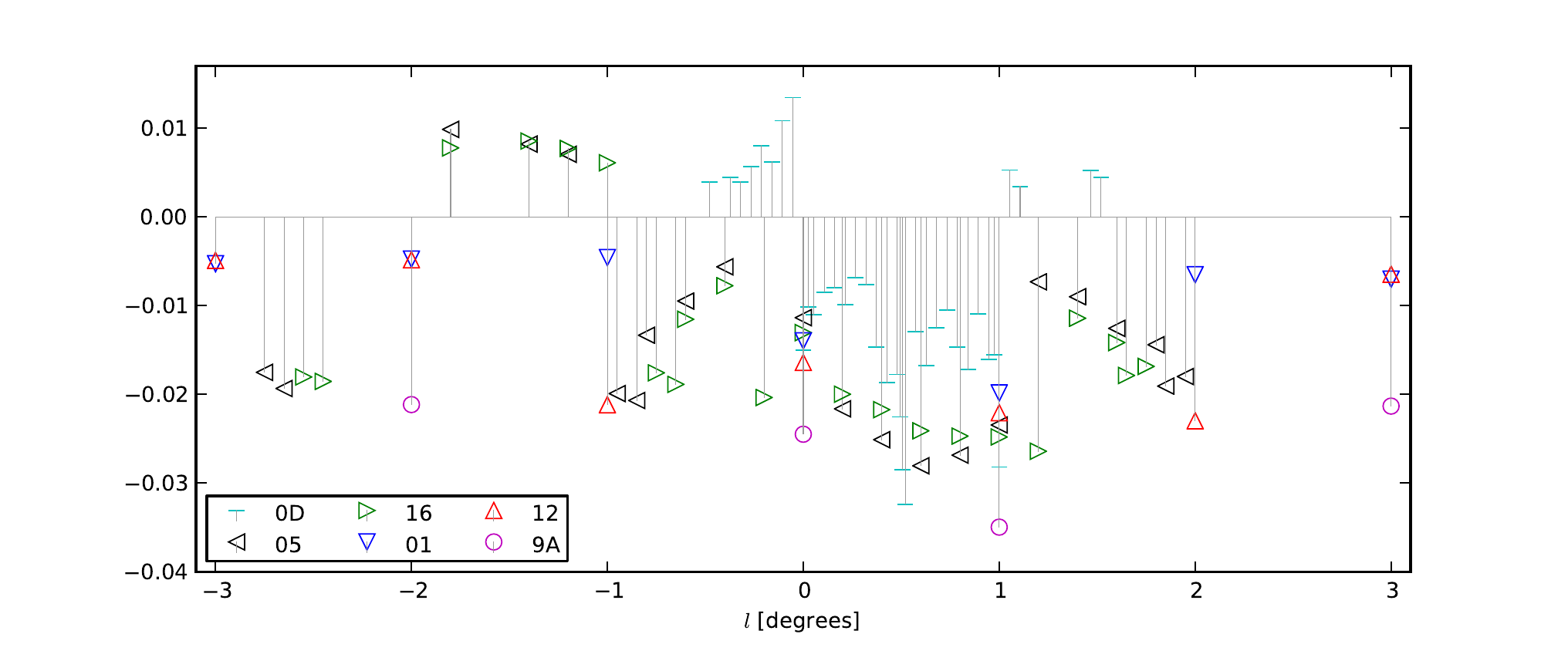}
\caption{ALS ghost pattern for baselines  9A (36m), 01 and 12 (144m), 05 and 16 (720m), 0D (2.7km).}
\label{fig:als_stem} 
\includegraphics[width=.949\textwidth]{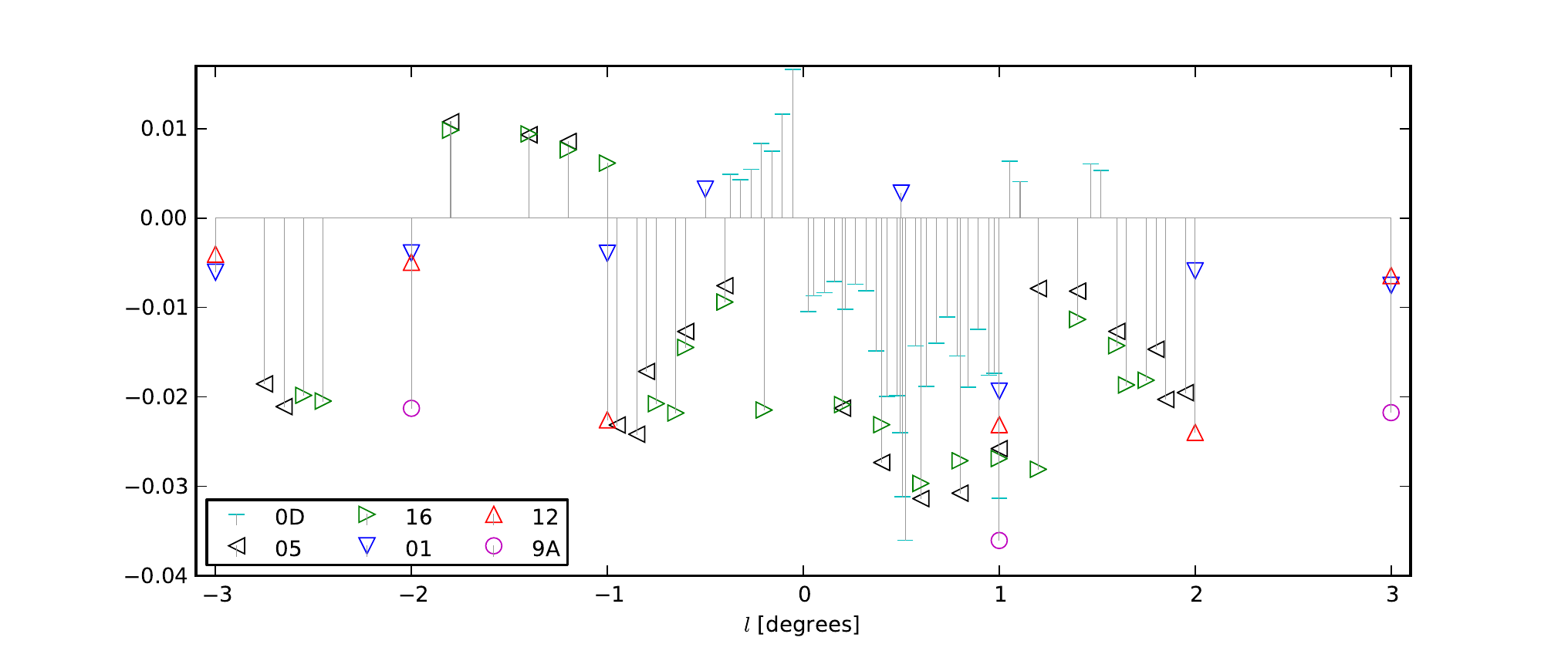}
\caption{LS ghost pattern for baselines  9A (36m), 01 and 12 (144m), 05 and 16 (720m), 0D (2.7km).}
\label{fig:ls_stem} 
\end{figure*}

\newlength{\plotheight}
\plotheight = .3\columnwidth

\begin{figure*}%
\centering
\begin{minipage}{\columnwidth}\begin{center}\parbox{0cm}{\scriptsize~~ALS:~01~(144m)}%
\includegraphics[height=\plotheight]{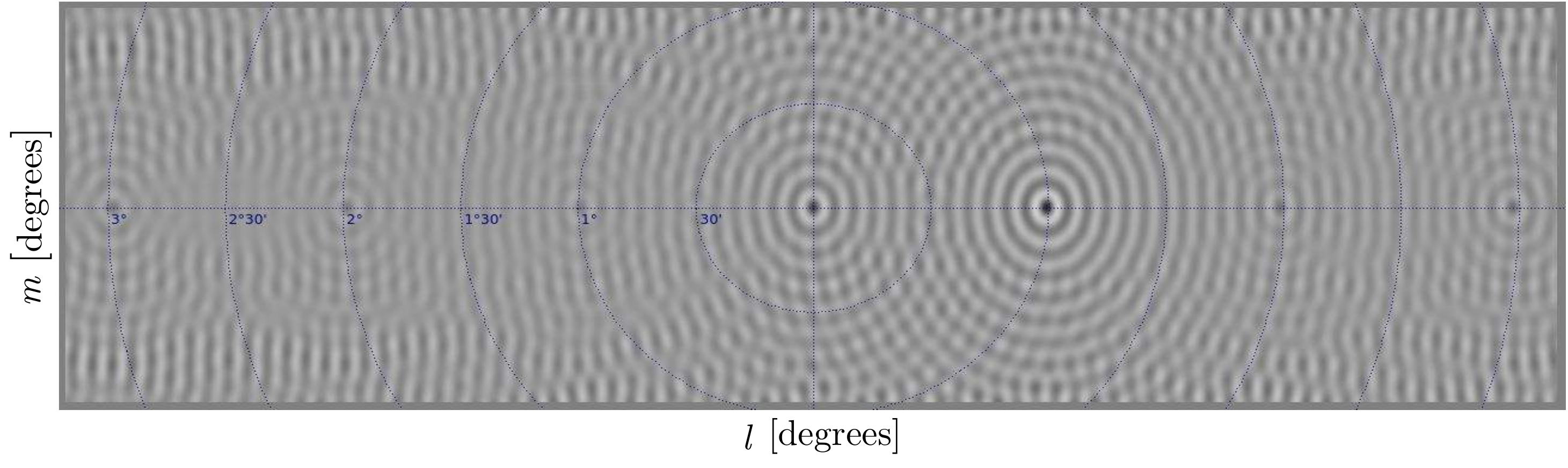}%
%\caption{ALS: Ghost pattern for baseline 01 (144m).}\label{fig:fig_3a}
\end{center}\end{minipage}%
\begin{minipage}{\columnwidth}\begin{center}\parbox{0cm}{\scriptsize~~ALS:~05~(720m)}%
\includegraphics[height=\plotheight]{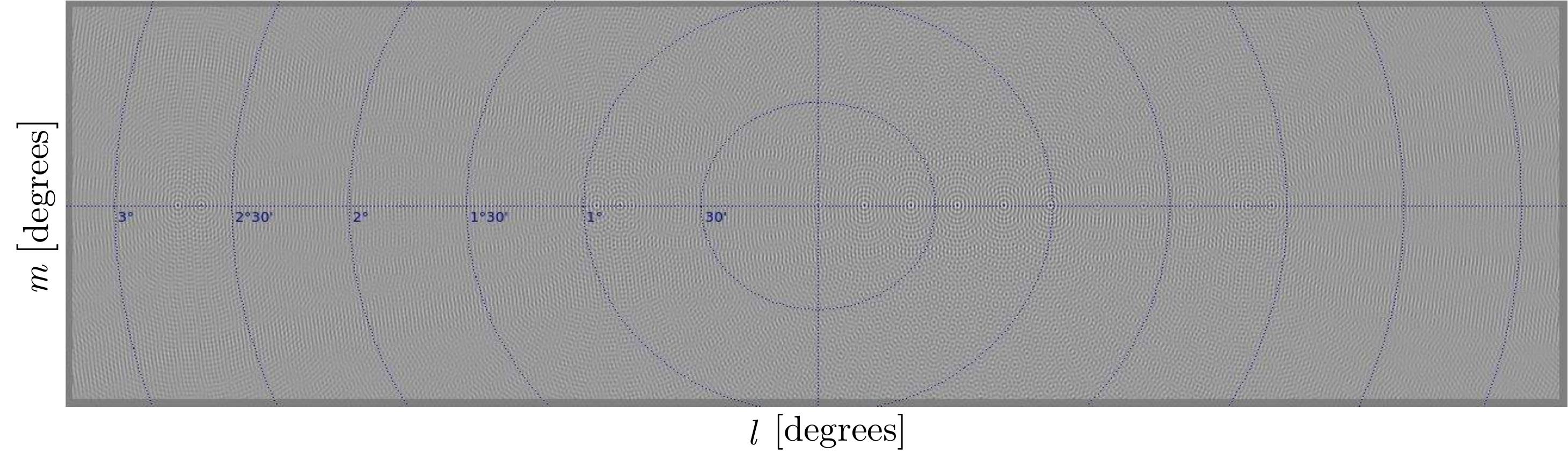}
%\caption{ALS: Ghost pattern for baseline 05 (720m).}\label{fig:fig_3b}%
\end{center}\end{minipage}\\
\begin{minipage}{\columnwidth}\begin{center}\parbox{0cm}{\scriptsize~~ALS:~12~(144m)}%
\includegraphics[height=\plotheight]{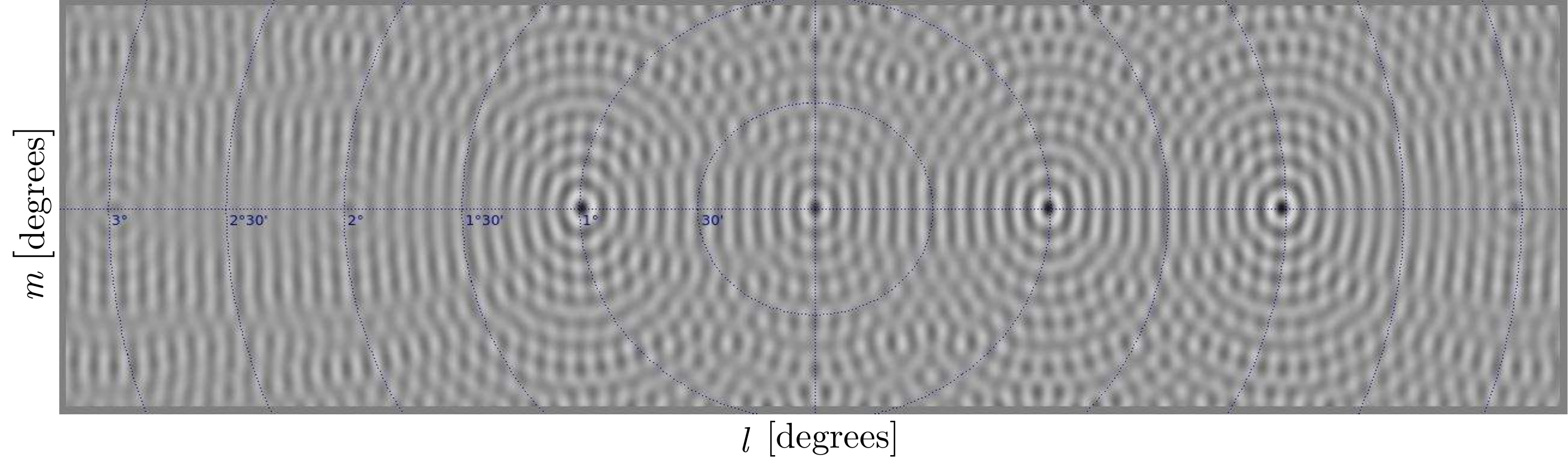}
%\caption{ALS: Ghost pattern for baseline 12 (144m).}\label{fig:fig_3c}%
\end{center}\end{minipage}%
\begin{minipage}{\columnwidth}\begin{center}\parbox{0cm}{\scriptsize~~ALS:~16~(720m)}%
\includegraphics[height=\plotheight]{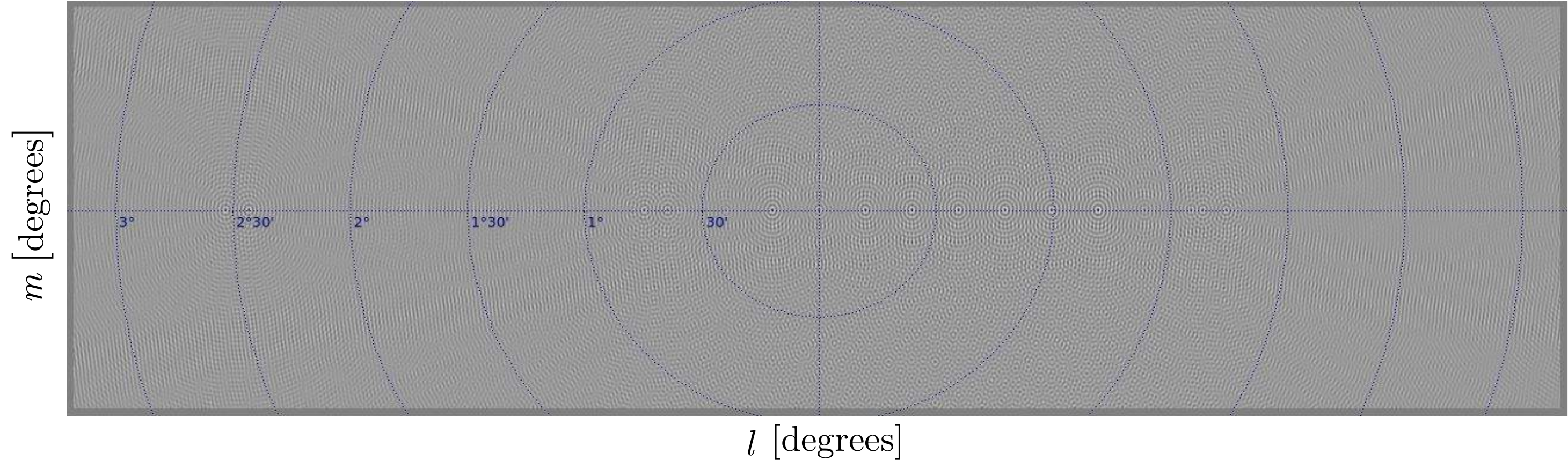}
%\caption{ALS: Ghost pattern for baseline 16 (720m).}\label{fig:fig_3d}
\end{center}\end{minipage}\\
\begin{minipage}{\columnwidth}\begin{center}\parbox{0cm}{\scriptsize~~ALS:~9A~(36m)}%
\includegraphics[height=\plotheight]{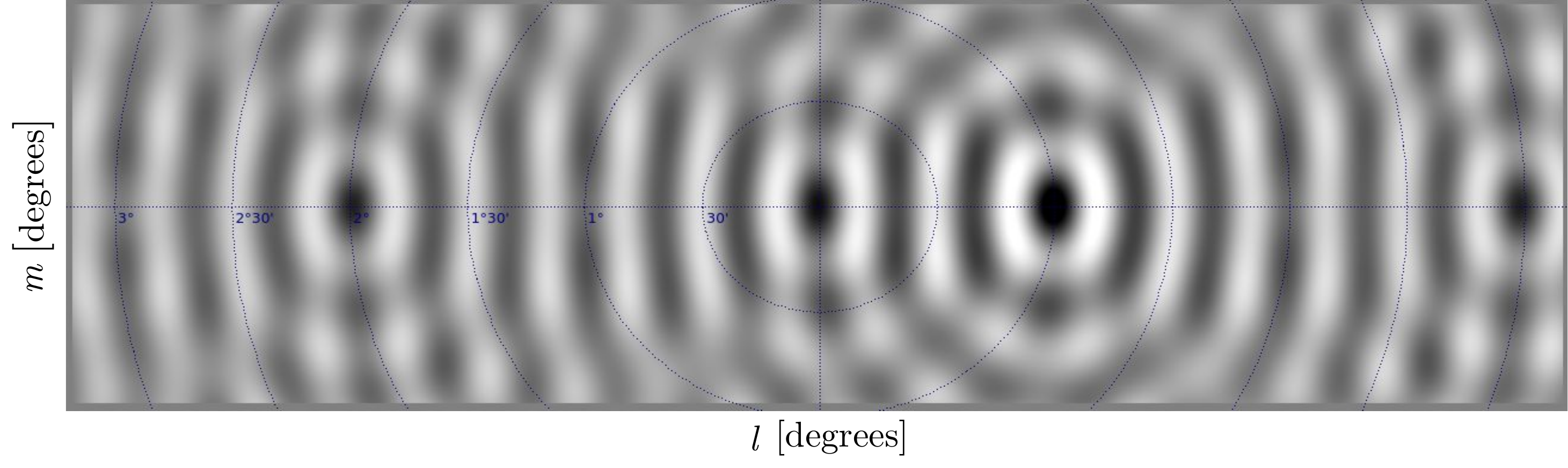}
%\caption{ALS: Ghost pattern for baseline 9A (36m).}\label{fig:fig_3e}%
\end{center}\end{minipage}%
\begin{minipage}{\columnwidth}\begin{center}\parbox{0cm}{\scriptsize~~ALS:~0D~(2.7km)}%
\includegraphics[height=\plotheight]{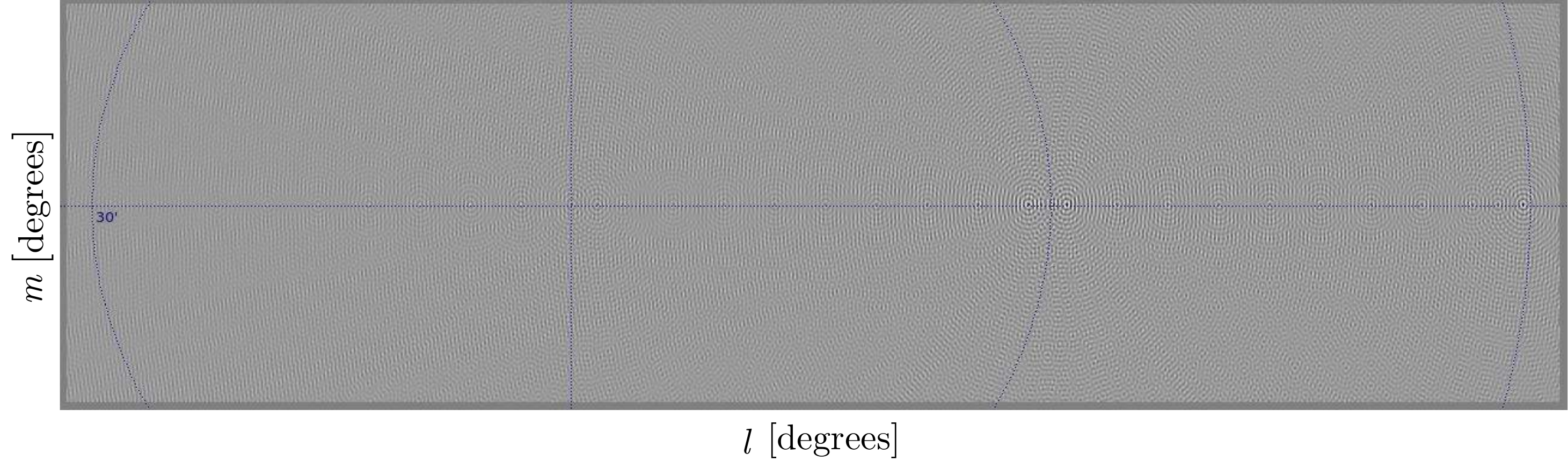}
%\caption{ALS: Ghost pattern for baseline 0D (2.7km).}\label{fig:fig_3f}%
\end{center}\end{minipage}
\begin{minipage}{\columnwidth}\begin{center}\parbox{0cm}{\scriptsize~~LS:~01~(144m)}%
\includegraphics[height=\plotheight]{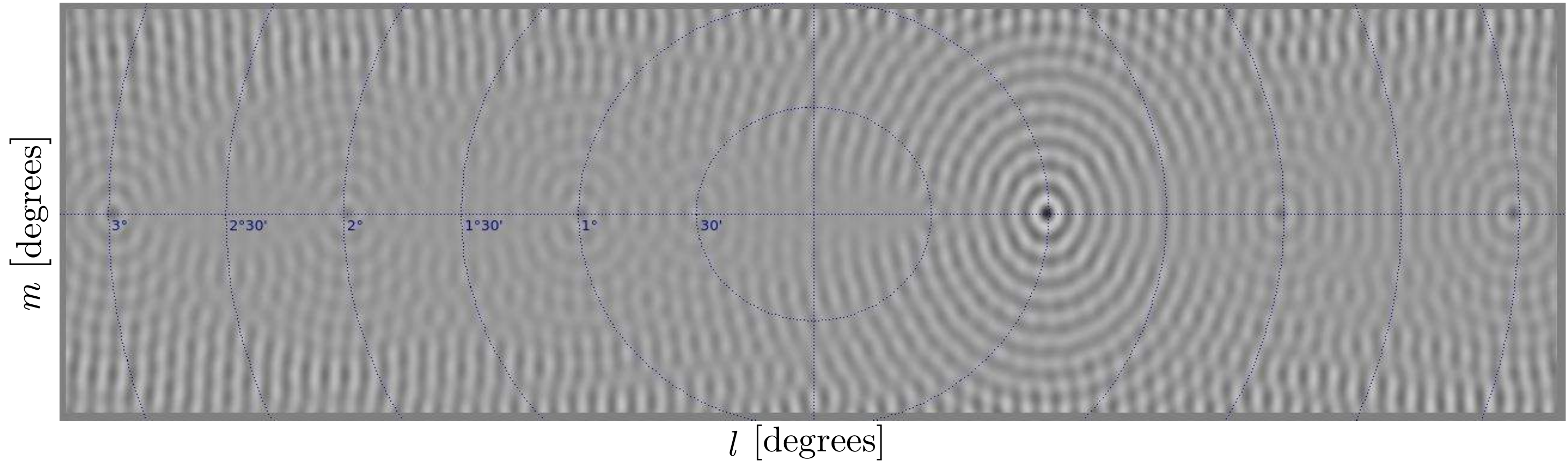}
%\caption{LS: Ghost pattern for baseline 01 (144m).}\label{fig:fig_4a}%
\end{center}\end{minipage}
\begin{minipage}{\columnwidth}\begin{center}\parbox{0cm}{\scriptsize~~LS:~05~(720m)}%
\includegraphics[height=\plotheight]{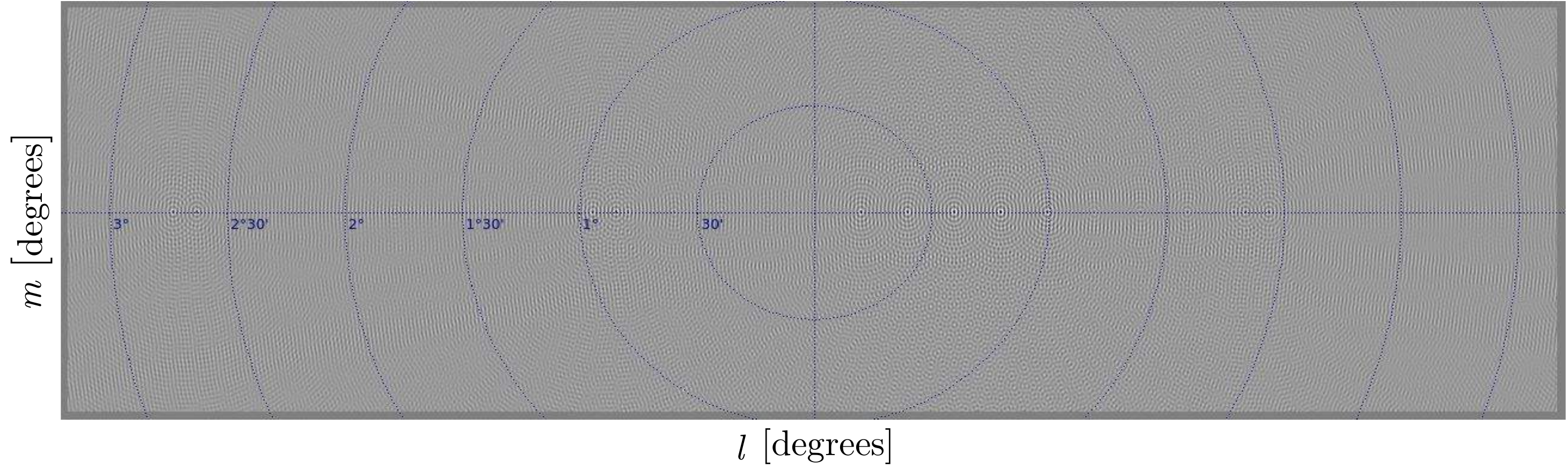}
%\caption{LS: Ghost pattern for baseline 05 (720m).}\label{fig:fig_4b}%
\end{center}\end{minipage}\\
\begin{minipage}{\columnwidth}\begin{center}\parbox{0cm}{\scriptsize~~LS:~12~(144m)}%
\includegraphics[height=\plotheight]{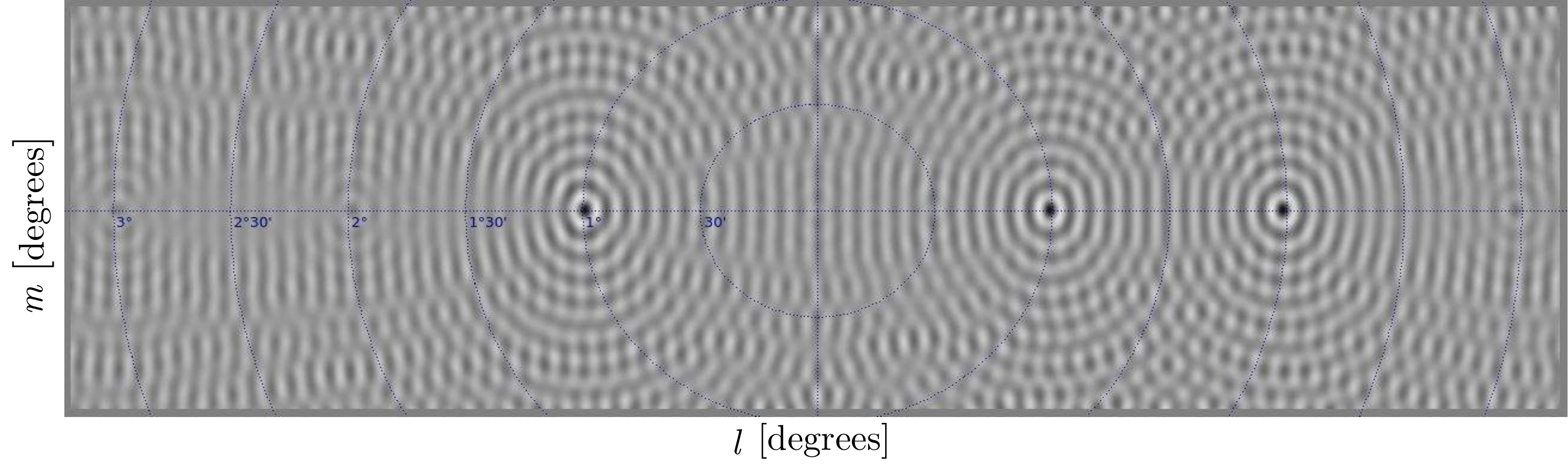}
%\caption{LS: Ghost pattern for baseline 12 (144m).}\label{fig:fig_4c}%
\end{center}\end{minipage}
\begin{minipage}{\columnwidth}\begin{center}\parbox{0cm}{\scriptsize~~LS:~16~(720m)}%
\includegraphics[height=\plotheight]{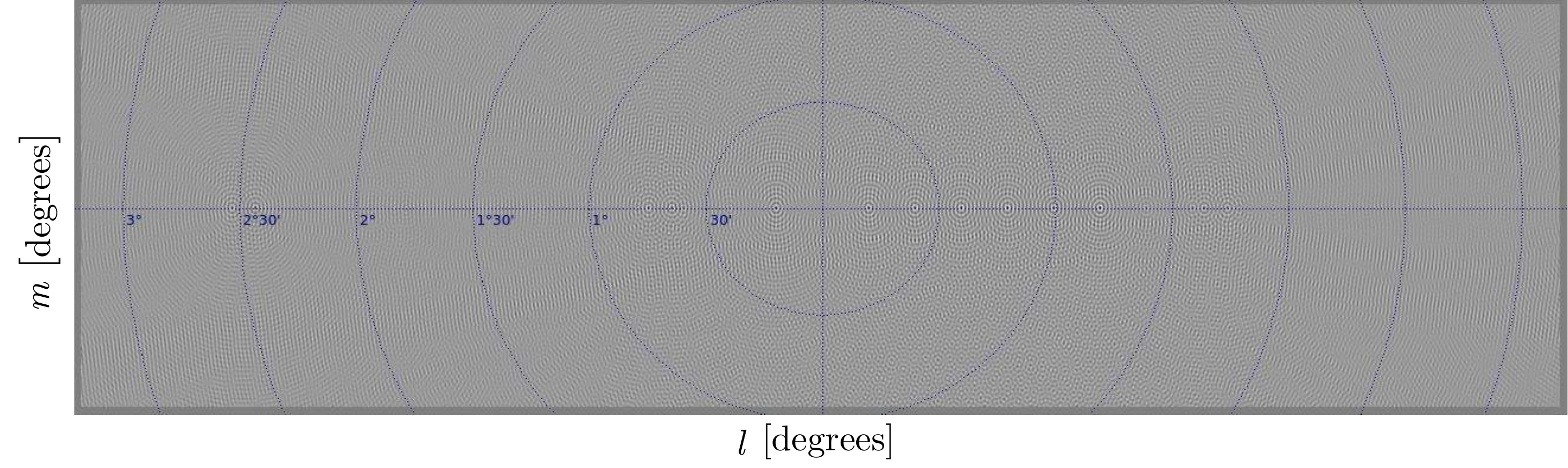}
%\caption{LS: Ghost pattern for baseline 16 (720m).}\label{fig:fig_4d}%
\end{center}\end{minipage}\\
\begin{minipage}{\columnwidth}\begin{center}\parbox{0cm}{\scriptsize~~LS:~9A~(36m)}%
\includegraphics[height=\plotheight]{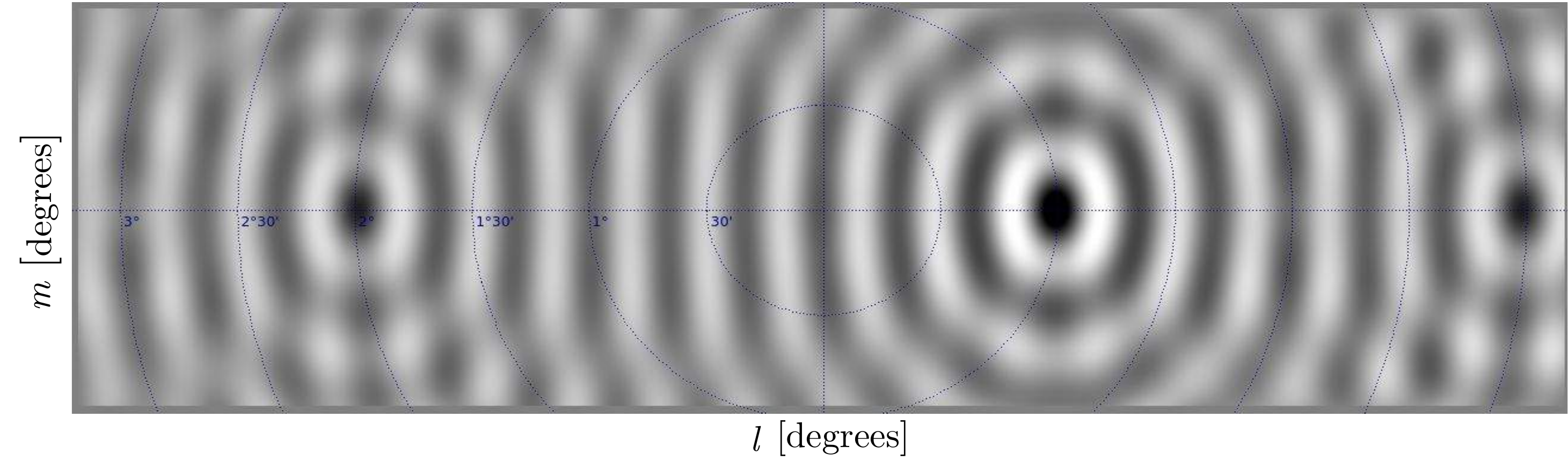}
%\caption{LS: Ghost pattern for baseline 9A (36m).}\label{fig:fig_4e}%
\end{center}\end{minipage}
\begin{minipage}{\columnwidth}\begin{center}\parbox{0cm}{\scriptsize~~LS:~0D~(2.7km)}%
\includegraphics[height=\plotheight]{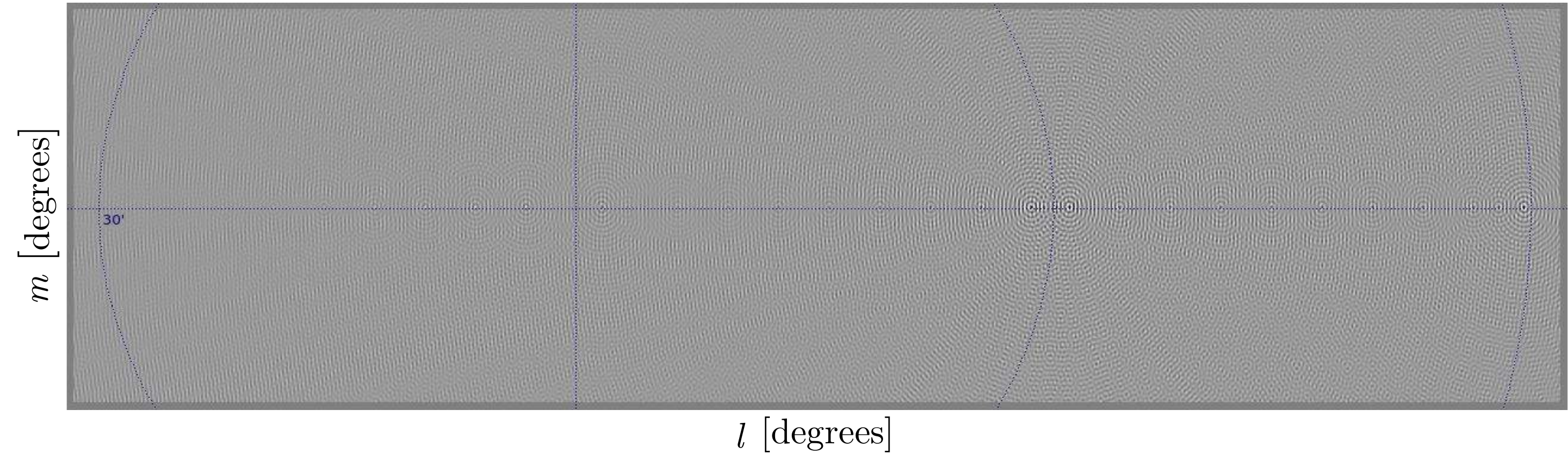}
%\caption{LS: Ghost pattern for baseline 0D (2.7km).}\label{fig:fig_4f}
\end{center}\end{minipage}
\caption{Dirty images of ALS and LS ghost patterns for baselines
9A (36m), 01 and 12 (144m), 05 and 16 (720m), and 0D (2.7km).}\label{fig:sim_ghosts}
\end{figure*}

Fig.~\ref{fig:theor_stem} displays the theoretically determined distilled ghost patterns for a selection of baselines
(9A: 36m, 01 and 12: 144m, 05 and 16: 720m, 0D: 2.7 km). We also obtain a set of simulated distilled ghost patterns 
for the same set of baselines (Figs.~\ref{fig:als_stem}, \ref{fig:ls_stem}) as follows:

\begin{itemize}

\item We run ALS or LS calibration on a set of simulated visibilities, and derive the calibrated visibilities;

\item We image the calibrated visibilities for each baseline (using the \emph{lwimager} program -- an FFT-based
imager derived from the CASA libraries and functionally equivalent to the CASA imager). The resulting
dirty maps are given in Fig.~\ref{fig:sim_ghosts};

\item We measure fluxes at the ghost source positions in the resulting dirty images, resulting in Figs.~\ref{fig:als_stem} and
\ref{fig:ls_stem}.

\end{itemize}

As predicted by Sect.~\ref{sec:t_der}, short baselines yield a few coarsely-spaced ghosts (e.g. 9A), while long baselines (e.g. 0D) yield many finely-spaced ghosts. 

% It is important to note that the subscripts in Fig.~\ref{fig:fig_2a}--Fig.~\ref{fig:fig_6d} should not be confused with matrix indices, but should rather be interpreted as 
% baseline labels (these two notational interpretations are used interchangeably in the article). 
%% OMS: same thing, isn't it?

Comparing Figs.~\ref{fig:theor_stem}, \ref{fig:als_stem} and \ref{fig:ls_stem}, the following general observations (and subsequent conclusions) can be made:
\begin{itemize}
 \item The bright ghost sources in Fig.~\ref{fig:theor_stem} and the bright sources in 
 Figs.~\ref{fig:als_stem}--\ref{fig:ls_stem} show up at the same $lm$ coordinates. This validates Eq.~\ref{eq:G_inv}.    
 \item The weaker sources given by the theoretical ghost patterns (Fig.~\ref{fig:theor_stem}) are not 
 visible in Fig.~\ref{fig:als_stem}. Furthermore, there are small differences in flux between the 
 theoretical ghost patterns and the measured fluxes of the corresponding ghost sources in Fig.~\ref{fig:als_stem}.
 Note, however, that the dirty images are dominated by the sidelobes of the brighter ghost sources -- which in general (for $n > 2$) are not amenable to normal deconvolution, since each ghost spread function (Eq.~\ref{eq:gsf}) is different. This both masks the fainter ghost sources, and distorts the flux measurements.

 \item The ghost patterns yielded by ALS and LS calibration (Fig.~\ref{fig:als_stem} and Fig.~\ref{fig:ls_stem}) are qualitatively similar, but show different amplitudes. This is understandable, as they are products of slightly different optimization problems, and therefore yield slightly different calibration solutions. In particular, there are negative ghosts at $0^{\circ}$ in Fig.~\ref{fig:als_stem}, while there are none in
 Fig.~\ref{fig:ls_stem}. This implies that ALS tends to also suppress the flux of the modelled source, 
 while LS doesn't. This can be explained by the following argument. The total flux of the sky and the calibration model
 is given by the diagonal terms (autocorrelations) of $\Rcal$ and $\Mcal$, while the total power in the off-diagonal terms is zero.
 When the autocorrelation constraints are ignored (as in LS), there is no 
 restriction on the total flux in the model, which leaves the gain solutions $g_p\conj{g}_q$ in Eq.~\ref{eq:LM} more 
 freedom to fit the mean amplitude of $r_{pq}$ over time. If the autocorrelations are also fitted (as is the case in 
 ALS, Eq.~\ref{eq:cal2}), then the gain solutions must also account for the total flux of the sky ($A_1+A_2$) using a model 
 containing a total flux of only $A_1$. This yields mean gain-amplitudes of slightly above unity in $\Gcal$, and thus 
 below unity in $\Gcal^\top$. This results in a negative ghost source at the phase centre in the distilled 
 ghost pattern $I_{pq}^{\mathcal{G}^{\top}\!-1}$, i.e. flux suppression of the primary source.
 \end{itemize}

\begin{figure*}%
\includegraphics[height=.4\columnwidth,width=\columnwidth]{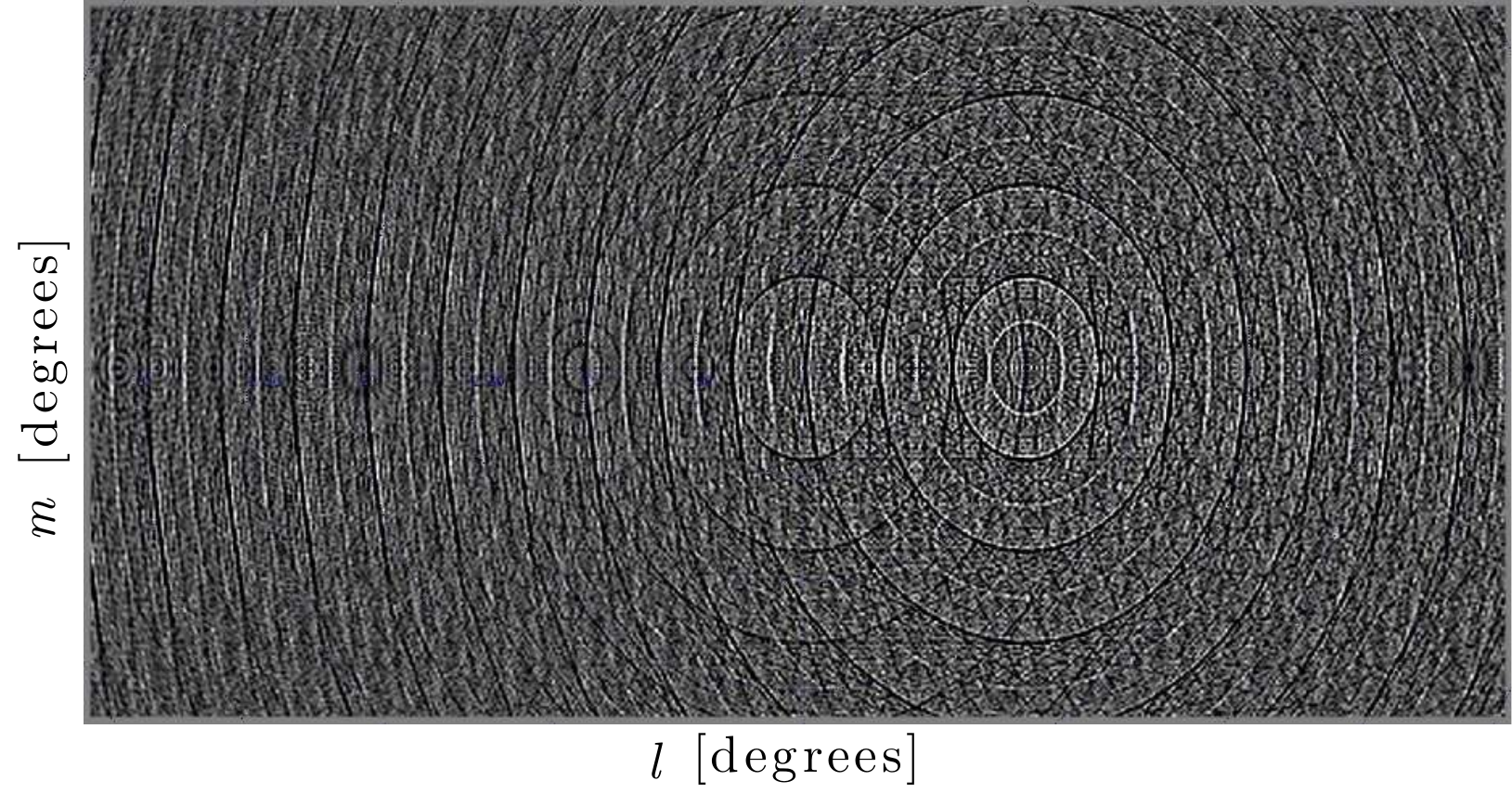}\hfill%
\includegraphics[height=.4\columnwidth,width=\columnwidth]{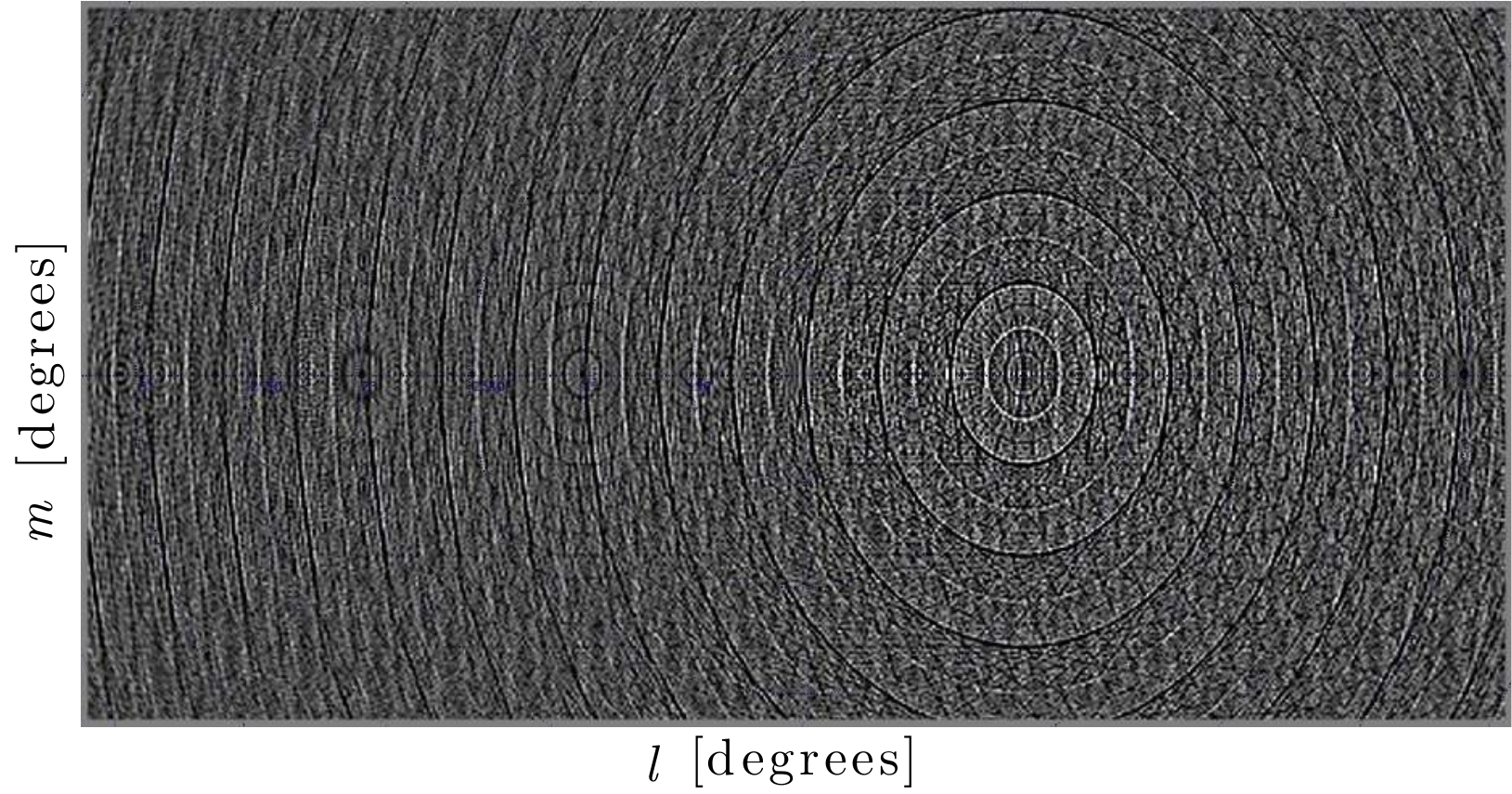}%
\caption{ALS (left) and LS (right) distilled ghost pattern for the full WSRT array, with the $A_2$ source at $1^\circ$.}
\label{fig:patt_ls_1deg}
\end{figure*}

\begin{figure*}%
\includegraphics[height=.4\columnwidth,width=\columnwidth]{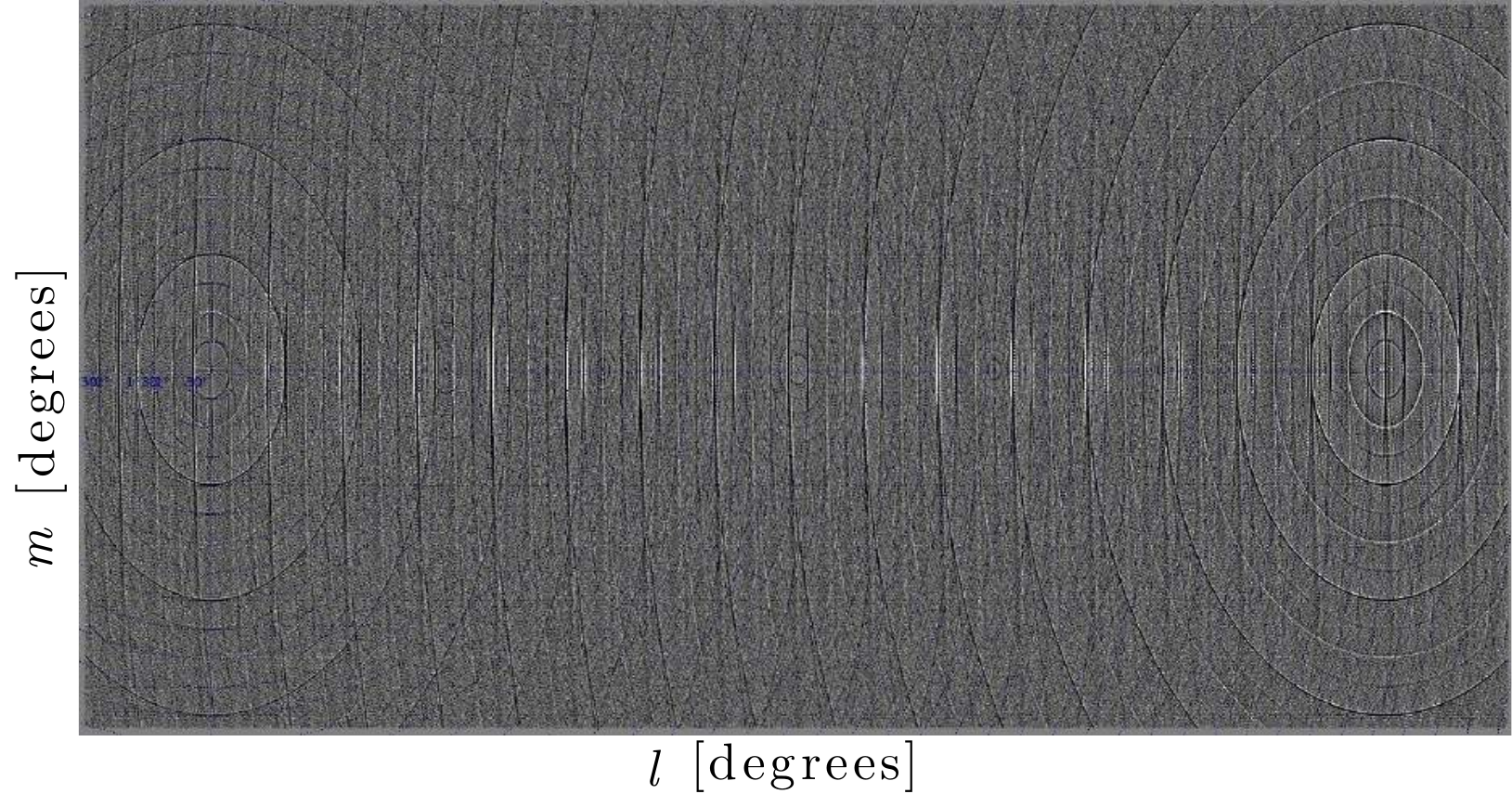}\hfill%
\includegraphics[height=.4\columnwidth,width=\columnwidth]{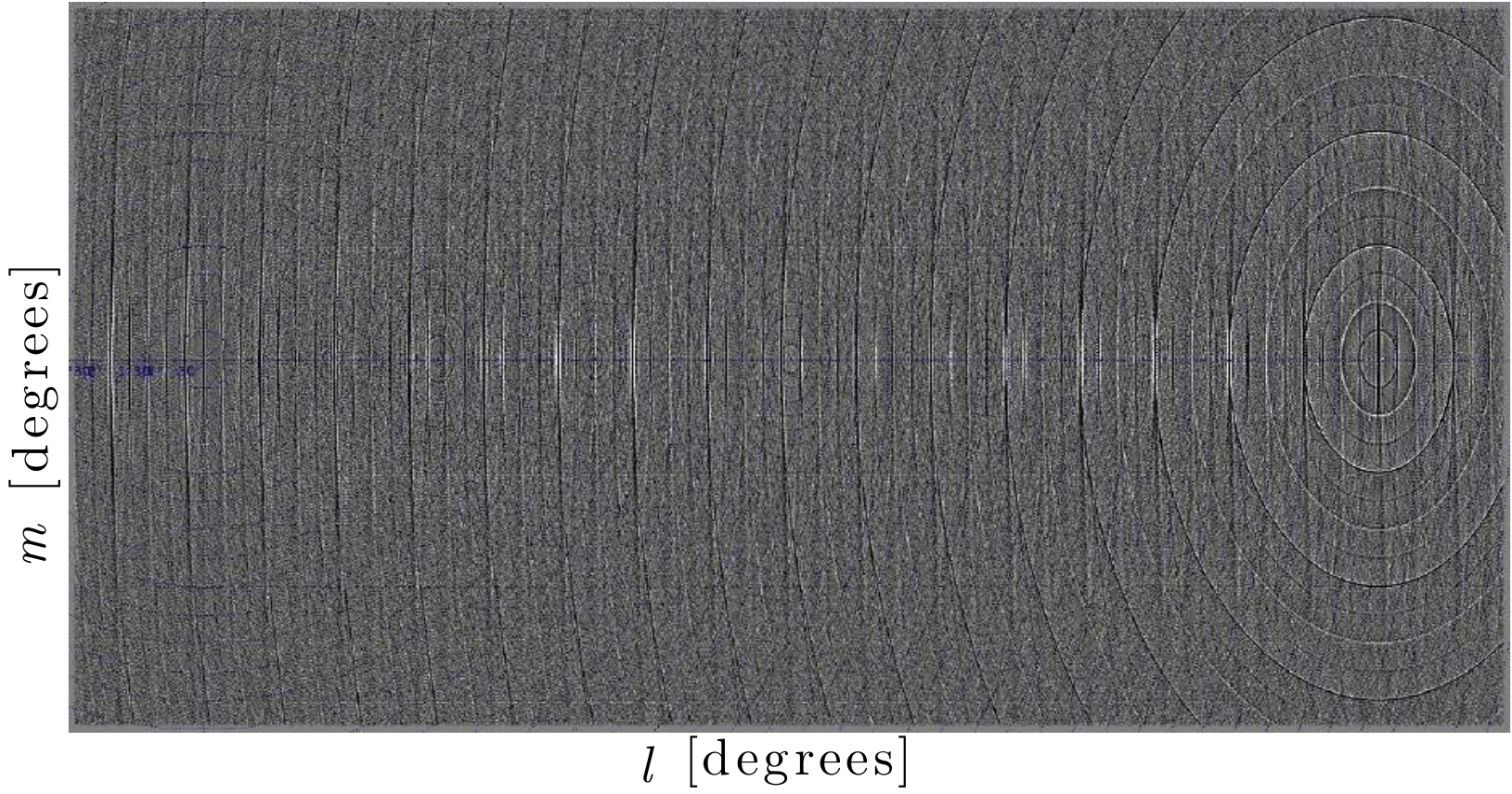}%
\caption{ALS (left) and LS (right) distilled ghost  pattern for the full WSRT array, with the $A_2$ source at $20^\circ$.}
\label{fig:patt_ls_20deg}
\end{figure*}

Figs.~\ref{fig:patt_ls_1deg}--\ref{fig:patt_ls_20deg} display the ghost patterns that are
obtained for ALS and LS calibration using all antennas during imaging. In Fig.~\ref{fig:patt_ls_1deg}, the
secondary source is at $l=1^{\circ}$, while in Fig.~\ref{fig:patt_ls_20deg} it is placed at $l=20^{\circ}$
(thus qualitatively reproducing the observational scenario of Fig.~\ref{fig:2004ghosts}). In
Fig.~\ref{fig:patt_ls_20deg} only the ``inner ghosts'' (the ghosts between the primary and
secondary source) are visible, while some ``outer ghosts'' are emerging in Fig.~\ref{fig:patt_ls_1deg}.

%Furthermore, the field center of
%the measurement set that was used was at right ascension 1.49401 rad and declination 0.870954 rad.

Fig.~\ref{fig:stem_14ant} displays the theoretically determined distilled ghost pattern for 
the full WSRT array, as percentage of $A_2$ flux (for the $A_2=0.2$ at $1^\circ$ case). Compare these to the per-baseline 
patterns in Fig.~\ref{fig:theor_stem}. The pattern exhibits a number of interesting features:

\begin{itemize}

\item Most (though not all) ghosts have negative amplitudes; the positive ghosts tend to be fewer and much weaker.

\item The strongest response is the ``flux suppression'' ghost at the $A_2$ position ($1^\circ$, or $k=\phi_0$). 
At about $13\%$, it is perfectly consistent with the amount of flux suppression normally observed when calibrating WSRT 
data with LS. As discussed  in Sect.~\ref{sec:imaging}, the $k=n\phi_0$ positions are shared by the ghost patterns of 
all baselines, and thus favour the formation of strong ghosts. It is not surprising that the $k=\phi_0$ position 
shows the strongest response overall, as that is where the missing flux that the calibration process is trying to fit 
is located. The next-brightest ghost ($\sim6\%$), is at $0^\circ$. As discussed above, this particular ghost is 
specific to ALS. 

\item Curiously, the other ``favoured'' positions ($-2^\circ$, $-1^\circ$, $-2^\circ$, etc.) show a much diminished
response -- only about $1\sim2\%$ -- i.e. weaker than the strongest of the inner ghosts (see e.g. the ``halfway ghost'' 
at $1/2^\circ$, with $2.5\%$). Comparing this to Fig.~\ref{fig:theor_stem}, we can partially understand 
how this comes about: different baselines show a mix of positive and negative responses at these positions, whereas
the $0^\circ$ and $1^\circ$ responses are consistently negative. The terms in the sum of Eq.~\ref{eq:gsf} 
therefore average down at the other positions.

\item Most of the strongest remaining ghosts are the inner ones between the two sources ($0^\circ-1^\circ$). 
However, the ``outer ghosts'' seem to extend indefinitely at the $0.1\sim0.2\%$ level.

\end{itemize}

The latter two points are especially puzzling, and there should be some fundamental mathematical reason for why this should 
be so, but it escapes us at present.

\begin{figure*}%
\begin{center}%
\includegraphics[width=\columnwidth]{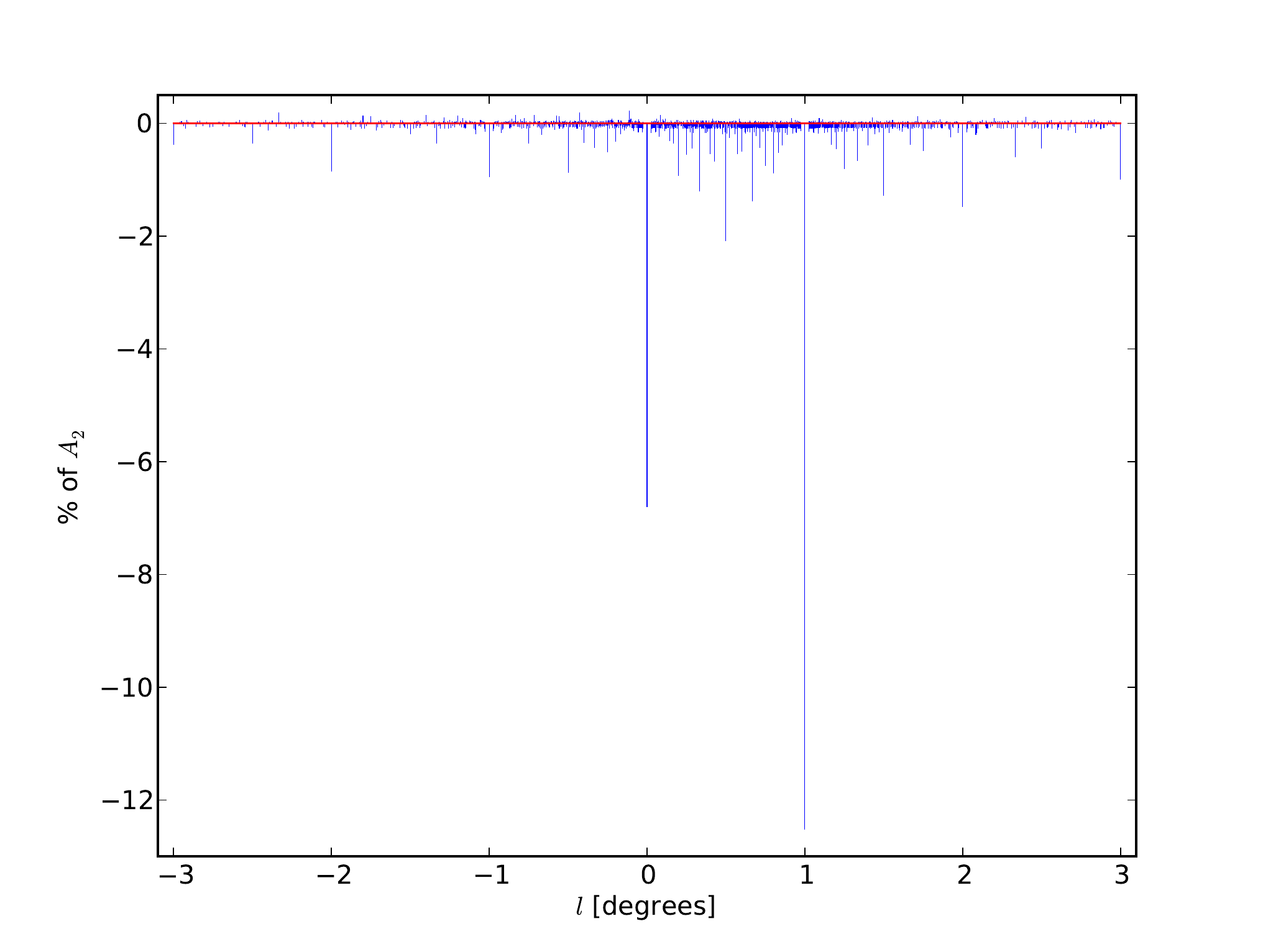}\hfill\includegraphics[width=\columnwidth]{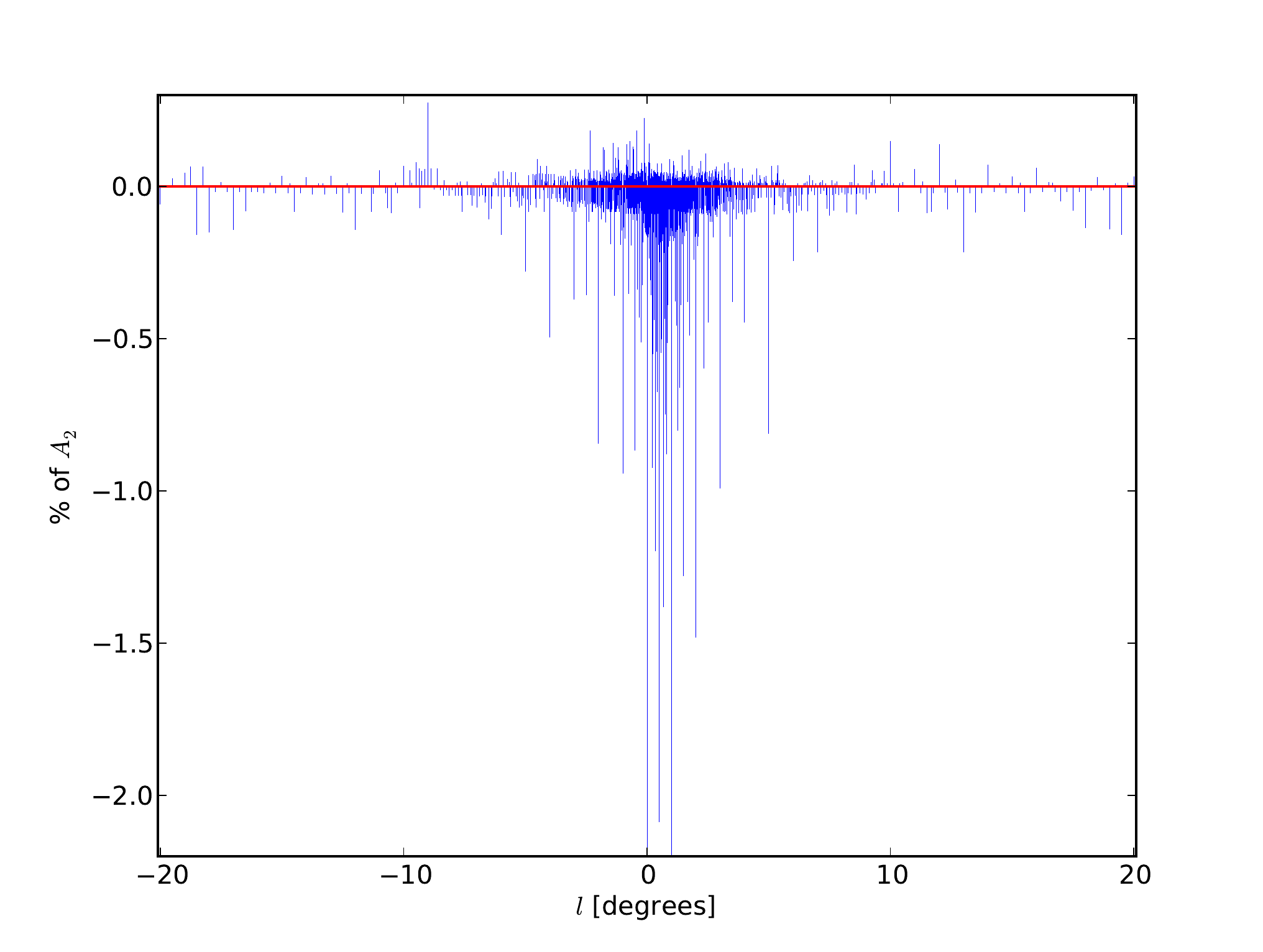}
\caption{\label{fig:stem_14ant}Theoretical residual ghost pattern within $3^\circ$ (left plot) and $20^\circ$ (right plot) 
of phase centre, for the full WSRT array. The $A_2$ source was at $1^\circ$. Amplitudes are given as a percentage of 
$A_2$ flux. The $3^\circ$ plot shows all the ghosts, while in the $20^\circ$ plot, only ghosts with 
amplitudes in excess of $\approx 0.009$\% of $A_2$ are shown, and 
the $y$ axis is cut off just below the $1/2^\circ$  ghost -- the $0^\circ$ and $1^\circ$ ghost response thus extends 
well below the plot limits.}%
\end{center}%
\end{figure*}

% \begin{figure*}%
% \centering
% \begin{minipage}[t]{0.48\linewidth}\begin{center}%
% \caption{\label{fig:stem_14ant}Theoretical residual ghost pattern within $3^\circ$ of phase centre, for the full WSRT array. 
% The $A_2$ source was at $1^{\circ}$. Amplitudes are given as a percentage of $A_2$ flux.}%
% \end{center}\end{minipage}
% \begin{minipage}[t]{0.48\linewidth}\begin{center}%
% \includegraphics[width=\columnwidth]{g_pattern_14_20.pdf}%
% \caption{\label{fig:stem_14ant_20}Theoretical residual ghost pattern within $20^\circ$ of phase centre, for the full WSRT array. 
% The $A_2$ source was at $1^{\circ}$.
% Only ghosts with amplitudes in excess of $\approx 0.009$\% of $A_2$ are shown. The $y$ axis is cut off just below the $1/2^\circ$ 
% ghost -- the $0^\circ$ and $1^\circ$ ghost response extends below the plot limits.}%
% \end{center}\end{minipage}
% \end{figure*}

\subsection{Dependence on flux ratio}

Now that Eq.~\ref{eq:G_inv} has been validated, the natural question arises of how $l_0,m_0,A_1$ and $A_2$
influence the amplitudes of the ghost pattern. Proposition~\ref{prop:3} implies that the position of the secondary
source $l_0,m_0$ has no influence on the amplitudes -- it only stretches or shrinks the ghost patterns, and determines
their orientation. (This also explains why it was sufficient to verify the validity of Eq.~\ref{eq:G_inv} at only
one position of the secondary source.) The source fluxes, obviously, do have an effect. As discussed in
Section~\ref{sec:sky}, the matrix $\bmath{\mathcal{G}}$ is determined by $A_2$ (which is equivalent to the flux ratio, since
we've been assuming $A_1=1$ throughout), which implies that the ghost amplitudes given by the various $d_{k,pq}$
coefficients are dependent on $A_2$.

%$\gamma_{\bmath{s},pq}^{\mathcal{G}^{\top}\!-1}$ and
%$\gamma_{\bmath{s},pq}^{(\mathcal{G}^{\top}\!-1)\odot\mathcal{R}}$ .

The actual ghost amplitudes do not have a simple analytic representation as they are ultimately determined by the interaction between the largest eigenvalue of $\Rcal$ and its associated eigenvector (Eq.~\ref{eq:eq_2}).
We can, however, empirically show an approximately linear dependence on $A_2$. Let us postulate this dependence:

\begin{equation}
\label{eq:E_1}
\zeta_k \approx K_k A_2, 
\end{equation}

and find an estimate for each $K_k$ over a range of $A_2$ values using least-squares. The relative magnitude of the
error of the fit:

\begin{equation}
\label{eq:E_1err}
\epsilon_k = \bigg | \frac{\zeta_k - K_k A_2}{\zeta_k} \bigg |, 
\end{equation}

as a function of $A_2$ for the thirteen brightest ghosts is plotted in Fig.~\ref{fig:errorG}. This shows that most ghosts
vary linearly with $A_2$ to within 10\%. Curiously, the flux suppression ghost ($1^\circ$) is linear to within $1\%$,
but the ghosts at $2^\circ$ and $3^\circ$ are the least linear of all.\footnote{Yet another mathematical 
puzzle raised by the ghost phenomenon. We have no theoretical explanation for this at present!} Approximate linear models for the
amplitudes of  all the ghosts in the distilled ghost pattern can be derived in this manner.

Consider now the dependence on $A_1$ (for which we've used a fixed value of 1 Jy until now). It is obvious from the calibration equation that 
rescaling the true sky and the model sky by the same factor will have no effect on the solutions matrix $\Gcal$, which completely determines the $\zeta_k$ coefficients, i.e. the distilled ghost amplitudes. This means that a calibration problem with fluxes of 
$A_1=A_1' \neq 1$, $A_2=A_2'$ will produce the same coefficients as one with $A_1=1,$ $A_2=A_2'/A_1'$. In other words, {\em the distilled ghost pattern is determined by the flux ratio of the two sources rather than their absolute fluxes}. 

The ghost pattern in the corrected or residual visibilities, on the other hand, is a convolution of the distilled pattern with the sky, and therefore {\em will} scale with absolute flux. To be more precise, for fluxes $A_1'$, $A_2'$, the resulting residual ghost amplitudes (Eq.~\ref{eq:zeta}) will be:

\begin{equation}
\label{eq:zeta1}
\zeta^\Delta_k \approx A_1' K_k \frac{A_2'}{A_1'} + A_2' K_{k-\phi_0} \frac{A_2'}{A_1'} = K_k A_2' + K_{k-\phi_0} \frac{A_2'^2}{A_1'}. 
\end{equation}

When $A_2' \ll A_1'$, the first term in the sum dominates, which makes the ghost patterns in the residual image 
\emph{nearly independent of $A_1'$}. This explains the behaviour observed by \citet{Smirnov2010ghosts} and discussed 
in the introduction. 

Fig.~\ref{fig:art_G} shows the theoretically-derived relative ghost source amplitudes
$\zeta_k/A_2$ for the 13 strongest\footnote{Strongest at $A_2=0.5$, to be precise. As the plots show, the 
relative ranking of the ghosts can actually change as a function of $A_2.$} 
ghosts of the distilled ghost pattern, as a function of $A_2$. A true linear dependence
would have yielded  constant horizontal lines; deviation from horizontal indicates deviation from 
linearity. As we saw above, the $2^\circ$ and $3^\circ$ ghosts appears to be the least linear.

Fig.~\ref{fig:art_R} shows the same amplitudes for the residual ghost pattern. Since the residual pattern
is a superposition of two scaled fundamental patterns, the dependence is different due to the additional
linear component given by the second term of Eq.~\ref{eq:zeta1} (since we're plotting relative
amplitudes, the equation should be divided by $A_2'$). As expected, this component becomes negligible for 
$A_2 \ll A_1$. For larger $A_2$ the linear component can actually come to dominate -- note how the $2^\circ$
ghost becomes stronger than the $1^\circ$ ghost for larger values of $A_2$ (which is not surprising, since it
contains the $1^\circ$ component from the distilled ghost pattern, scaled by $A_2$).

\subsection{The role of array redundancy}

WSRT's highly redundant configuration plays a very important role in ghost formation. Theoretically,
this is explained by Eq.~\ref{eq:IpqG1}. The set of all possible ghost positions is discrete, and given by 
$\{k \bmath{s}_0/\phi_0 \}$. Each baseline $pq$ yields ghosts at a specific subset of these positions, i.e. 
at intervals (in $k$) of $\phi_0/\phi_{pq}$. For short baselines, $\phi_{pq}$ is small, and few ghosts are 
produced, and vice versa for long baselines. Positions corresponding to redundant baselines, or more generally 
to common integer factors of multiple $\phi_{pq}$'s, will then host stronger ghosts  due to a contribution 
from multiple baselines. 

This effect is vividly illustarted by Figs.~\ref{fig:theor_stem}--\ref{fig:ls_stem}. The shortest baseline 
(9A, 36m, circle symbol) produces the most widely-spaced ghosts, at intervals of $1^\circ$. The 144m baselines 
(01 and 12, up/down triangles) produce ghosts at $0.25^\circ$, the 720m baselines (05 and 16, left/right triangles)
produce ghosts at $0.05^\circ$, and the longest baseline (0D, 2.7km, horizontal tick marks) produces the most finely
spaced ghosts. Groups of redundant baselines (01 and 12, 05 and 16) yield ghosts at exactly the same positions,
but with different amplitudes (sometimes even of different sign). The difference is explained by the fact that 
the antennas constituting redundant spacings form slightly different sets of baselines to other antennas, and are
thus subject to different calibration constraints.

The positions corresponding to $\{k=n\phi_0\}$ (in this case, multiples of 1 degree) will have contributions from all 
baselines, and indeed (as we've shown above), the $0$ and $1$ positions yield the strongest ghosts in the 
combined pattern. Likewise, the next-strongest ghost appears at the halfway point ($k=\phi_0/2$), since many $\phi_{pq}$'s are even in the (36,108,1332,1404m) WSRT configuration. Other prominent ghosts may be expected at other rational 
fractions of $\bmath{s}_0$, which fully explains earlier observations. 

Equation~\ref{eq:IpqG1} also provides us with a qualitative understanding of ghost patterns for a less 
regularly-spaced East-West array. As a mental experiment, we may pick a length for $\bmath{b}_0$ (say, 1m), and 
imagine  moving the WSRT antennas to new positions such that the spacings are still integer multiples of $\bmath{b}_0$, 
but are mutually prime. The least common multiple $\phi_0$ would then be the product of all spacings, and would be very large
(and most entries of the geometry matrix $\bmath{\Phi}$ would be very large). Furthermore,  no two baselines would yield ghosts
at any common position apart from $\{k=n\phi_0\}$. The resulting pattern would then consist of very many finely spaced 
and weak ghosts, with a lot of interaction between the GSF sidelobes, and would therefore be a lot more noise-like.
Further decreasing $\bmath{b}_0$ would  increase $\phi_0$ even more, thus spacing the ghosts even finer and further 
washing out the overall response. Of course, to within some fraction of the dish size, any conceivable array layout 
can be considered regularly-spaced (with a very large $\phi_0$), so we can only properly talk
about arrays that are more regular (WSRT, small $\phi_0$) or less regular (large $\phi_0$). The argument above suggests
that highly regular array layouts result in more widely separated and stronger ghosts. This is a hitherto unforeseen
disadvantage to redundancy, and should be investigated and quantified in light of current arguments promoting redundancy
in future telescope designs \citep{noorishad-thesis}  so as to exploit the redundancy calibration technique 
\citep{JEN-redundancy}.

Similar considerations apply to fully 2D/3D arrays such as LOFAR and the JVLA, where we may expect the overall ghost responce 
to be a lot more scattered and noise-like. An upcoming Paper II (Grobler et al., in prep) will study this subject in 
more detail. Here we will just note that an exception to the above considerations are ghosts occupying 
the $\{k=n\phi_0\}$ positions, which the theory shows must be yielded 
by all baselines, regardless of redundancy. The two strongest ones, at positions 1 and 0 -- the ``flux suppression ghost''
and the primary source ghost -- sit on top of actual sources, and are therefore not easy to detect as separate 
artefacts. The other $\{k=n\phi_0\}$ positions seem to yield much weaker ghosts in practice (see above). However, 
ghosts at the -1 position have recently been spotted in LOFAR data in two independent instances (de Bruyn, priv. comm., 
Fender, priv. comm.) The latter in particular was associated with a bright transient source. This suggests that the 
-1 ghost can yield a strong response under some conditions. We will  investigate this phenomenon in Paper II.
%applications of number theory to interferometry.}. 

\begin{figure}
\includegraphics[width=\columnwidth]{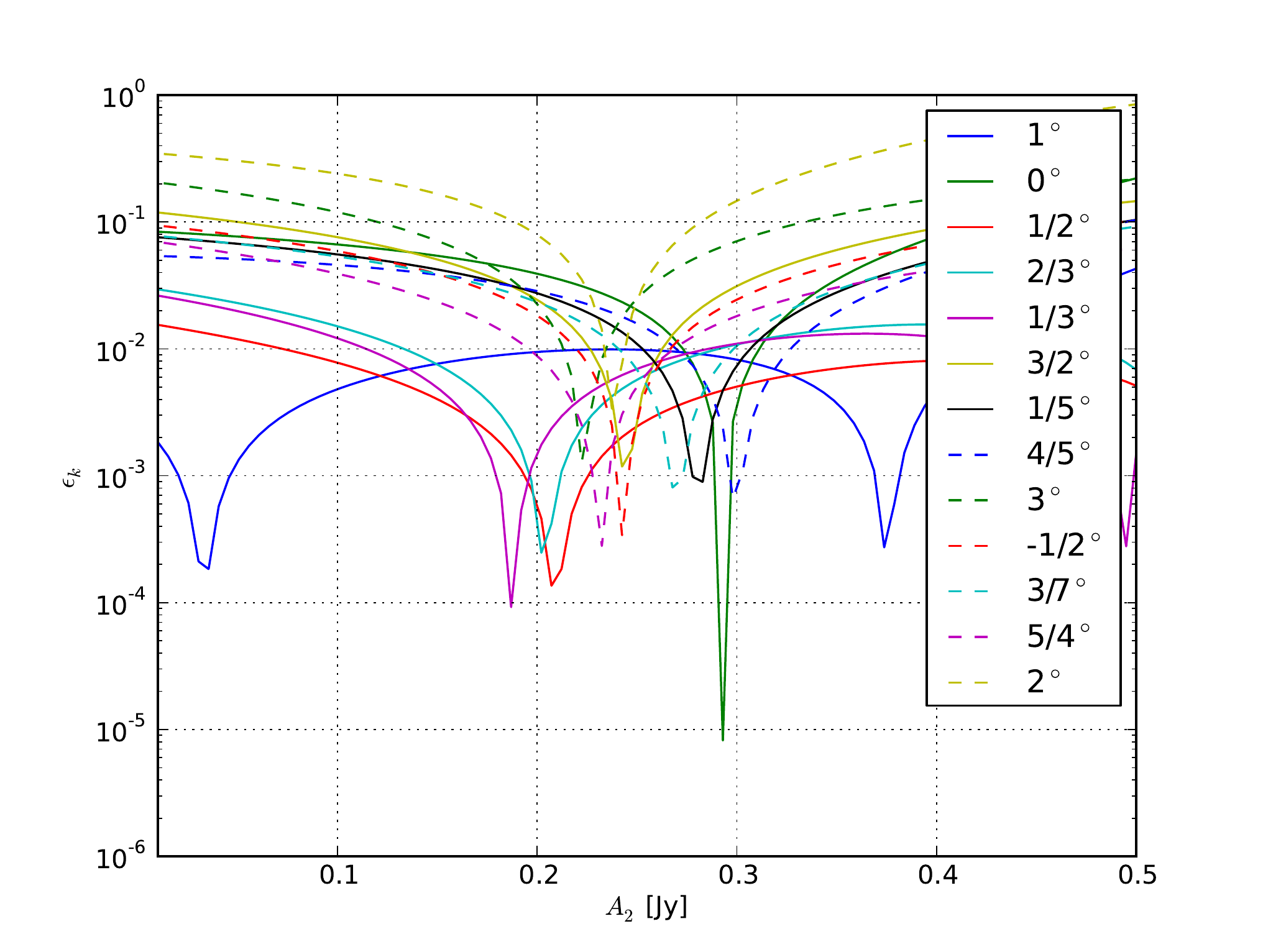}
\caption{The relative error magnitude of the linear fit (Eq.~\ref{eq:E_1err}) for the thirteen brightest ghosts. 
The legends are the ghost locations, with source $A_2$ being at $1^\circ$.}
\label{fig:errorG} 
\end{figure}

\begin{figure*}%
\includegraphics[width=\columnwidth]{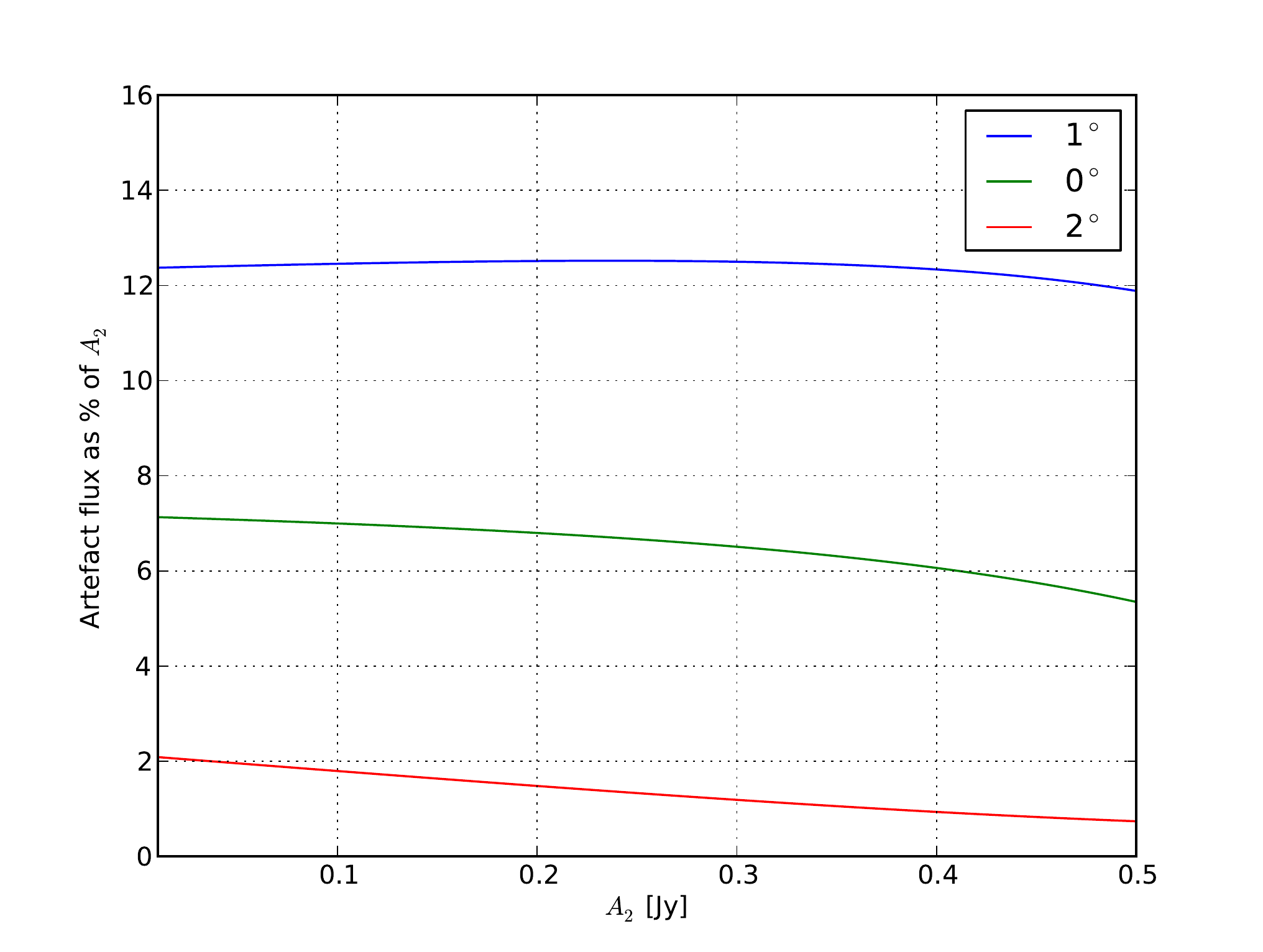}\hfill
\includegraphics[width=\columnwidth]{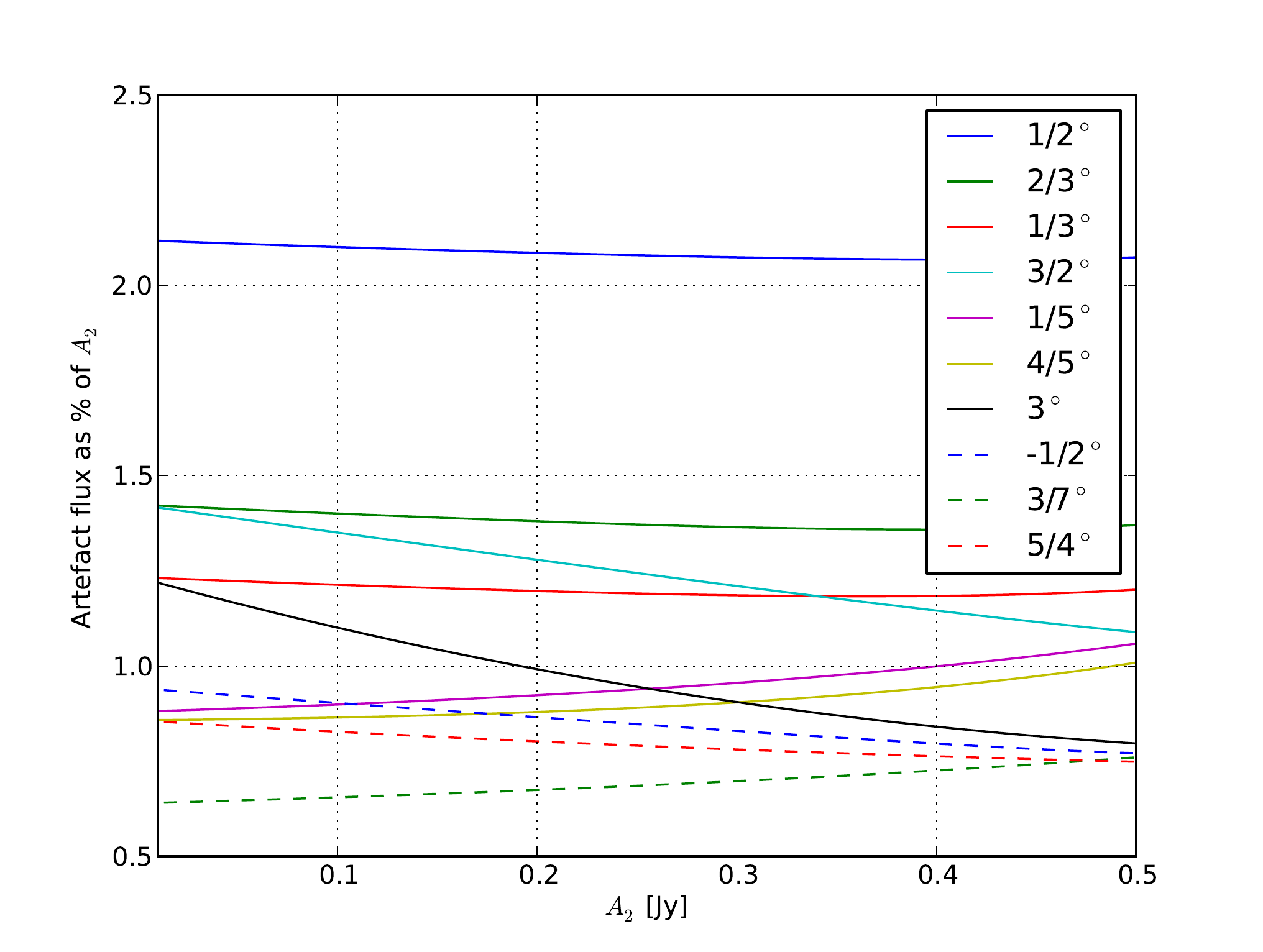}%
\caption{The relative amplitude $|\zeta_k/A_2|$ of the top 13 (top 3 on the left, 4--13 on the right) ghosts in 
the distilled ghost pattern, as a  function of $A_2$. Ranking is by ghost amplitude at $A_2=0.5.$
The ghost positions are indicated by the legend,
with source $A_2$ being at $1^\circ$.}
\label{fig:art_G}
\end{figure*}

\begin{figure*}%
\includegraphics[width=\columnwidth]{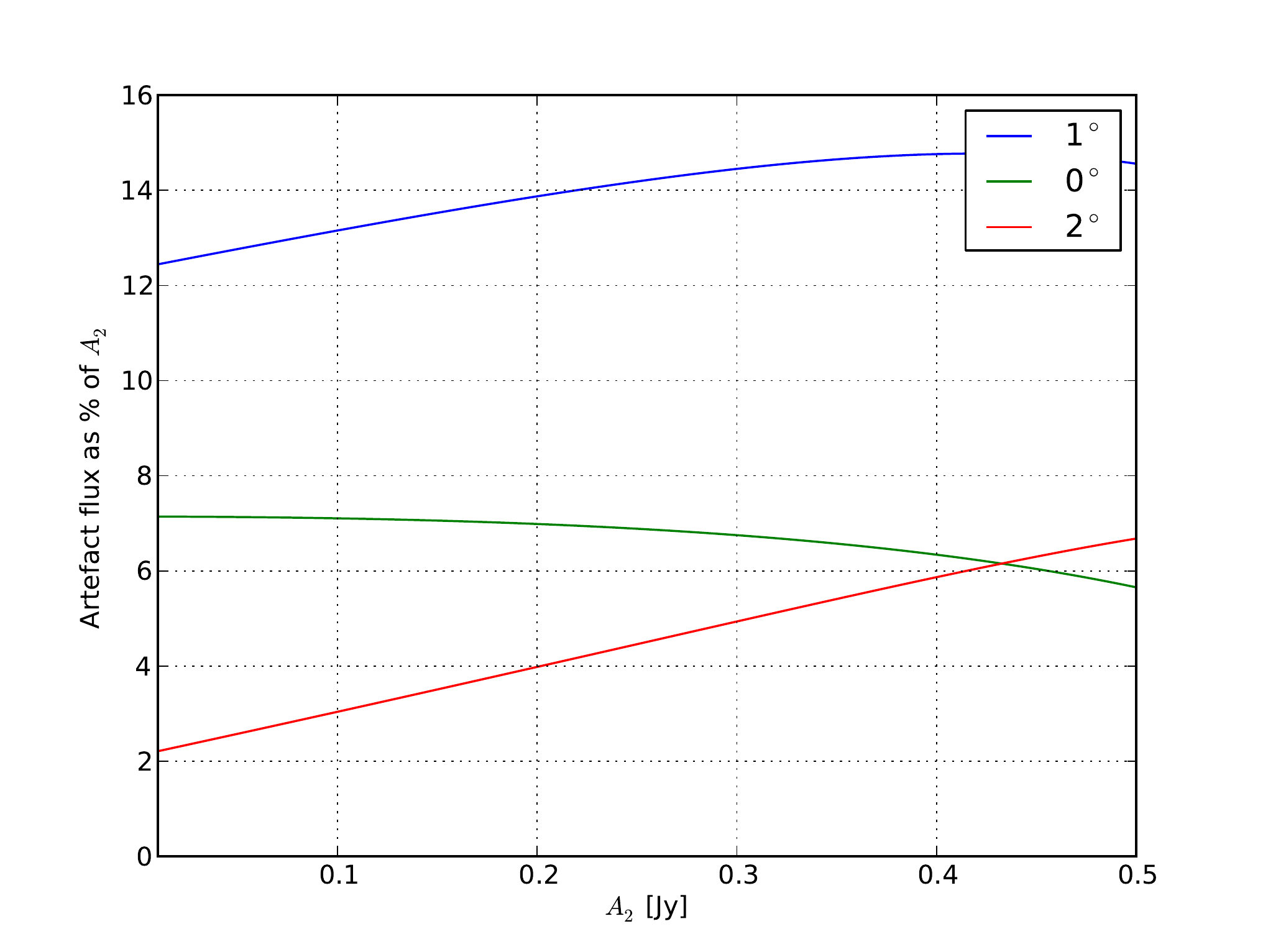}\hfill
\includegraphics[width=\columnwidth]{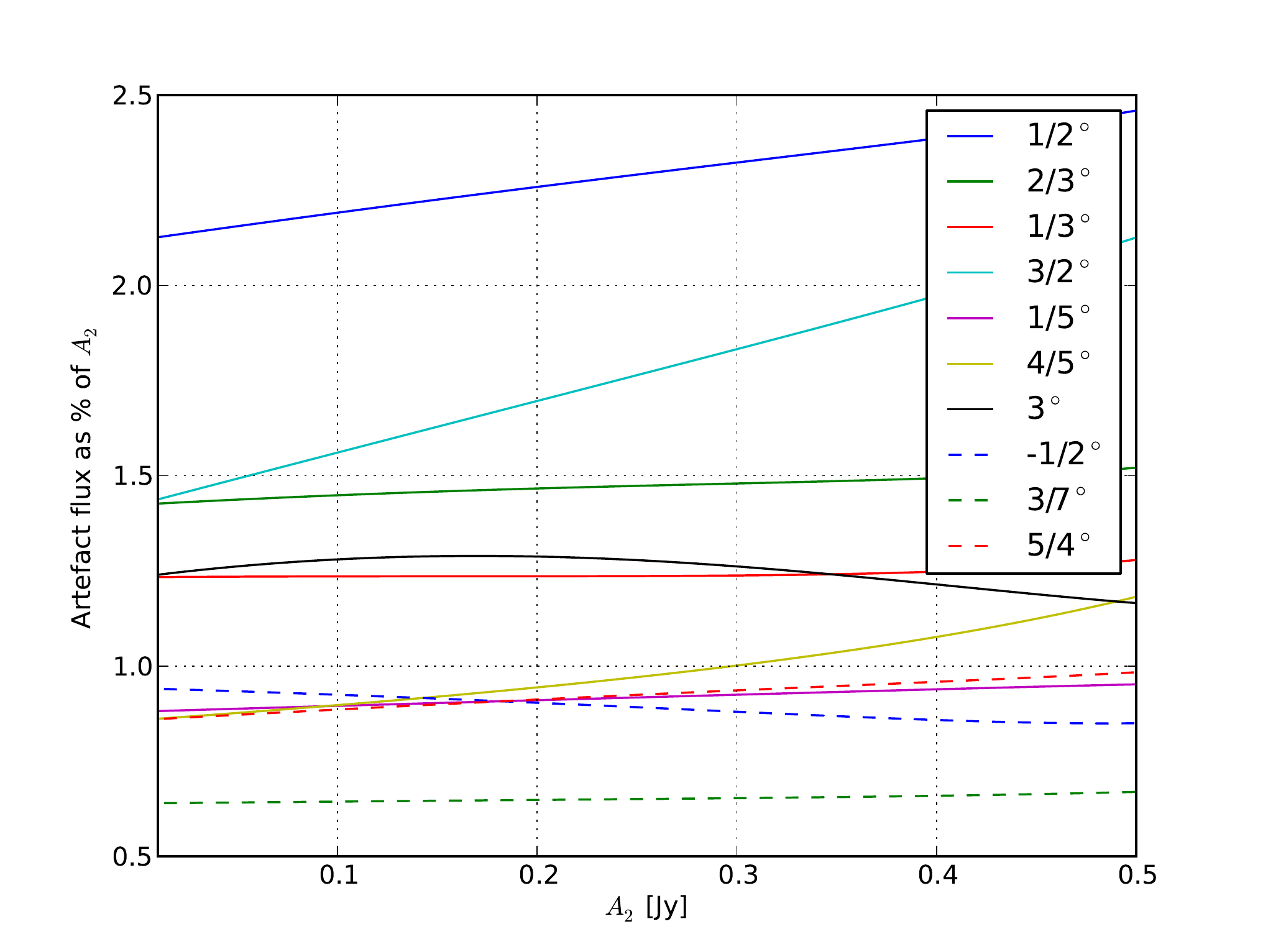}%
\centering
\caption{The relative amplitude $|\zeta^\Delta_k/A_2|$ of the top 13 (top 3 on the left, 4--13 on the right) ghosts in 
the residual ghost pattern, as a  function of $A_2$. Ranking is by ghost amplitude at $A_2=0.5.$
The ghost positions are indicated by the legend, with source $A_2$ being at $1^\circ$.}
\label{fig:art_R}
\end{figure*}

\section{Conclusions and future work}

In this work, we have demonstrated that the ghost source phenomenon is a fundamental aspect of the selfcal process,
and will invariably arise during calibration with an incomplete sky model. It is perhaps surprising that these 
features were not explicitly noted before 2004, but radio interferometrists are accustomed to esoteric instrumental 
artefacts showing up in poorly-calibrated maps, so something that is faint to begin with and goes away with a more accurate 
calibration would not necessarily have attracted attention. The strongest ghost -- the one sitting on top of the
missing model source -- has been indirectly observed as a matter of course, under the guise of ``flux suppression'', 
but being masked by the source itself, it did not give enough clues to the true extent of the phenomenon. Consider also 
that the spacing between ghosts (given by $\bmath{s}_0/\phi_0$) decreases as $\bmath{s}_0$ becomes small, while the
size of the ghosts remains the same (as it is similar to the PSF size). There is therefore some critical $\bmath{s}_0$ 
beyond which the ghosts will begin to blend together into a single ``spoke'' connecting the sources. In the 
authors' experience, poorly calibrated WSRT maps will often exhibit spokes connecting nearby brighter sources, which 
tend to go away once the sky model is improved. It is quite plausible that we've been seeing ``blended ghosts'' all along. 
Since the primary beam makes it more likely for apparently bright sources to be near rather than distant, such blends
may be the most frequent manifestation of the ghost mechanism.

Although we have only studied the two-source case theoretically, empirical results (as well as the observations of
Fig.~\ref{fig:2010ghosts}) suggest that with multiple missing  sources (which would be a far more typical case), many
individual ``ghost strings'' are produced at different orientations.  This would cause the artefacts and their
positive/negative sidelobes to overlap and, to an extent, average out, thus washing out the regular geometric pattern
and further contributing to their non-detection. It took a special set of circumstances --  i.e. a single relatively
strong and distant unmodelled source, high sensitivity, and WSRT's regular and highly redundant geometry  -- to make the ghosts
stand out (Fig.~\ref{fig:2004ghosts}) as a regular and peculiar geometric feature. 

Given the increased sensitivity of current and upcoming instruments, and the emergence of statistical detection techniques,
as argued in the introduction, a good understanding of the ghost phenomenon is vital. We have developed a theoretical 
framework for explaining ghost formation in the two-source WSRT case, which has yielded predictions that closely 
match simulations, qualitatively match actual observations (which are of course always more complex than just two sources), 
and explain all the previously puzzling features of the ghost phenomena that were empirically observed earlier 
(positioning at rational fractions along a line connecting two sources, a ``ghost spread function'' that differs from the 
PSF, independence -- to first order -- on the flux of the modelled source). We have also established that the 
well-known phenomenon of flux suppression is actually just one manifestation of the same underlying mechanism.

A particularly intriguing avenue of further research is ghost formation in the presense of direction-dependent effects.
Recall that in the original QMC result (Fig.~\ref{fig:2010ghosts}), the ghosts were found to be associated with
uncorrected DDEs, and went away once direction-dependent solutions had been applied. In this case, the unmodelled
flux responsible for ghost formation must have been due to errors in the voltage beam, caused by the artifically large
pointing error. We may postulate that any unmodelled flux, whether due to errors in the sky model, or unaccounted-for 
DDEs, will lead to ghosts on some level.

More work is required to describe the ghost phenomenon, both empirically and theoretically, for 2D arrays, for 
multiple and extended unmodelled sources, and for various forms of direction-dependent calibration. In
particular, flux suppression, as the strongest manifestation of the ghost mechanism, 
needs to be studied in more detail (especially in light of upcoming deep, blind surveys). We must also investigate  
alternative calibration approaches. In particular, robust calibration \citep{Kazemi2013a} has been shown to result in
less flux suppression, and must necessarily exhibit different ghost behaviour. 

Ultimately, we need to develop a theoretical and/or numerical mechanism for answering the following fundamental
questions. Given an observational scenario, how deep/accurate does a sky model need to be in order to suppress
ghosts to a given level? And, what then are the statistical properties and signatures (power spectrum, etc.) of the
remaining ghost artefacts? Building on the theory and numerical tools developed in this work, prospects are good that we
can eventually provide rigorous answers to these questions.

%{\it IRAS\/}

\section*{Acknowledgments}

This work is based upon research supported by the South African Research Chairs Initiative of the Department of 
Science and Technology and National Research Foundation.

Jan Noordam lovingly maintained a picture of the 2004 ghosts on a wall of his office, thus keeping the puzzle alive
until the next clue dropped in 2010, and contributed many challenging theories subsequently. The link to ALS was made 
thanks to fruitful discussions with Tobia Carozzi and Griffin Foster. We thank Gianni Bernardi and Cyril Tasse 
for critical comments on earlier drafts.

We would like to thank the anonymous referee for a thorough review, and very helpful suggestions that improved the paper.

\bibliographystyle{mn2e}
\bibliography{g_paper}

\appendix

\section[]{Lemmas and Propositions}

This appendix contains formal mathematical derivations of the propositions used in Sect.~\ref{sec:t_der}. 

\begin{proposition}
\label{prop:1}
If the function-valued matrix $\Rcal(\bmath{b})$:$\mathbfss{R}^2\rightarrow\mathbfss{C}^{n\times n}$ is defined as stated in Definition~\ref{def:R} its rank does not exceed two and its eigenvalues 
are either equal to zero or 
\begin{equation}
\label{eq:lam}
\frac{n(A_1+A_2)}{2} \pm h,
\end{equation}
where $h=\frac{1}{2}\sqrt{[n^2-4 {n \choose 2}][A_1+A_2]^2+\kappa}$ and
$\kappa = 4\sum_{p<q}(A_1^2+A_2^2+2A_1A_2\cos(2\pi\phi_{pq}\bmath{b}\cdot\bmath{s}_0))$.
\end{proposition}
\begin{proof}
The fact that the rank of $\Rcal(\bmath{b})$ does not exceed two follows trivially from Lemma~\ref{lemma:1}--\ref{lemma:3}. Since the rank of
$\Rcal(\bmath{b})$ is at most two its characteristic equation is equal to \citep{Ikramov2009,Blinn1996}
\begin{equation}
\label{eq:char}
\Bigg(\lambda^2 -\rmn{tr}(\Rcal(\bmath{b}))\lambda + \sum_{p<q} \begin{array}{|ll|}r_{pp} & r_{pq}\\ r_{qp} & r_{qq}\end{array}\Bigg)\lambda^{n-2}=0.
\end{equation}
Solving for $\lambda$ in Eq.~\ref{eq:char} produces the result.
\end{proof}

Proposition~\ref{prop:1} states that the rank of $\Rcal(\bmath(b))$ is two and gives an analytic expression of its largest eigenvalue, $\lambda(\bmath{b})$ (Eq.~\ref{eq:lam}).
The expression of $\lambda(\bmath{b})$ is also used in Proposition~\ref{prop:2} and Lemma~\ref{lemma:7}. Proposition~\ref{prop:1} have three direct dependencies namely, Lemma~\ref{lemma:1},
Lemma~\ref{lemma:2} and Lemma~\ref{lemma:3}. Lemma~\ref{lemma:1} gives the properties that a matrix must have so that its rank does not exceed $k\in\mathbfss{N}$.
Lemma~\ref{lemma:2} and Lemma~\ref{lemma:3} show that $\Rcal(\bmath{b})$ has the required properties so that its rank does not exceed two. 
The validity of Lemma~\ref{lemma:2} and Lemma~\ref{lemma:3} follows from Lemma~\ref{lemma:geo_matrix}, which gives the mathematical properties of $\bmath{\Phi}$.

\begin{proposition}
\label{prop:2} 
The entries $g_{pq}(\bmath{b})$ of $\Gcal(\bmath{b})$ are
differentiable Hermitian functions. Moreover, $g_{pq}(u,v)=g_{pq}\big(u+\frac{j}{l_0},v+\frac{k}{l_0}\big)$ and
$g_{pq}\big(u,-\frac{l_0}{m_0}u + c\big)=g_{pq}(0,c)$ $\forall j,k\in\mathbfss{Z}$ and $\forall u,v,c\in\mathbfss{R}$.
\end{proposition}

\begin{proof}
By Lemma~\ref{lemma:4} and Eq.~\ref{eq:lam}
\begin{eqnarray}
\label{eq:G_per}
\Gcal\bigg(u+\frac{j}{l_0},v+\frac{k}{l_0}\bigg)&=&\lambda(u,v)\mathbf{x}(u,v)\mathbf{x}^{H}(u,v)\\
&=&\Gcal(u,v),\nonumber
\end{eqnarray}
for all $j,k\in\mathbfss{Z}$. Eq.~\ref{eq:G_per} implies that $g_{pq}(u,v)=g_{pq}\big(u+\frac{j}{l_0},v+\frac{k}{l_0}\big)$, $\forall j,k\in\mathbfss{Z}$.
Similarly, $g_{pq}\big(u,-\frac{l_0}{m_0}u + c\big)=g_{pq}(0,c),$ $\forall u,c\in\mathbfss{R}$ (by Lemma~\ref{lemma:6}). 
%By Lemma 2.2 and Eq.~\ref{eq:lam}
%\begin{align}
%\label{eq:G_conj}
%\Gcal(-\bmath{b})&=\lambda(-\bmath{b})\mathbf{x}(-\bmath{b})\mathbf{x}^{H}(-\bmath{b})\\
%&=\lambda(u,v)\conj{\mathbf{x}}(u,v)\conj{\mathbf{x}}^{H}(u,v)\nonumber\\ 
%&=\conj{\Gcal}(u,v)\nonumber
%\end{align}
%Equation \ref{eq:G_conj} implies that $g_{ij}(-\bmath{b}) = \conj{g}_{ij}(\bmath{b})$. 
The fact that $g_{pq}(\bmath{b})$ is a differentiable function is established by Lemma~\ref{lemma:7}.

The function $g_{pq}(u,v)$ is the best possible (in a least squares sense) fit of $r_{pq}(u,v)$ (see Eq.~\ref{eq:cal2}). From this
observation and the fact that $\Rcal$ and $\Gcal$ are Hermitian matrices (as well as the fact that $\Rcal(-\bmath{b})=\conj{\Rcal(\bmath{b})} \Rightarrow \Gcal(-\bmath{b})=\conj{\Gcal(\bmath{b})}$) the following statements logically follow
\begin{enumerate}
 \item The best possible fit of $r_{pq}(-u,-v)=\conj{r}_{pq}(u,v)$ is $g_{pq}(-u,-v)$.
 \item The best possible fit of $r_{qp}(u,v)=r_{pq}(-u,-v)=\conj{r}_{pq}(u,v)$ is $g_{qp}(u,v)=\conj{g}_{pq}(u,v)$.
\end{enumerate}
The above statements imply that $g_{pq}(-u,-v)=\conj{g}_{pq}(u,v)$.
\end{proof}

Proposition~\ref{prop:2} shows that the elements $g_{pq}(\bmath{b})$ of $\Gcal(\bmath{b})$ are periodic, effectively one-dimensional, differentiable, Hermitian functions. The
properties of $g_{pq}(\bmath{b})$ follow from Lemma~\ref{lemma:4} (periodicity), Lemma~\ref{lemma:6} (one-dimensionality) and Lemma~\ref{lemma:7} (differentiability).
Lemma~\ref{lemma:7} is a consequence of Rellich's theorem.

\begin{proposition}
\label{prop:3} 
Each element $g_{pq}(\bmath{b})$ of $\Gcal(\bmath{b})$ can be written as the following sum
\begin{equation}
\label{eq:f_series}
 g_{pq}(\bmath{b}) = \sum_{j=-\infty}^{\infty}c_j e^{2\pi i j\bmath{b}\cdot\bmath{s}_0},
\end{equation}
that is
\begin{equation}
 g_{pq}(u,v) = \sum_{j=-\infty}^{\infty}c_j e^{2\pi i j(ul_0 + vm_0)},
\end{equation}
where 
\begin{equation}
  \label{eq:eq_2_p}
 c_j = \mu \int_{-\frac{1}{2|m_0|}}^{\frac{1}{2|m_0|}}\int_{-\frac{1}{2|l_0|}}^{\frac{1}{2|l_0|}}g_{pq}(u,v)e^{-2\pi i j(ul_0 + vm_0)}~dudv,
\end{equation}
with $\mu = |l_0||m_0|$ and $c_j\in\mathbfss{R}$. 
\end{proposition}

\begin{proof}
Since $g_{pq}$ is a differentiable periodic function in $\mathbfss{R}^2$ (Proposition~\ref{prop:2}), consider the standard Fourier series expansion
\begin{equation}
\label{eq:series_2}
g_{pq}(u,v) = \sum_{j=-\infty}^{\infty} \sum_{k=-\infty}^{\infty} c_{jk}e^{2\pi i(j l_0 u + k m_0 v)},
\end{equation}
with
\begin{equation}
c_{jk} = \mu\int_{-\frac{1}{2|m_0|}}^{\frac{1}{2|m_0|}}\int_{-\frac{1}{2|l_0|}}^{\frac{1}{2|l_0|}}g_{pq}(u,v)e^{-2\pi i (jul_0 + kvm_0)}~dudv
\end{equation}
and $\mu = |l_0||m_0|$.

Since $g_{pq}$ is Hermitian (Proposition~\ref{prop:2}), the coefficients $(c_{jk})$ are real numbers. Fix $c\in\mathbfss{R}$. Note that
$g_{pq}\big(u,-\frac{l_0}{m_0}u + c\big)=g_{pq}(0,c),$ $\forall u\in\mathbfss{R}$ (Proposition~\ref{prop:2}). We'll denote this
constant $g_{pq}(0,c)$ by $\alpha\in\mathbfss{C}$.

Evaluating the ``diagonal'' of the series in Eq.~\ref{eq:series_2}, at $(u,v)=\big(u,-\frac{l_0}{m_0}u + c\big)$, results in another
constant (i.e. independent of $u$), say $\beta\in\mathbfss{C}$.

Thus
\begin{equation}
h(u) \equiv \sum_{j=-\infty}^{\infty} \sum_{k=-\infty \atop k \neq j}^{\infty} c_{jk} e^{2\pi i (j-k)l_0 u + km_0c} = \alpha - \beta.  
\end{equation}
So, setting $n=j-k$, we get
\begin{eqnarray}
 \alpha - \beta &=& \sum_{j=-\infty}^{\infty} \sum_{k=-\infty \atop k \neq j}^{\infty}  c_{jk}e^{2\pi i k m_0 c}\cdot e^{2\pi i (j-k)l_0 u},\\
 &=& \sum_{j=-\infty}^{\infty} \sum_{n=-\infty \atop n \neq 0}^{\infty} c_{j,j-n}e^{2\pi i (j-n)m_0c}\cdot e^{2\pi i n l_0 u},\\
 &=& \sum_{n=-\infty \atop n \neq 0}^{\infty} d_n e^{2\pi i n l_0 u}, 
\end{eqnarray}
$\rmn{where}~d_n = \sum_{j = -\infty}^{\infty} c_{j,j-n}e^{2\pi i (j-n)m_0c}$. Thus $h(u)$, which is a one dimensional Fourier series wihtout a constant term, is a constant $\alpha - \beta$. This is only possible if 
$\alpha - \beta = 0$ and $d_n = 0$ whenever $n\neq 0$. This is again only possible if each $c_{j,j-n}=0$ whenever $n\neq 0$. Thus $c_{jk}=0$ whenever $j\neq k$.

Therefore, in the two dimensional Fourier series expansion of $g_{pq}$, only the terms with $j=k$ contribute.
\end{proof}

Proposition~\ref{prop:3} states that $g_{pq}(\bmath{b})$ can be expressed as an effectively one-dimensional Fourier-series and follows from Proposition~\ref{prop:2}.

Note that Proposition~\ref{prop:3} can also be stated using $e^{-2\pi ij\bmath{b}\cdot\bmath{s}_0}$ instead of $e^{2\pi ij\bmath{b}\cdot\bmath{s}_0}$ in which case
Eq.~\ref{eq:eq_2_p} becomes
\begin{equation}
  \label{eq:eq_2}
 c_j = \mu \int_{-\frac{1}{2|m_0|}}^{\frac{1}{2|m_0|}}\int_{-\frac{1}{2|l_0|}}^{\frac{1}{2|l_0|}}g_{pq}(u,v)e^{2\pi i j(ul_0 + vm_0)}~dudv,
\end{equation}
with $\mu = |l_0||m_0|$ and $c_j\in\mathbfss{R}$. It is also important to note that Proposition~\ref{prop:3} assumes that $l_0 \neq 0$ and $m_0 \neq 0$. When either $l_0$ or $m_0$ is 
zero the derivation simplifies and becomes one dimensional. To avoid cluttering the derivation of the one dimensional case is not repeated here.

\begin{proposition}
\label{prop:4} 
Let $h(u,v)=\frac 1 {g_{pq}(u,v)},$ then $h(u,v)$ will be a differentiable Hermitian function if $g_{pq}(u,v)\neq0,\forall u,v\in\mathbfss{R}$. 
Moreover, $h\big(u+\frac{j}{l_0},v+\frac{k}{l_0}\big)=h(u,v)$ and $h\big(u,-\frac{l_0}{m_0}u + c\big)=h(0,c)$,
 $\forall j,k\in\mathbfss{Z}$ and $u,v,c\in\mathbfss{R}$.
 \end{proposition}
\begin{proof} 
To see that $h$ has the same period as
$g_{pq}$, notice that for any
$j,k\in\mathbfss{Z}$ we have $h(u+j\lp,v+k\kp)=(g_{pq}(u+j\frac 1 {l_o},v+k\kp))^{-1}=g_{pq}(u,v)^{-1}=h(u,v).$ Similarly $h(u,u-\frac{l_0}{m_0}u + c)=h(0,c)~\forall u,c\in\mathbfss{R}$.
To see that $h$ is Hermitian, recall that complex conjugation
satisfies $\overline{\frac 1 z}=\frac 1 {\overline z}.$ Thus one computes $h(-u,-v)=\frac 1
{g_{pq}(-u,-v)}=\frac 1 {\overline{(g_{pq}(u,v))}}=\overline{(\frac 1 {g_{pq}(u,v)})}=\overline{h(u,v)}$. Finally, $h(u,v)$ is also differentiable since 
\begin{eqnarray}
 \frac{\partial h(u,v)}{\partial u} &=& -\frac{\frac{\partial g_{{pq}}(u,v)}{\partial u}}{g_{{pq}}^2(u,v)},\\
 \frac{\partial h(u,v)}{\partial v} &=& -\frac{\frac{\partial g_{{pq}}(u,v)}{\partial v}}{g_{{pq}}^2(u,v)},\\
\end{eqnarray}
exist ($g_{pq}(u,v)\neq0$ by assumption).
\end{proof}

Proposition~\ref{prop:4} shows that the elements of $\Gtop(\bmath{b})$ are also periodic, effectively one-dimensional, differentiable, Hermitian functions.
Proposition~\ref{prop:3} and Proposition~\ref{prop:4} therefore implies that the elements of $\Gtop(\bmath{b})$ also have a one-dimensional Fourier-series representation.

\begin{lemma}
\label{lemma:1}
Let $\bmath{A}$ be symmetric or Hermitian. If all principal submatrices having $k+1$ rows or $k+2$ rows are singular, the rank of
$\bmath{A}$ does not exceed $k$ \citep{Perlis1952}. 
\end{lemma}

\begin{lemma}
\label{lemma:2}
All $3 \times 3$ function-valued principal submatrices of $\Rcal(\bmath{b})$ are singular.
\end{lemma}
\begin{proof}
Due to the construction of $\Rcal(\bmath{b})$ all $3 \times 3$ function-valued principal submatrices of $\Rcal(\bmath{b})$ have the following form (see Lemma~\ref{lemma:geo_matrix})
\arraycolsep=0.1pt % default: 5pt
%\medmuskip = 0.1mu % default: 4mu plus 2mu minus 4mu
\footnotesize
\begin{equation}
\label{eq:mat_3}
\left[ \begin{array}{lll}
A_1+A_2 & A_1 + A_2e^{-2\pi i a \bmath{b}\cdot\bmath{s}_0} & A_1 + A_2e^{-2\pi i A \bmath{b}\cdot\bmath{s}_0}\\
A_1 + A_2e^{2\pi i a \bmath{b}\cdot\bmath{s}_0} & A_1+A_2 & A_1 + A_2e^{-2\pi i b \bmath{b}\cdot\bmath{s}_0}\\
A_1 + A_2e^{2\pi i A \bmath{b}\cdot\bmath{s}_0} & A_1 + A_2e^{2\pi i b \bmath{b}\cdot\bmath{s}_0} & A_1+A_2\end{array} \right],
\end{equation}
\normalsize
where $A = a+b$ and $a,b\in\mathbfss{N}$. The determinant of the matrix in Equation \ref{eq:mat_3} is equal to zero \citep{Kopp2008}.
\end{proof}

\begin{lemma}
\label{lemma:3}
All $4 \times 4$ function-valued principal submatrices of $\Rcal(\bmath{b})$ are singular.
\end{lemma}
\begin{proof}
Due to the construction of $\Rcal(\bmath{b})$ all $4 \times 4$ function-valued principal submatrices of $\Rcal(\bmath{b})$ have the following form (see Lemma~\ref{lemma:geo_matrix})\\
\arraycolsep=0.075pt % default: 5pt
%\medmuskip = 0.1mu % default: 4mu plus 2mu minus 4mu
\footnotesize
\begin{equation}
\label{eq:mat_4}
\left[ \begin{array}{llll}
A_1+A_2 & A_1 + A_2e^{-ka} & A_1 + A_2e^{-k A }&A_1 + A_2e^{-kC}\\
A_1 + A_2e^{ka} & A_1+A_2 & 1 + Ae^{-k b}& A_1 + A_2e^{-k B}\\
A_1 + A_2e^{kA} & A_1 + A_2e^{kb} & A_1+A_2 & A_1 + A_2e^{-kc}\\
A_1 + A_2e^{kC} & A_1 + A_2e^{kB} &  A_1 + A_2e^{k c} & A_1+A_2\end{array} \right],
\end{equation}
\normalsize
where $k=2\pi i \bmath{b}\cdot\bmath{s}_0,A = a + b,B=b+c,C=a+b+c$ and $a,b,c\in\mathbfss{N}$. The determinant of the matrix in Equation \ref{eq:mat_4} is equal to zero.
\end{proof}

\begin{definition}
Let $|\bmath{A}|_{d+1}$, where $\bmath{A} \in \mathbfss{Z}^{k\times k}$, be defined as $\sum_{p=1}^{n-1}a_{pp+1}$.
\end{definition}

\begin{definition}
Let $\mathcal{A}$ denote the set of all $m\times m$ principal sub-matrices of $\bmath{\Phi}$, with $m=n-1$. Let $\mathcal{B}$ denote the set of all $k\times k$ principal sub-matrices of $\bmath{\Phi}$, with $2 \leq k \leq n$.
\end{definition}

\begin{lemma}
\label{lemma:geo_matrix}
The array geometry matrix $\bmath{\Phi}$ has the following properties:
\begin{enumerate}
 \item $\phi_{pp}=0$ (diagonal),
 \item $\phi_{pq} \neq 0$ (non-diagonal),
 \item $\phi_{pq} > 0;~\forall q>p$,
 \item $\phi_{pq} = -\phi_{qp}$,
 \item gcd($\{\phi_{pq}\}_{q>p}$) $=1$.
 \item $|\bmath{\Phi}|_{d+1}=\phi_{1n}$.
 \item $|\bmath{B}|_{d+1} = b_{1k}$, $\forall \bmath{B}\in\mathcal{B}$.
\end{enumerate}

\end{lemma}

\begin{proof}
Property (i) is true since $\phi_{pp} = \phi_p-\phi_p=0$. Properties (ii)--(iv) follow trivially from the assumption that the antenna positions $\{\bmath{u}_p\}$ satisfy $||\bmath{u}_{q}||_2>||\bmath{u}_{p}||_2~\forall q > p$. 
Property (v) is true since gcd($\{\phi_{pq}\}_{q>p}$) = gcd(gcd($\{\phi_{1q}\}$),$\{\phi_{dq}\}_{q>d,d>1}$) = gcd(gcd($\{\phi_{q}\}$),$\{\phi_{dq}\}_{q>d,d>1}$) =  gcd(1,$\{\phi_{dq}\}_{q>d,d>1}$) = 1. Property (vi)
is true since $|\bmath{\Phi}|_{d+1} = \sum_{p=1}^{n-1}\phi_{pp+1} = \sum_{p=1}^{n-1}\phi_{p+1}-\phi_p=\phi_{1n}$. 

Property (vii) can be proven using the following argument. Assume that $\bmath{A}^j\in\mathcal{A}$ is obtained from $\bmath{\Phi}$ by deleting the $j$-th row and column from $\bmath{\Phi}$ (where $j$ was chosen arbitrarily).
When calculating $|\bmath{A}^j|_{d+1}$ three separate cases arise,
\begin{itemize}
 \item $1< j<n$: $|\bmath{A}^j|_{d+1} = \sum_{i=1}^{m-1} a_{ii+1}^j$ = $\sum_{p=1 \atop p \neq j,j+1}^{n-1}\phi_{pp+1}+\phi_{j,j+2}$ = $\sum_{p=1 \atop p \neq j,j+1}^{n-1}\phi_{pp+1}+\phi_{jj+1}+\phi_{j+1j+2}$ = $\phi_{1n}$=$a_{1m}$,  
 \item $j=n$: $|\bmath{A}^j|_{d+1} = \phi_{1n-1}=a_{1m}$,
 \item $j=1$: $|\bmath{A}^j|_{d+1} = \phi_{2n}=a_{1m}$.
\end{itemize}
The above shows that $|\bmath{A}|_{d+1} = a_{1m}~\forall \bmath{A} \in \mathcal{A}$ (since $j$ was chosen arbitrarily). Expanding the above derivation by using $1<t\leq n-2$ arbitrary deletions yields the required result. 
\end{proof}

%\begin{enumerate}
% \item Lemma~\ref{lemma:1}, this lemma gives the required conditions that is required by a matrix to have a rank that does not exceed $k\in\mathbfss{N}$.
% \item Lemma~\ref{lemma:2} and Lemma~\ref{lemma:3}, these lemmas prove that all $4\times 4$ and $3 \times 3$ principal submatrices of $\Rcal(\bmath{b})$ is singular. They
% provide the required conditions needed by Lemma~\ref{lemma:1} to calculate the rank of $\Rcal(\bmath{b})$.
% \begin{enumerate}
% \item Lemma~\ref{lemma:geo_matrix}, this lemma derives the mathematical properties of $\bmath{\Phi}$ (see Definition~\ref{def:phi}), which is required by Lemma
% \end{enumerate}

%\begin{corollary}
%The function-valued matrix $\Rcal(\bmath{b})$:$\mathbfss{R}^2\rightarrow\mathbfss{C}^{n\times n}$ defined in Definition~\ref{def:R} is a positive semi definite function-valued Hermitian matrix.
%\end{corollary}
%\begin{proof}
%The result follows trivially from Eq.~\ref{eq:lam}.  
%\end{proof}

\begin{lemma}
\label{lemma:4}
Let $\lambda(u,v)$ denote the largest eigenvalue of $\Rcal(u,v)$ and $\mathbf{x}(u,v)$ its associated normalized eigenvector, then
$\mathbf{x}(u,v)=\mathbf{x}\big(u+\frac{j}{l_0},v+\frac{k}{m_0}\big),\forall j,k\in\mathbfss{Z}$. 
%If $\mathbf{x}(u,v)$ is the unit eigenvector of $\Rcal(u,v)$ with its corresponding eigenvalue being equal to $\lambda(u,v)$, where 
%$\lambda(u,v)$ is largest eigenvalue of $\Rcal(u,v)$, then
%$\mathbf{x}(u,v)=\mathbf{x}\bigg(u+\frac{j}{l_0},v+\frac{k}{m_0}\bigg),\forall j,k\in\mathbfss{Z}$.
\end{lemma}
\begin{proof}
Notice that $\forall j,k\in\mathbfss{Z}$
\begin{eqnarray}
\Rcal\bigg(u+\frac{j}{l_0},v+\frac{k}{m_0}\bigg) &=& \Rcal(u,v),\\
\lambda\bigg(u+\frac{j}{l_0},v+\frac{k}{l_0}\bigg) &=& \lambda(u,v);
\end{eqnarray}
implying that $\mathbf{x}(u,v)=\mathbf{x}\big(u+\frac{j}{m_0},v+\frac{k}{l_0}\big),\forall j,k\in\mathbfss{Z}$. 
\end{proof}

%\begin{lemma}
%\label{lemma:5}
%Let $\lambda(u,v)$ denote the largest eigenvalue of $\Rcal(u,v)$ and $\mathbf{x}(u,v)$ its associated normalized eigenvector, then
%$\conj{\mathbf{x}}(\bmath{b})$ and $\mathbf{x}(\bmath{-b})$ are eigenvectors of $\conj{\Rcal}(\bmath{b})$ and their corresponding eigenvalue is equal to $\lambda(\bmath{b})$.
%\end{lemma}
%\begin{proof}
%By definition
%\begin{equation}
%\Rcal(\bmath{b})\mathbf{x}(\bmath{b})=\lambda(\bmath{b})\mathbf{x}(\bmath{b}).
%\end{equation}
%Now
%\begin{equation}
%\label{eq:l_2_2_a}
%\conj{\Rcal(\bmath{b})\mathbf{x}(\bmath{b})}=\conj{\lambda(\bmath{b})\mathbf{x}(\bmath{b})},
%\end{equation}
%therefore 
%\begin{equation}
%\conj{\Rcal}(\bmath{b})\conj{\mathbf{x}}(\bmath{b})=\lambda(\bmath{b})\conj{\mathbf{x}}(\bmath{b})\nonumber
%\end{equation}
%and
%\begin{equation}
%\label{eq:l_2_2_b}
%\Rcal(-\bmath{b})\mathbf{x}(-\bmath{b})=\lambda(-\bmath{b})\mathbf{x}(-\bmath{b}),
%\end{equation}
%therefore
%\begin{equation}
%\conj{\Rcal}(\bmath{b})\mathbf{x}(-\bmath{b})=\lambda(\bmath{b})\mathbf{x}(-\bmath{b}).\nonumber
%\end{equation} 
%\end{proof}

\begin{lemma}
\label{lemma:6}
Let $\lambda(u,v)$ denote the largest eigenvalue of $\Rcal(u,v)$ and $\mathbf{x}(u,v)$ its associated normalized eigenvector, then
$\mathbf{x}\big(u,-\frac{l_0}{m_0}u + c\big)=\mathbf{x}(0,c),~\forall u,c\in\mathbfss{R}$.
\end{lemma}
\begin{proof}
Notice that $\forall u,c\in\mathbfss{R}$
\begin{eqnarray}
\label{eq:lem_6_1}
\Rcal\bigg(u,-\frac{l_0}{m_0}u + c\bigg)&=&\Rcal(0,c),\\
\label{eq:lem_6_2}
\lambda\bigg(u,-\frac{l_0}{m_0}u + c\bigg)&=&\lambda(0,c),
\end{eqnarray}  
implying that $\mathbf{x}\big(u,-\frac{l_0}{m_0}u + c\big)=\mathbf{x}(0,c),~\forall u,c\in\mathbfss{R}$.
\end{proof}

%\begin{conjecture}
%\label{lemma:7}
%Given a square function-valued matrix $\bmath{A}(\bmath{t})$ whose entries depend smoothly on a vector of real parameters $\bmath{t}$ , 
%if $\lambda = \lambda_0$ is a simple eigenvalue at $\bmath{t} = \bmath{t}_0$ with a corresponding unit eigenvector $\mathbf{x}_0$, then for all $\bmath{t}$ 
%near $\bmath{t}_0$ there is a corresponding eigenvalue and unique (normalized) eigenvector that depend smoothly on $\bmath{t}$. 
%Also, if the elements of $\bmath{A}(\bmath{t})$ are functions of $\bmath{t}$, then so are the eigenvalue $\lambda_0$ and eigenvector $\mathbf{x}_0$ \citep{Garcia2011,Lax1996}.\\
%\end{conjecture}

\begin{lemma}
\label{lemma:7}
Let $\lambda(u,v)$ denote the largest eigenvalue of $\Rcal(u,v)$ and $\mathbf{x}(u,v)$ its associated normalized eigenvector. The real function
$\lambda(u,v)$ and the function-valued vector $\mathbf{x}(u,v)$ are differentiable. 
\end{lemma}
\begin{proof}
The parameter dimension of $\Rcal(\bmath{b})$ is effectively one, i.e. $\Rcal_t(t):=\Rcal(u,v), \lambda_t(t) :=  \lambda(u,v)$ and $\mathbf{x}_t(t):=\mathbf{x}(u,v)$ with $t(u,v) := \bmath{b}\cdot\bmath{s}_0$. The entries of $\Rcal_t(t)$ are analytic functions depending on $t\in\mathbfss{R}$.
Therefore, Rellich's theorem (Lemma~\ref{lemma:rellich}) implies that $\lambda(t)$ and $\mathbf{x}(t)$ are analytic (the largest eigenvalue is simple-Eq.~\ref{eq:lam}) and therefore also differentiable. 
We can therefore calculate
\begin{eqnarray}
\frac{\partial\lambda(u,v)}{\partial u} &=& \frac{d\lambda_t(t(u,v))}{dt}\frac{\partial t(u,v)}{\partial u},\\
\frac{\partial\lambda(u,v)}{\partial v} &=& \frac{d\lambda_t(t(u,v))}{dt}\frac{\partial t(u,v)}{\partial v}.
\end{eqnarray}
The above equations imply that the real function $\lambda(u,v)$ is differentiable. A similar argument can be used to prove that $\mathbf{x}(u,v)$ is also differentiable. 
\end{proof}

%To the authors' knowledge Rellich's theorem implies that the largest eigenvalue of $\Rcal_t(t)$ and its associated eigenvector are differentiable (since the largest eigenvalue is simple in our case--Eq.~\ref{eq:lam}). If however, there 
%exists a subtle reason which invalidates this inference then Lemma~\ref{lemma:7} is no longer true. In that case an alternative
%proof would be required to prove that $\lambda(u,v)$ and $\mathbf{x}(u,v)$ are differentiable for all $u,v\in\mathbfss{R}$. It should be emphasized that, empirical evidence supports the validity of Eq.~\ref{eq:f_series} and therefore also the validity of Lemma~\ref{lemma:7}.

%The above result is stated as a conjecture, since we did not find a rigorous proof of it in the literature. However, when $\bmath{A}(\bmath{t})=\Rcal(\bmath{t})$ its validity is supported by 
%experimental evidence and the fact that Eq.~\ref{eq:lam} is clearly continious. Moreover, extending the results in \citet{Lax1996} should yield the result.

\begin{lemma}[Rellich's Theorem]
\label{lemma:rellich} 
Let $\bmath{A}(t) : \mathbfss{R}\rightarrow\mathbfss{C}^{n \times n}$ be a Hermitian function-valued matrix that depends
on $t$ analytically.
\begin{enumerate}
 \item The $n$ roots of the characteristic polynomial of $\bmath{A}(t)$ can be arranged so that each root $\lambda_j(t)$ for $j = 1,\cdots,n$ is an analytic function of $t$.
 \item There exists an eigenvector $\mathbf{x}_j(t)$ associated with $\lambda_j(t)$ for $j = 1,\cdots,n$ satisfying
       \begin{enumerate}
              \item $||\mathbf{x}_j(t)||_2=1~\forall t\in\mathbfss{R}$,
              \item $\mathbf{x}_j(t)$ is an analytic function-valued vector of $t$ \citep{Reed1978,Lax1996,Kilicc2011}.
       \end{enumerate}
\end{enumerate}
\end{lemma}

%Theorem XII.1 and XII.4 in \citet{Reed1978} show that the above result is valid for the largest eigenvalue of a rank deficient function-valued Hermitian matrix $\bmath{A}(t):\mathbfss{R}\rightarrow\mathbfss{C}^{n \times n}$ and its associated normalized eigenvector if the largest eigenvalue is
%simple. 

\bsp

\label{lastpage}

\end{document}